\documentclass[sigconf,screen]{acmart}
\pdfoutput=1

\renewcommand\footnotetextcopyrightpermission[1]{}

\newif\iflongversion
\longversiontrue

\usepackage[T1]{fontenc}
\usepackage[utf8]{inputenc}
\usepackage{booktabs}
\usepackage{multirow}
\usepackage{balance}
\usepackage{xcolor}

\usepackage{verbatim}
\usepackage{graphicx}
\usepackage{caption}
\usepackage{subcaption}
\usepackage{setspace}
\usepackage{float}
\usepackage{tikz}
\usepackage{amsmath}
\usepackage{breqn}
\usepackage{hhline}
\usetikzlibrary{arrows}

\usepackage{algorithm}
\usepackage[noend]{algpseudocode}
\usepackage{listings}
\definecolor{key-color}{rgb}{0.8, 0.47, 0.196}
\usepackage{enumitem}
\usepackage{hyphenat}

\usepackage{doi}
\usepackage{hyperref}

\DeclareMathAlphabet{\mathpzc}{OT1}{pzc}{m}{it}

\expandafter\def\expandafter\UrlBreaks\expandafter{\UrlBreaks\do\/\do\-\do\.} 
\definecolor[named]{Purple}{cmyk}{0.55,1,0,0.15}
\definecolor[named]{DarkBlue}{cmyk}{1,0.58,0,0.21}
\hypersetup{bookmarksnumbered,unicode,naturalnames,colorlinks,breaklinks,
linkcolor=Purple,
citecolor=Purple,
urlcolor=DarkBlue,
filecolor=DarkBlue}

\graphicspath{ {images/} }

\newenvironment{squishedlist}
{
  \begin{list}{$\bullet$}
   {
     \setlength{\itemsep}{0pt}
     \setlength{\parsep}{2pt}
     \setlength{\topsep}{1.0pt}
     \setlength{\partopsep}{0pt}
     \setlength{\leftmargin}{01.5em}
     \setlength{\labelwidth}{1em}
     \setlength{\labelsep}{0.5em}
   }
}
{
   \end{list}
}

\newcommand{\noun}[1]{\textsc{#1}}
\newcommand{\graphsurge}[1]{\noun{Graphsurge}}

\newtheorem{theorem}{Theorem}[section]

\newtheorem{corollary}{Corollary}[section]
\newtheorem{proposition}{Proposition}[section]

\newcounter{theorem3}
\newtheorem{definition}[theorem3]{Definition}

\newcommand{\A}{\mathcal{A}}
\newcommand{\R}{\mathcal{R}}
\newcommand{\D}{\mathcal{D}}
\newcommand{\E}{\mathcal{E}}
\newcommand{\I}{\mathcal{I}}

\newtheorem{innercustomgeneric}{\customgenericname}
\providecommand{\customgenericname}{}
\newcommand{\newcustomtheorem}[2]{%
  \newenvironment{#1}[1]
  {%
   \renewcommand\customgenericname{#2}%
   \renewcommand\theinnercustomgeneric{##1}%
   \innercustomgeneric
  }
  {\endinnercustomgeneric}
}

\newcustomtheorem{customquestion}{Question}

\makeatletter
\newcommand*{\relrelbarsep}{.386ex}
\newcommand*{\relrelbar}{%
  \mathrel{%
    \mathpalette\@relrelbar\relrelbarsep
  }%
}
\newcommand*{\@relrelbar}[2]{%
  \raise#2\hbox to 0pt{$\m@th#1\relbar$\hss}%
  \lower#2\hbox{$\m@th#1\relbar$}%
}
\providecommand*{\rightrightarrowsfill@}{%
  \arrowfill@\relrelbar\relrelbar\rightrightarrows
}
\providecommand*{\leftleftarrowsfill@}{%
  \arrowfill@\leftleftarrows\relrelbar\relrelbar
}
\providecommand*{\xrightrightarrows}[2][]{%
  \ext@arrow 0359\rightrightarrowsfill@{#1}{#2}%
}
\providecommand*{\xleftleftarrows}[2][]{%
  \ext@arrow 3095\leftleftarrowsfill@{#1}{#2}%
}
\makeatother

\setcopyright{none}

\begin{document}

\title{Accurate Summary-based Cardinality Estimation Through the Lens of Cardinality Estimation Graphs
}

\author{Jeremy Chen}
\affiliation{\institution{University of Waterloo}}
\email{jeremy.chen@uwaterloo.ca}

\author{Yuqing Huang}
\affiliation{\institution{University of Waterloo}}
\email{y558huan@uwaterloo.ca}

\author{Mushi Wang}
\affiliation{\institution{University of Waterloo}}
\email{m358wang@uwaterloo.ca}

\author{Semih Salihoglu}
\affiliation{\institution{University of Waterloo}}
\email{semih.salihoglu@uwaterloo.ca}

\author{Ken Salem}
\affiliation{\institution{University of Waterloo}}
\email{ken.salem@uwaterloo.ca}

\begin{abstract}

We study two classes of summary-based cardinality estimators that 
use statistics about input relations and small-size joins 
in the context of graph database management systems: (i)
optimistic estimators that make uniformity and conditional independence 
assumptions; and  (ii) the recent pessimistic estimators that 
use information theoretic linear programs. 
We begin by addressing the problem of how to make accurate estimates for 
optimistic estimators. We model these estimators
as picking bottom-to-top paths in a {\em cardinality estimation graph} (CEG), 
which contains sub-queries as nodes and weighted edges between sub-queries 
that represent average degrees. 
We outline a space of heuristics to make an optimistic estimate
in this framework and show that effective heuristics depend on the structure of the input
queries. We observe that on acyclic queries and queries with small-size cycles, 
using the maximum-weight path is an effective technique to address the well known underestimation
problem for optimistic estimators. 
We show that on a large suite of datasets and workloads, the accuracy of such estimates is 
up to three orders of magnitude more accurate in mean q-error than some prior heuristics that 
have been proposed in prior work. 
In contrast, we show that on queries with larger cycles these estimators tend to overestimate, which 
can partially be addressed by using minimum weight paths and more effectively by using an alternative CEG.
We then show that CEGs can also model the recent pessimistic estimators. This surprising result 
allows us to connect two disparate lines of 
work on optimistic and pessimistic estimators, adopt an optimization from pessimistic 
estimators to optimistic ones, and provide insights into the pessimistic estimators, 
such as showing that there are alternative combinatorial solutions to the linear programs that define them.

\end{abstract}

\maketitle
\pagestyle{plain}

\section{Introduction}
\label{sec:introduction}

The problem of estimating the output size of a natural multi-join
query (henceforth {\em join query} for short), is a fundamental
problem that is solved in the query optimizers of database management
systems when generating efficient query plans. This problem arises
both in systems that manage relational data as well those that manage
graph-structured data where systems need to estimate the
cardinalities of subgraphs in their input graphs. It is well known
that both problems are equivalent, since subgraph queries can
equivalently be written as join queries over binary relations that
store the edges of a graph.

We focus on the prevalent technique used by existing systems 
of using statistics about the base relations or outputs of small-size joins to estimate
cardinalities of joins. These techniques use these statistics in algebraic formulas
that make independence and uniformity assumptions to generate estimates for 
queries~\cite{aboulnaga:markov, maduko:md-tree, 
mhedhbi:optimizer, neumann:characteristic-sets}.
We will refer to these as \emph{summary-based optimistic estimators} (optimistic estimators for short),
to emphasize that these estimators can make both under and overestimations.
This contrasts with the recent {\em pessimistic estimators} that are based on
worst-case optimal join size
bounds~\cite{khamis:cllp, atserias:agm, cai:pessimistic, 
gottlob:glvv, joglekar:degree} and avoid underestimation, albeit 
using very loose, so inaccurate, estimates~\cite{park:gcare}. 
In this work, we study how to make accurate estimations using optimistic estimators 
using a new framework that we call {\em cardinality estimation graphs} (CEGs) to represent them.
We observe and address several shortcomings of these estimators under the CEG framework.
We show that this framework is useful in two additional ways: (i) CEGs can also
represent the pessimistic estimators, establishing that these two classes of 
estimators are in fact surprisingly connected;  and (ii) CEGs are
useful mathematical tools to prove several theoretical properties of 
pessimistic estimators.

\begin{figure}[t!]
	\centering
	\captionsetup{justification=centering}
    	\includegraphics[width=0.45\columnwidth]{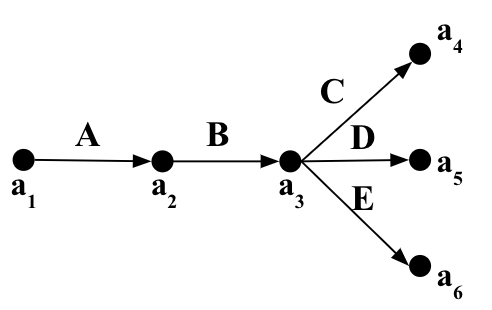}
		\vspace{-10pt}
  \caption{Example subgraph query $Q_{5f}$.}
	\label{fig:running-query}
	\vspace{-10pt}
\end{figure}

We begin by showing that the algebraic formulas of prior optimistic estimators can be 
modeled as picking a bottom-to-top path in a weighted CEG, which we call $CEG_O$, 
for {\bf O}ptimistic. In this CEG nodes are intermediate 
sub-queries and edges weights are average degree statistics that extend sub-queries to 
larger queries. For example, consider the subgraph query in 
Figure~\ref{fig:running-query} and the input dataset in Figure~\ref{fig:running-dataset},
whose $CEG_O$ is shown in Figure~\ref{fig:mt-jg-h3}.
We observe that in estimators that can be represented with CEGs, 
there is often more than one way to generate an estimate for a query,
corresponding to different bottom-to-up paths and that this decision
does not have a clear answer in the case of $CEG_O$.
For example, consider the subgraph query in 
Figure~\ref{fig:running-query}. Given that we have the accurate 
cardinalities of all subqueries of size $\le$ 2 available,
there are 252 formulas (or bottom-to-up paths in the CEG 
of optimistic estimators) 
to estimate the cardinality of the query.
Examples of these formulas are:
\begin{squishedlist}
\item  $|\xrightarrow{A}\xrightarrow{B}| \times \frac{|\xrightarrow{B}\xrightarrow{C}|}{|\xrightarrow{B}|} \times \frac{|\xleftarrow{C}\xrightarrow{D}|}{|\xrightarrow{C}|} \times \frac{|\xleftarrow{D}\xrightarrow{E}|}{|\xrightarrow{D}|}$
\item  $|\xrightarrow{A}\xrightarrow{B}| \times \frac{|\xrightarrow{B}\xrightarrow{D}|}{|\xrightarrow{B}|} \times \frac{|\xleftarrow{C}\xrightarrow{D}|}{|\xrightarrow{D}|} \times \frac{|\xrightarrow{B}\xrightarrow{E}|}{|\xrightarrow{B}|}$
\end{squishedlist}

In previous work~\cite{aboulnaga:markov, maduko:md-tree, 
mhedhbi:optimizer}, 
the choice of which of these estimates 
to use has either been unspecified or decided by a heuristic without acknowledging 
possible other choices or empirically justifying these choices.
As our first main contribution, we systematically describe a space of heuristics 
for making an estimate for optimistic estimators
and show empirically that the better performing heuristics depend on the structure of the query. 
We show that on acyclic queries and queries with small-size cycles, whose statistics
are available, using the \emph{maximum-weight path} through the CEG is an effective way to make 
accurate estimations. We observe that as in the relational setting, estimators that use independence assumptions
tend to underestimate the true cardinalities on these queries, and the use of maximum-weight 
path in the CEG can offset these underestimations.
In contrast we observe that on queries that contain larger cycles, the optimistic estimators estimate
modified versions of these queries that break these cycles into paths, which leads to overestimations. We 
show that on such queries, using the minimum weight paths leads to generally more accurate estimations
than other heuristics.  
However, we also 
observe that these estimates can still be highly inaccurate, and address this shortcoming by defining a new 
CEG for these queries, which we call $CEG_{OCR}$, for {\bf O}ptimistic {\bf C}ycle closing {\bf R}ate,
that use as edge weights statistics that estimate the probabilities that two edges that 
are connected through a path can close into a cycle. 
We further show that in general considering only the 
``long-hop'' paths, i.e., with more number of edges, that make more independence 
assumption but by conditioning on larger sub-queries
performs better than paths with fewer edges and comparably to considering every path.

\begin{figure}[t!]
	\centering
	\captionsetup{justification=centering}
    	\includegraphics[width=0.8\columnwidth]{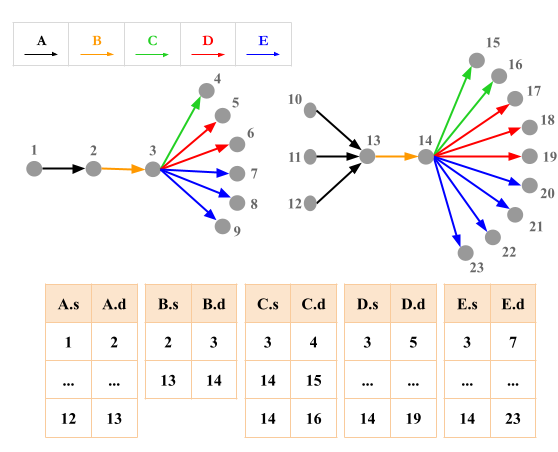}
	\vspace{-15pt}
  \caption{Example dataset in graph and relational formats.}
	\label{fig:running-dataset}
		\vspace{-10pt}
\end{figure}

As our second main contribution, we show that CEGs are expressive enough to model also the recent 
linear program-based pessimistic estimators. 
Specifically, we show that we can take $CEG_O$ and 
replace its edge weights (which are average degrees) 
with maximum degrees of base relations and small-size joins, and construct a new
CEG, which we call $CEG_M$, for {\bf M}OLP bound from reference~\cite{joglekar:degree}.
We show that each path in $CEG_M$ is guaranteed to be an overestimate for the cardinality of the query, 
therefore picking the minimum weight path on this CEG would be the most accurate estimate. 
We show that this path is indeed equivalent to the MOLP pessimistic estimator
from reference~\cite{joglekar:degree}. 
The ability to model pessimistic estimators as CEGs 
allows us to make several contributions. First, we connect 
two seemingly disparate classes of estimators: both subgraph summary-based 
optimistic estimators and the recent LP-based pessimistic ones can be modeled as different 
instances of estimators that pick paths through CEGs. 
Second, this result sheds light into the nature of the arguably opaque LPs that define pessimistic estimators.
Specifically we show that in addition to their numeric representations, the pessimistic estimators
have a combinatorial representation. 
Third, 
we show that a {\em bound sketch} optimization from the recent pessimistic estimator from references~\cite{cai:pessimistic} can
be directly applied to any estimator using a CEG, specifically to the optimistic estimators,
and empirically demonstrate its benefits in some settings.

CEGs further turn out to be very useful mathematical tools to prove certain properties 
of the pessimistic estimators, which may be of independent interest to readers 
who are interested in the theory of pessimistic estimators. 
Specifically using CEGs in our proofs, we show that MOLP can be simplified 
because some of its constraints are unnecessary,  provide several alternative combinatorial proofs to
some known properties of MOLP, such as
the theorem that MOLP is tighter than another bound called DBPLP~\cite{joglekar:degree},
and that MOLP is at least as tight as 
the pessimistic estimator proposed by Cai et al~\cite{cai:pessimistic} and
are identical on acyclic queries over binary relations.

\begin{figure}[t!]
	\centering
	\captionsetup{justification=centering}
    	\includegraphics[width=\columnwidth]{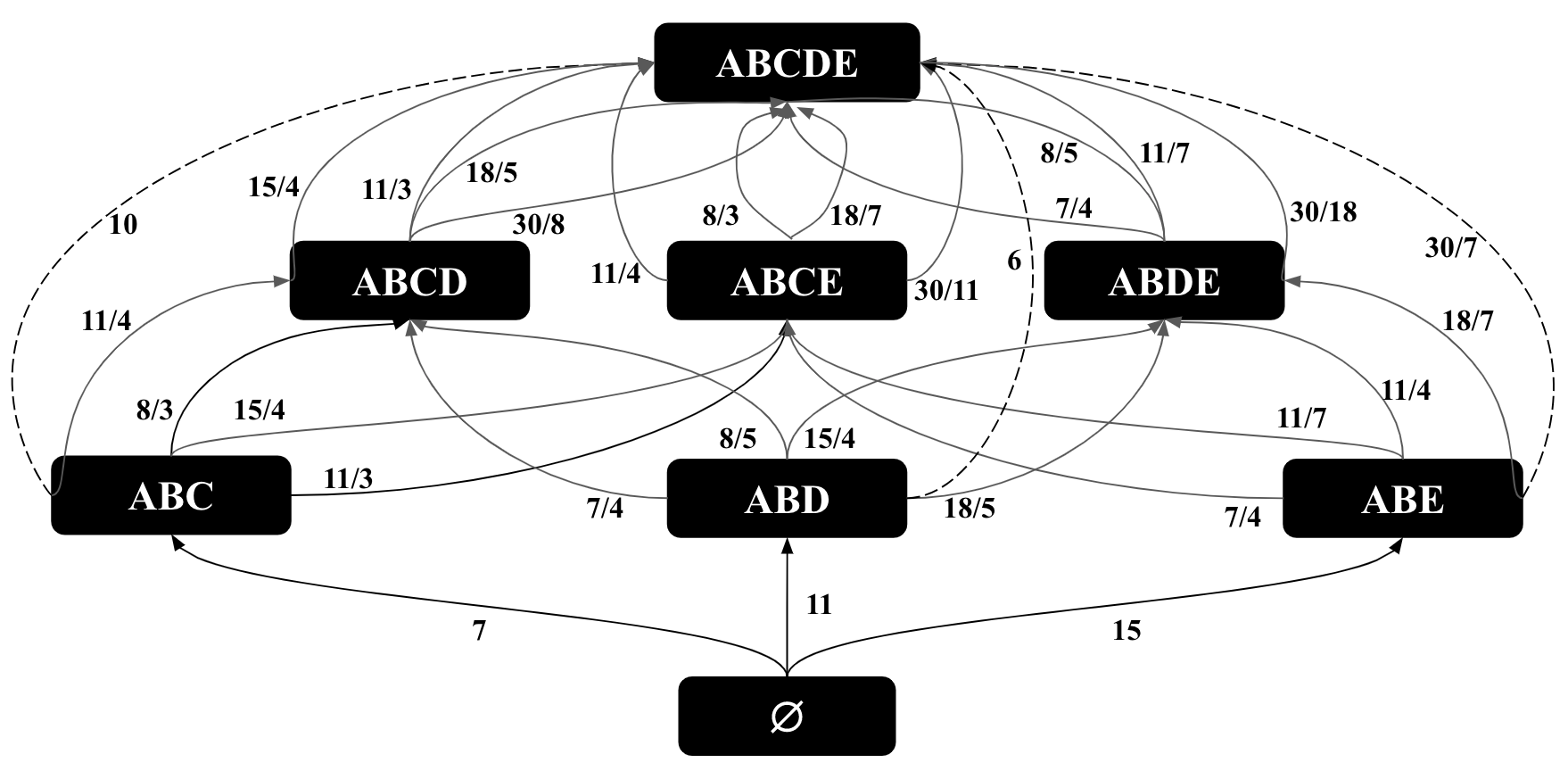}
	\vspace{-10pt}
        \caption{$CEG_{O}$ for query $Q_{5f}$ in Figure~\ref{fig:running-query} when the Markov table (\S\ref{sec:optimistic-estimators})
          contains joins up to size 3. Only a subset of the vertices and edges of the full CEG is shown.}
	\label{fig:mt-jg-h3}
	\vspace{-10pt}
\end{figure}

The remainder of this paper is structured as follows. Section~\ref{sec:running-ex} 
provides our query and database notation. 
Section~\ref{sec:cel} gives an overvi-ew of generic estimators that can be seen as
picking paths from a CEG. Section~\ref{sec:optimistic-estimators} reviews optimistic
estimators that can be modeled with $CEG_{O}$ and outlines the space of
possible heuristics for making estimates using $CEG_{O}$. We also discuss the shortcoming
of $CEG_O$ when estimating queries with large cycles and present $CEG_{OCR}$ to
address this shortcoming. 
Section~\ref{sec:pessimistic-estimators} reviews the pessimistic estimators,
the $CEG_M$ of the MOLP estimator and the bound sketch refinement 
to pessimistic estimators. 
Using $CEG_M$, we prove several properties of MOLP and connect 
some of these pessimistic estimators. 
Section~\ref{sec:evaluation} presents extensive experiments evaluating our space of
optimistic estimators both on $CEG_O$, and $CEG_{OCR}$ and the benefits of bound sketch optimization. 
We compare the optimistic estimators against other summary-based and 
sampling-based techniques and evaluate the effects of our estimators on plan quality on an
actual system. We also confirm as in reference~\cite{park:gcare} that the pessimistic estimators in our 
settings are also not competitive and lead to highly inaccurate estimates.
Finally, Sections~\ref{sec:rw} and~\ref{sec:fw} cover related work and conclude, respectively.

\section{Query and Database Notation}
\label{sec:running-ex}
We consider conjunctive queries of the form
\[
Q(\mathcal{A}) = R_1(\mathcal{A}_1),\ldots,R_m(\mathcal{A}_m)
\]
where $R_i(\mathcal{A}_i)$ is a relation with attributes $\mathcal{A}_i$ and
$\mathcal{A}=\cup_i\mathcal{A}_i$.
Most of the examples used in this paper involve edge-labeled subgraph queries,
in which case each $R_i$ is a binary relation containing a subset of the edges in
a graph as source/destination pairs.
Figure~\ref{fig:running-dataset} presents an example showing a graph
with edge labels $A$, $B$, $C$, $D$, and $E$, shown in black, orange,
green, red, and blue.   This graph can be represented using five binary
relations, one for each of the edge labels.   These relations are
also shown in Figure~\ref{fig:running-dataset}.

We will often represent queries over such relations
using a graph notation.
For example, consider the relations $A$ and $B$
from Figure~\ref{fig:running-dataset}.
We will represent the query
$Q(a_1,a_2,a_3) = A(a_1, a_2) \bowtie B(a_2, a_3)$ as 
$a_1\xrightarrow{A}a_2\xrightarrow{B}a_3$.
Similarly, the query
$Q(a_1,a_2,a_3) = A(a_1, a_2) \bowtie B(a_3, a_2)$
will be represented as $a_1\xrightarrow{A}a_2\xleftarrow{B}a_3$.

\section{CEG Overview}
\label{sec:cel}

Next, we offer some intution for {\em cardinality estimation graphs}
(CEGs).
In Sections~\ref{sec:optimistic-estimators}
and~\ref{sec:pessimistic-estimators} we will define specific CEGs
corresponding to different classes of optimistic and pessimistic estimators.
However, all of these share a common structure for representing
cardinality estimations.
Specifically, a CEG for a query $Q$ will consist of:
\begin{squishedlist}
\item Vertices labeled with subqueries of $Q$, where subqueries are defined
	by subsets of $Q$'s relations or attributes. 
\item Edges from smaller subqueries to  larger subqueries, labeled
  with {\em extension rates} which represent the
  cardinality of the larger subquery relative to that of the smaller subquery.
\end{squishedlist}

\begin{figure}[t!]
	\centering
	\captionsetup{justification=centering}
    	\includegraphics[width=\columnwidth]{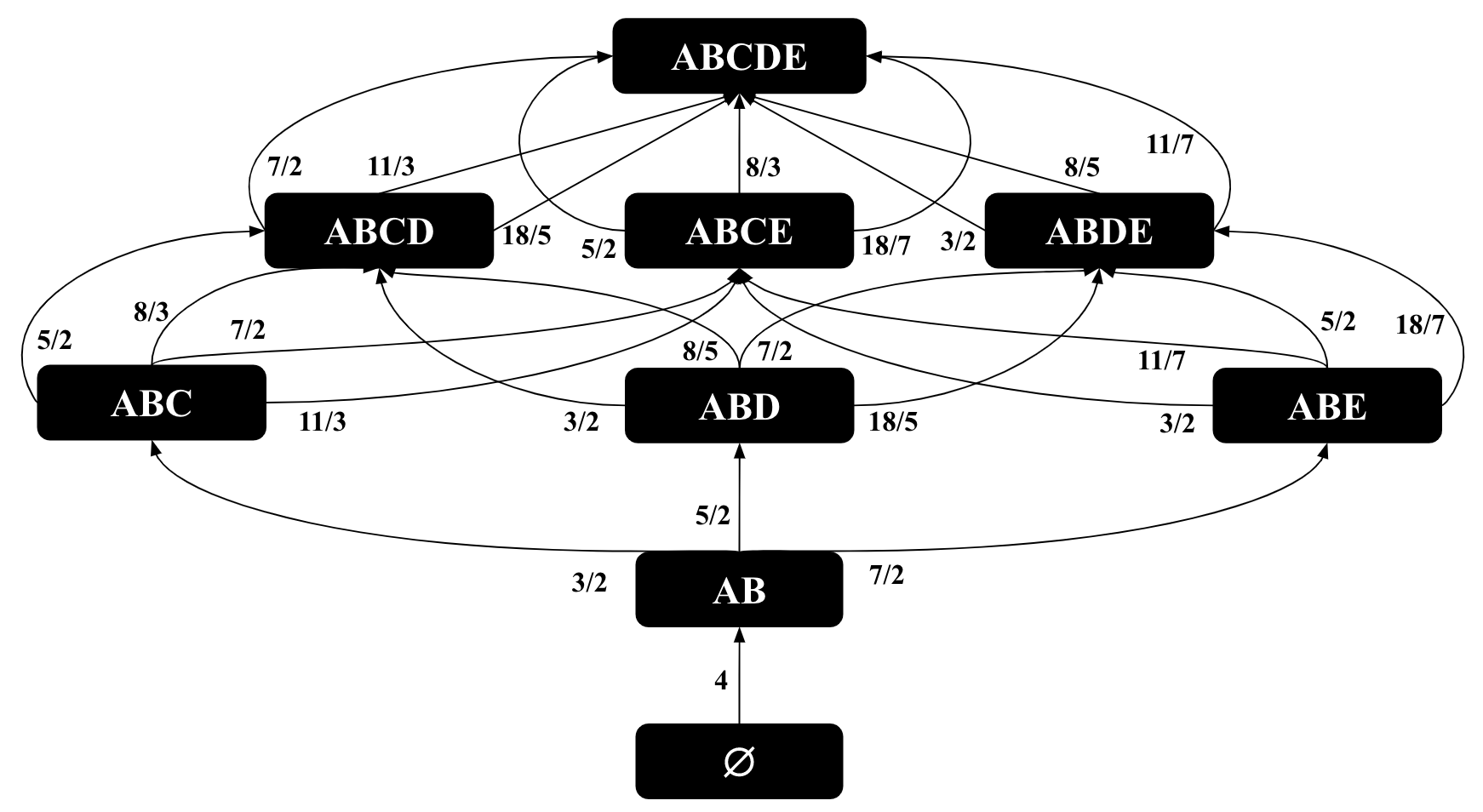}
	\vspace{-10pt}
  \caption{$CEG_{O}$ for query $Q_{5f}$ from Figure~\ref{fig:running-query}.
         The figure shows a subset of the vertices and edges of the full CEG.}
	\label{fig:opt-cel-h2}
	\vspace{-10pt}
\end{figure}
\noindent Each bottom-to-top path (from $\emptyset$ to $Q$) in a CEG represents a different way
of generating a cardinality estimate for $Q$. An estimator using a CEG picks one of these paths
as an estimate.  The estimate of a path is the product of the 
extension rates along the edges of the path. Equivalently one can put
the logarithms of the extension rates as edge weights and sum the
logarithms (and exponentiate the base of the logarithm) to compute the estimate.

Figure~\ref{fig:opt-cel-h2} illustrates
a $CEG$\footnote{Specifically, it is a $CEG_{O}$, defined in Section~\ref{sec:optimistic-estimators}.}
for  the query
$Q_{5f}$ shown in Figure~\ref{fig:running-query}
over the relations shown in Figure~\ref{fig:running-dataset}, assuming that
statistics are available for any size-2 subqueries of $Q_{5f}$. 
For example, the leftmost path starts with 
$a_1\xrightarrow{A}a_2\xrightarrow{B}a_3$, then
extends to $a_1\xrightarrow{A}a_2\xrightarrow{B}a_3\xrightarrow{C}a_4$, 
then to the subquery of 4-fork involving $A$, $B$, $C$, and $D$, and finally
extends the 4-fork subquery to $Q_{5f}$. The first extension rate
from $\emptyset$ to  $a_1\xrightarrow{A}a_2\xrightarrow{B}a_3$ is simply
the known cardinality of $a_1\xrightarrow{A}a_2\xrightarrow{B}a_3$, which is 
4, and the second extension rate makes the uniformity assumption of 
$|a_2\xrightarrow{B}a_3\xrightarrow{C}a_4|/|a_2\xrightarrow{B}a_3|$ $=\frac{3}{2}$.
The final estimate of this path is $4\times\frac{3}{2}\times\frac{5}{2}\times\frac{7}{2}=52.5$.

In the rest of this paper,
we will show how some of the optimistic and pessimistic estimators from literature 
can be modeled as instances of this generic estimator using different CEGs.
We will show that while it is clear that the minimum-weight 
path should be the estimate chosen in the CEG of 
pessimistic estimators, it is
not clear which path should be chosen
from the CEG of optimistic estimators. We will also discuss a shortcoming
of the CEG for optimistic estimators and address it by defining a new CEG.

\section{Optimistic Estimators}
\label{sec:optimistic-estimators}
The estimators that we refer to as {\em optimistic} in this paper use known 
statistics about the input database in formulas that make uniformity and 
independence or conditional independence assumptions. 
The cardinality estimators of several systems fall under this category. We 
focus on three estimators: {\em Markov tables}~\cite{aboulnaga:markov} from 
XML databases, graph summaries~\cite{maduko:md-tree} from RDF 
databases, and the graph catalogue estimator of the Graphflow 
system~\cite{mhedhbi:optimizer} for managing property graphs.  
These estimators are extensions of each other and use the statistics of the cardinalities of small-size joins. 
We give an overview of these estimators and then describe their CEGs, which we will refer to 
as $CEG_{O}$, and then first describe a space of possible optimistic estimates that an optimistic 
estimator can make. We then discuss a shortcoming of the $CEG_O$ when queries contain 
large cycles whose statistics are missing and describe a modification to $CEG_O$ to make more accurate estimations.

\subsection{Overview}
  We begin by giving an overview of the Markov tables estimator~\cite{aboulnaga:markov}, which was used to estimate the cardinalities of paths in XML documents. A Markov table of length $h \ge 2$ stores the cardinality of each path in an XML document's element tree up to length $h$ and uses these to make predications for the cardinalities of longer paths. Table~\ref{table:ex-markov} shows a subset of the entries in an example Markov table for $h=2$ for our running example dataset shown in Figure~\ref{fig:running-dataset}. The formula to estimate a 3-path using a Markov table with $h=2$ is to multiply the cardinality of the leftmost 2-path with the consecutive 2-path divided by the cardinality of the common edge. For example, consider the query $Q_{3p}=\xrightarrow{A}\xrightarrow{B}\xrightarrow{C}$ against the dataset in Figure~\ref{fig:running-dataset}. The formula for $Q_{3p}$ would be: $|\xrightarrow{A}\xrightarrow{B}| \times (|\xrightarrow{B}\xrightarrow{C}|/|\xrightarrow{B}|)$. Observe  that this formula is inspired by the Bayesian probability rule that $Pr(ABC) = Pr(AB)Pr(C|AB)$ but makes a conditional independence assumption between $A$ and $C$, in which case the Bayesian formula would simplify to $Pr(ABC) = Pr(AB)Pr(C|B)$. For $Pr(AB)$ the formula uses the true cardinality $|\xrightarrow{A}\xrightarrow{B}$$|$. For $Pr(C|B)$ the formula makes a uniformity assumption that the number of $C$ edges that each $B$ edge extends to is equal for each $B$ edge and is $r=|\xrightarrow{B}\xrightarrow{C}|/|\xrightarrow{B}|$. Equivalently, this can be seen as an ``average degree'' assumption that 
 on average the $C$-degree of nodes in the $\xrightarrow{B}\xrightarrow{C}$ paths is $r$.
  The result of this formula is $4\times\frac{3}{2}=6$, which underestimates the true cardinality of 7. The graph summaries~\cite{maduko:md-tree} for RDF databases and the graph catalogue estimator~\cite{mhedhbi:optimizer} for property graphs have extended the contents of what is stored in Markov tables, respectively, to other acyclic joins, such as stars, and small cycles, such as triangles, but use the same uniformity and conditional independence assumptions.

\begin{table}[t!]
  \begin{center}
    \label{tab:table1}
    \begin{tabular}{c|c}
      \textbf{Path} & \textbf{|Path|} \\
       \hline
      $\xrightarrow{B}$ & 2\\
       \hline
      $\xrightarrow{A}\xrightarrow{B}$ & 4\\
       \hline
       $\xrightarrow{B}\xrightarrow{C}$ & 3\\
       \hline
	... & ...\\
      \hline
    \end{tabular}
  \end{center}
  \vspace{5pt}
	 \caption{Example Markov table for $h$=2.}
	 \label{table:ex-markov}
	\vspace{-15pt}
\end{table}

\subsection{Space of Possible Optimistic Estimators}
\label{subsec:opt-space}
We next represent such estimators using a CEG that we call $CEG_{O}$. This will help us describe the space of possible estimations that can be made with these estimators. We assume that the given query $Q$ is connected. $CEG_{O}$ consists of the following:

\begin{squishedlist}
\item {\em Vertices:} For each connected subset of relations $S \subseteq \R$ of $Q$, we have a vertex in $CEG_{O}$ with label $S$. This represents the sub-query $\bowtie_{R_i \in S} R_i$. 

\item {\em Edges:} Consider two vertices with labels $S$ and $S'$ s.t., $S \subset S'$. Let $\D$, for {\bf d}ifference be $S' \setminus S$, and let $\E \supset \D$, for {\bf e}xtension be a join query in the Markov table, and let $\I$, for {\bf i}ntersection, be $\E \cap S$. If $\E$ and $\I$ exist in the Markov table, then there is an edge with weight $\frac{|\E|}{|\I|}$ from $S$ to $S'$ in $CEG_{O}$. 
\end{squishedlist}

When making estimates, we will apply two basic rules from prior work that limit the paths considered in $CEG_{O}$. First is that if the Markov table contains size-$h$ joins, the formulas use size $h$ joins in the numerators in the formula. For example, if $h=3$, we do not try to estimate the cardinality of a sub-query $\xrightarrow{A}\xrightarrow{B}\xrightarrow{C}$ by a formula $\xrightarrow{A}\xrightarrow{B} \times \frac{\xrightarrow{B}\xrightarrow{C}}{\xrightarrow{B}}$ because we store 
the true cardinality of $\xrightarrow{A}\xrightarrow{B}\xrightarrow{C}$ in the Markov table.
Second, for cyclic graph queries, which was covered in reference~\cite{mhedhbi:optimizer}, an additional early cycle closing rule is used in the reference when generating formulas. In CEG formulation this translates to the rule that if $S$ can extend to multiple $S'$s and some of them contain additional cycles that are not in $S$, then only such outgoing edges of $S$ to such $S'$ are considered in finding paths.

Even when the previous rules are applied to limit the number of paths considered in a CEG, in general there 
may
be multiple $(\emptyset, Q)$ paths that lead to different estimates. Consider the $CEG_O$ shown in Figure~\ref{fig:opt-cel-h2}
which uses a Markov table of size 2. There are 36 $(\emptyset, Q)$ paths leading to 7 different estimates. Two examples are: 
\begin{squishedlist}
\item  $|\xrightarrow{A}\xrightarrow{B}| \times \frac{|\xrightarrow{B}\xrightarrow{C}|}{|\xrightarrow{B}|} \times \frac{|\xrightarrow{B}\xrightarrow{D}|}{|\xrightarrow{B}|} \times \frac{|\xrightarrow{B}\xrightarrow{E}|}{|\xrightarrow{B}|} = 52.5$
\item  $|\xrightarrow{A}\xrightarrow{B}| \times \frac{|\xrightarrow{B}\xrightarrow{C}|}{|\xrightarrow{B}|} \times \frac{|\xleftarrow{C}\xrightarrow{D}|}{|\xrightarrow{C}|} \times \frac{|\xleftarrow{D}\xrightarrow{E}|}{|\xrightarrow{D}|} = 57.6$
\end{squishedlist}

Similarly, consider the fork query $Q_{5f}$ in Figure~\ref{fig:running-query} and a Markov table with 
up to 3-size joins. The $CEG_O$ of $Q_{5f}$ is shown in Figure~\ref{fig:mt-jg-h3}, which contains multiple paths 
leading to 2 different estimates:

\begin{squishedlist}
\item  $|\xrightarrow{A}\xrightarrow{B}\xrightarrow{C}| \times \frac{|\xleftarrow{C}\xrightrightarrows[E]{D}|}{|\xrightarrow{C}|}$
\item  $|\xrightarrow{A}\xrightarrow{B}\xrightarrow{C}| \times \frac{|\xrightarrow{A}\xrightarrow{B}\xrightarrow{D}|}{|\xrightarrow{A}\xrightarrow{B}|} \times \frac{|\xrightarrow{A}\xrightarrow{B}\xrightarrow{E}|}{|\xrightarrow{A}\xrightarrow{B}|}$
\end{squishedlist}

Both formulas start by using
$|\xrightarrow{A}\xrightarrow{B}\xrightarrow{C}|$. Then, the first
``short-hop'' formula makes one fewer conditional independence
assumption than the ``long-hop'' formula, which is an advantage. In
contrast, the first estimate also makes a uniformity assumption that
conditions on a smaller-size join. 
We can expect this assumption to be
less accurate than the two assumptions made in the long-hop estimate,
which condition on 2-size joins. 
In general, these two formulas can
lead to different estimates.

For many queries, there can be many more than 2 different estimates. 
Therefore, any optimistic estimator implementation needs to make choices about which formulas to use, which corresponds to picking paths in $CEG_{O}$.
Prior optimistic estimators have either left these choices unspecified or described procedures that implicitly pick 
a specific path yet without acknowledging possible other choices or empirically justifying their choice.
Instead, we systematically identify a space of choices that an optimistic estimator can make along two parameters that also capture the choices made in prior work:
\begin{squishedlist}

\item {\em Path length:} The estimator can identify a set of paths to consider based on the path lengths, i.e., number of edges or hops, in $CEG_{O}$, which can be: (i) maximum-hop (\texttt{max-hop}); (ii) minimum-hop (\texttt{min-hop}); or (iii) any number of hops (\texttt{all-hops}). 

\item {\em Estimate aggregator:} Among the set of paths that are considered, each path gives an estimate. The estimator then has to aggregate these estimates to derive a final estimate, for which we identify three heuristics: (i) the largest estimated cardinality path (\texttt{max-aggr}); (ii) the lowest estimated cardinality path (\texttt{min-aggr}); or (iii) the average of the estimates among all paths (\texttt{avg-aggr}).

\end{squishedlist}

\noindent Any combination of these two choices can be used to design an optimistic estimator. 
The original Markov tables~\cite{aboulnaga:markov} chose the \texttt{max-hop} heuristic. In this work, each query was a path, so when the first heuristic is fixed, any path in $CEG_{O}$ leads to the same estimate. Therefore an aggregator is not needed. Graph summaries~\cite{maduko:md-tree} uses the \texttt{min-hop} heuristic and leaves the aggregator unspecified. Finally, graph catalogue~\cite{mhedhbi:optimizer} picks the \texttt{min-hop} and \texttt{min-aggr} aggregator. We will
do a systematic empirical analysis of this space of estimators in Section~\ref{sec:evaluation}.
In fact, we will show that which heuristic to use depends on the structure of the query. For example,
on acyclic queries unlike the choice in reference~\cite{mhedhbi:optimizer}, 
systems can combat the well known underestimation problem of optimistic estimators by picking the `pessimistic' paths, so using \texttt{max-aggr} heuristic. Similarly, we find that using the \texttt{max-hop} heuristic leads generally to highly accurate estimates.

\subsection{$CEG_{OCR}$: Handling Large Cyclic Patterns}
Recall that a Markov table stores the cardinalities of patterns up to some size $h$.
Given a Markov table with $h \geq 2$, optimistic estimators can produce estimates
for any acyclic query with size larger than $h$.   But what about large \emph{cyclic} queries, i.e.,
cyclic queries with size larger than $h$?

Faced with a large cyclic query $Q$ , the optimistic estimators we have described do not actually
produce estimates for $Q$.   Instead, they produce an estimate for a similar acyclic $Q'$ that includes
all of $Q$'s edges but is not closed.
To illustrate this, consider estimating a 4-cycle query in Figure~\ref{fig:cq} using a Markov table with $h$=$3$. Recall that the estimates of prior optimistic estimators are captured by paths in the $CEG_O$ for 4-cycle query, shown in Figure~\ref{fig:square_ceg_o}. Consider as an example, the left most path, which corresponds to the formula: $|\xrightarrow{A}\xrightarrow{B}\xrightarrow{C}| \times |\xrightarrow{B}\xrightarrow{C}\xrightarrow{D}| / |\xrightarrow{B}\xrightarrow{C}|$. 
Note that this formula is in fact estimating a 4-path $\xrightarrow{A}\xrightarrow{B}\xrightarrow{C}\xrightarrow{D}$ rather than the 4-cycle shown in Figure~\ref{fig:cq}.
This is true for each path in $CEG_O$.

More generally, when queries contain cycles of length $> h$, $CEG_O$
breaks cycles in queries into paths. Although optimistic estimators generally tend to underestimate acyclic queries due to the independence assumptions they make, the situation is different for cyclic queries. Since there are often significantly more paths than cycles, estimates over $CEG_O$ can lead to highly inaccurate \emph{over}estimates. 
We note that this problem does not exist if a query contains a cycle $C$ of length $>h$ that contains smaller 
cycles in them, such as a clique of size $h+1$,
because the early cycle closing rule from Section~\ref{subsec:opt-space} will avoid formulas that estimate $C$ as a sub-query (i.e., no $(\emptyset, C)$ sub-path will exist in CEG).

\begin{figure}
  \includegraphics[width=0.5\columnwidth]{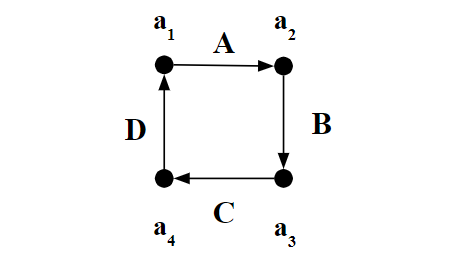}
  \vspace{-3mm}
  \caption{An example of a 4-cycle cyclic query}
  \label{fig:cq}
\end{figure}

\begin{figure*}
  \centering
  \begin{subfigure}{0.49\textwidth}
    \includegraphics[width=\textwidth]{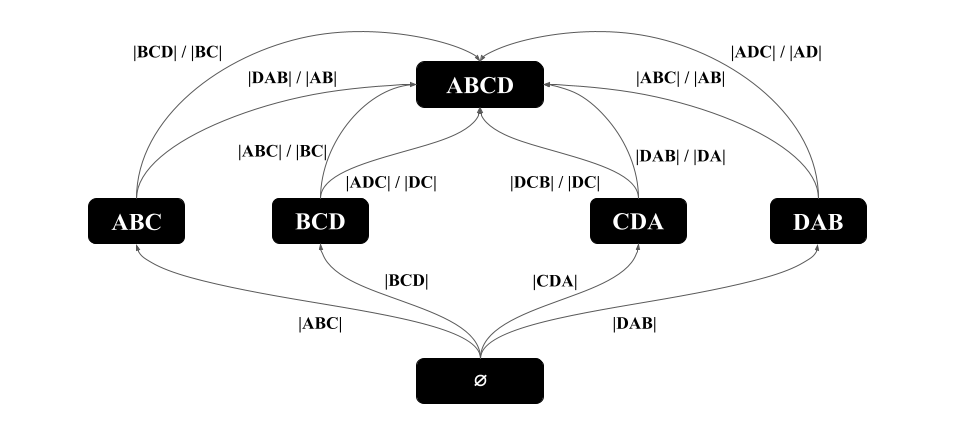}
    \caption{$CEG_O$}
    \label{fig:square_ceg_o}
  \end{subfigure}
  \begin{subfigure}{0.49\textwidth}
    \includegraphics[width=\textwidth]{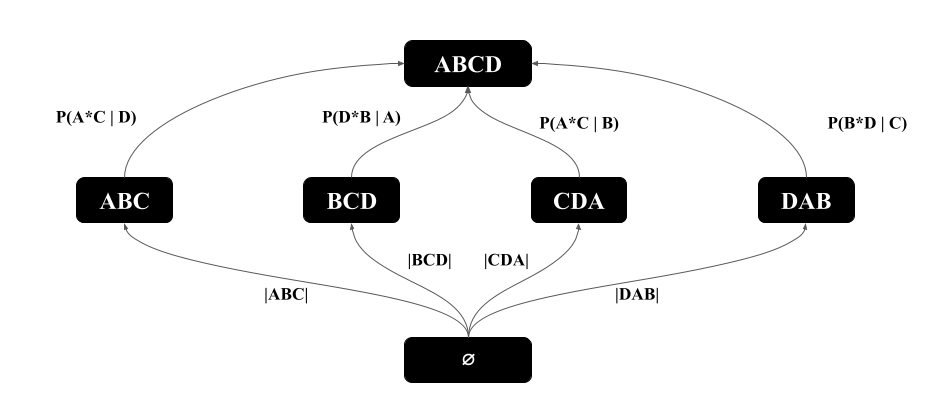}
    \caption{$CEG_{OCR}$}
    \label{fig:square_ceg_ocr}
  \end{subfigure}
  \caption{$CEG_O$ and $CEG_{OCR}$ of Figure ~\ref{fig:cq}.}
\end{figure*}

We next describe an alternative modified CEG to address this problem. 
Consider a query $Q$ with a $k$-cycle $C$ where $k>h$. 
Note that in order to not break cycles into paths, we need CEG edges whose weights capture 
the cycle closing effect when extending a sub-query $S$ that contains $k$$-$$1$ edges of $C$ to a sub-query $S'$
that contains $C$.  
We capture this in a new CEG called $CEG_{OCR}$, for {\bf O}ptimistic {\bf C}ycle closing {\bf R}ate,
which modifies $CEG_O$ as follows. We keep the same vertices as in $CEG_O$ and the same set of edges,
except when we detect two vertices $S$ and $S'$ with the above property. Then, instead of using the weights from
the original $CEG_O$ between $S$ and $S'$, we use pre-computed  cycle closing probabilities.
Suppose the last edge that closes the cycle $C$ is $E_i$ and it is between the query edges $E_{i-1}$ and $E_{i+1}$. 
In the Markov table, we store the probabilities of two connected $E_{i-1}$ and $E_{i+1}$ edges to be connected by an additional $E_i$ edge to close a cycle. We denote this statistic as $P(E_{i-1} * E_{i+1} | E_i)$. We can compute $P(E_{i-1} * E_{i+1} | E_i)$ by 
computing all paths that start from $E_{i-1}$ and end with $E_{i+1}$ of varying lengths and then counting the number of $E_i$ 
edges that close such paths to cycles. 
On many datasets there may be a prohibitively many such paths, so we can sample $p$ of such paths. 
Suppose these $p$ paths lead to $c$ many cycles, then we can take the probability as $c/p$. 
In our implementation we perform sampling through random walks that start from $E_{i-1}$ and end at $E_{i+1}$ but other sampling strategies can also be employed.  
Figure~\ref{fig:square_ceg_ocr} provides the $CEG_{OCR}$ for the 4-cycle query in Figure~\ref{fig:cq}.
We note that the Markov table
for $CEG_{OCR}$ requires computing additional $P(E_{i-1} * E_{i+1} | E_i)$ statistics 
that $CEG_O$ does not require.  
The number of entries is at most $O(L^3)$ where $L$ is the number of edge labels in the dataset. 
For many datasets, e.g., all of the ones we used in our evaluations, $L$ is small and in the order of 10s or 100s,
so even in the worst case these entries can be stored in several MBs. In contrast, storing 
large cycles with $h$$>$$3$ edges could potentially require $\Theta(L^h)$ more entries.

\section{Pessimistic Estimators}
\label{sec:pessimistic-estimators}
Join cardinality estimation is directly related to the following fundamental question: 
Given a query $Q$ and set of statistics over the relations $R_i$, such as their cardinalities or degree information about values in different columns, what is the
worst-case output size of $Q$? Starting from the seminal result by Atserias, Grohe, 
and Marx in 2008~\cite{atserias:agm}, several upper bounds have been provided to 
this question under different known statistics. For example the initial upper bound from reference~\cite{atserias:agm}, now called the {\em AGM bound}, used only the 
cardinalities of each relation, while later bounds, DBPLP~\cite{joglekar:degree}, 
MOLP~\cite{joglekar:degree}, and CLLP~\cite{khamis:cllp} used maximum 
degrees of the values in the columns and improved the AGM bound.
Since these bounds are upper bounds on the query size, they can be used 
as {\em pessimistic estimators}. This was done recently by Cai et al.~\cite{cai:pessimistic} in 
an actual estimator implementation. We refer to this as the CBS estimator, 
after the names of the authors. 

In this section, we show a surprising connection between the optimistic estimators from Section~\ref{sec:optimistic-estimators} and the recent pessimistic estimators~\cite{joglekar:degree, cai:pessimistic}. Specifically,
in Section~\ref{subsec:molp}, we show that similar to optimistic estimators, MOLP (and CBS) can also be modeled as an estimator using a CEG. 
CEGs further turn out to be useful mathematical tools to prove properties of pessimistic estimators.
We next show applications of CEGs is several of our proofs to
obtain several theoretical results that provide insights into the pessimistic estimators.  
Section~\ref{subsec:wbs} reviews the CBS estimator. 
Using our CEG framework,  we show that in fact the CBS estimator is 
equivalent to the MOLP bound on acyclic queries on which it was evaluated in 
reference~\cite{cai:pessimistic}. 
In Section~\ref{subsec:wbs}, we also review the bound sketch refinement of the 
CBS estimator from reference~\cite{cai:pessimistic}, which we show can also be applied
to any estimator using a CEG, specifically the optimistic ones we cover in this paper.
Finally, 
\iflongversion
Appendix~\ref{app:dbp}
\else
the longer version of our paper~\cite{chen:ceg-tr}
\fi
reviews 
the DBPLP bound and provides 
an alternative proof  that MOLP is tighter than DBPLP that also uses CEGs in the proof.

\subsection{MOLP}
\label{subsec:molp}

MOLP was defined in reference~\cite{joglekar:degree} as a 
tighter bound than the AGM bound that uses additional degree 
statistics about input relations  that AGM bound does not use.
We first review the formal notion of a degree. Let $\mathcal{X}$ be a subset 
of the attributes $\mathcal{A}_i$ of
some relation $\R_i$, and let $v$ be a possible value of $\mathcal{X}$.
The {\em degree} of $v$ in $\R_i$ is the number of times $v$ occurs in
$\R_i$, i.e. $deg(\mathcal{X}(v), \R_i) =|\{t \in \R_i | \pi_{\mathcal{X}}(t) =
v\}|$.
For example, in Figure~\ref{fig:running-dataset}, $deg(s(3), E)=3$
because the outgoing $E$-degree of vertex 3 is 3.
Similarly $deg(d(2), A)$ is 1 because the incoming $A$-degree of
vertex 2 is 1.
We also define $deg(\mathcal{X}, \R_i)$ to be the maximum degree in $\R_i$ 
of any value $v$ over $X$, i.e., $deg(X, \R_i)=\max_v deg(\mathcal{X}(v), \R_i)$.
So, $deg(d, A)=3$ because vertex 13 has 3 incoming $A$ edges, which is the maximum
A-in-degree in the dataset.
The notion of degree can be generalized to $deg(X(v), Y, \R_i)$, which 
refers to the  ``degree of a value $v$ over attributes $X$ in $\pi_Y \R_i$'', which 
counts the number of times $v$ occurs in $\pi_Y(\R_i)$. Similarly, we let 
$deg(X, Y, \R_i)=\max_v deg(X(v), Y, \R_i)$. Suppose a system has stored
 $deg(X, Y, \R_i)$ statistics for each possible $\R_i$ and  
 $X \subseteq Y \subseteq \A_i$. MOLP is:

\vspace{-10pt}
\begin{align*}
\begin{split}
& \text{\bf{Maximize\:}} s_{\A} \\
 & s_{\emptyset} = 0 \\
 & s_{X} \le s_{Y}, \: \forall X \subseteq Y \\
 & s_{Y \cup E} \le s_{X \cup E} \text{\small{+}} \log(deg(X, Y, \R_i)), \forall X, Y, E \text{$\subseteq$} \A, X \text{$\subseteq$} Y \text{$\subseteq$} \A_i
\end{split}
\end{align*}
The base of the logarithm can be any constant and we take it as 2.
Let $m_A$ be the optimal value of MOLP. Reference~\cite{joglekar:degree} has shown that $2^{m_A}$ is an upper bound on the size of $Q$. 
For example, in our running example, the optimal value of these inequalities is 96, which is an overestimate of the true cardinality of 78. It is not easy to directly see the solution of the MOLP on our running example. However, we will next show that we can represent the MOLP bound as the cost of minimum-weight $(\emptyset, Q)$ path in a CEG that we call $CEG_M$.

\noindent {\bf MOLP CEG ($CEG_M$):}
Let $Q_Z$ be the projection of $Q$ onto attributes $Z$, so $Q_Z = \Pi_Z Q$. Each variable $s_Z$ in MOLP represents the maximum size of $Q_Z$, i.e., the tuples in the projection of $Q_Z$ that contribute to the final output.  We next interpret the two sets of inequalities in MOLP:
\begin{squishedlist}
\item {\em Extension Inequalities} $s_{Y \cup E} \le s_{X \cup E} + \log(deg(X, Y, \R_i))$:  These inequalities intuitively indicate the following: each tuple $t_{X \cup E} \in Q_{X \cup E}$ can extend to at most $deg(X, Y, \R_i)$ $Q_{Y \cup E}$ tuples. For example, in our running example, let $X$$=$$\{a_2\}$, $Y$$=$$\{a_2a_3\}$ and $E$$=$$\{a_1\}$. So both $X$ and $Y$ are subsets of $B(a_2, a_3)$. The inequality indicates that each $a_1a_2$ tuple, so an $R_A$ tuple, can extend to at most $deg(\{a_2\}, \{a_2, a_3\}, B(a_2, a_3))$$=$$deg(a_2, B)$ $a_1a_2a_3$ tuples. This is true, because  $deg(a_2, B)$ is the maximum degree of any $a_2$ value in $B$ (in graph terms the maximum degree of any vertex with an outgoing $B$ edge). 
\item {\em Projection Inequalities} $s_{X} \le s_{Y}$ $(\forall X \subseteq Y)$: These indicate that the number of tuples in $Q_X$ is at most the number of $Q_Y$, if $Y$ is a larger subquery.
\end{squishedlist}

With these interpretations we can now construct $CEG_M$.
\begin{squishedlist}
\item {\em Vertices:} For each $X \subseteq \A$ there is a vertex. This represents the subquery $\Pi_{X} Q$.

\item {\em Extension Edges:}
Add an edge with weight $\log(deg(X, Y,$  $\R_i))$ between any $W_1=X \cup E$ and $W_2=Y \cup E$, for which there is $s_{Y \cup E} \le s_{X \cup E} + \log(deg(X, Y, \R_i))$ inequality. Note that there can be multiple edges between $W_1$ and $W_2$ corresponding to inequalities from different relations.

\item {\em Projection Edges:} $\forall X \subseteq Y$, add an edge with weight 0 from $Y$ to $X$. These edges correspond to projection inequalities and intuitively indicate that, in the worst-case instances, $\Pi_Y Q$ is always as large as $\Pi_X Q$. 
\end{squishedlist}

Figure~\ref{fig:running-molp-ceg} shows the $CEG_M$ of our running example. For simplicity, we use actual degrees instead of their logarithms as edge weights and omit the projection edges in the figure. Below we use $(\emptyset, \A)$ instead of  $(\emptyset, Q)$, to represent the 
bottom-to-top paths in  $CEG_M$  because $\Pi_{\A}Q=Q$. 
\begin{figure}[t!]
	\centering
	\captionsetup{justification=centering}
    	\includegraphics[width=\columnwidth]{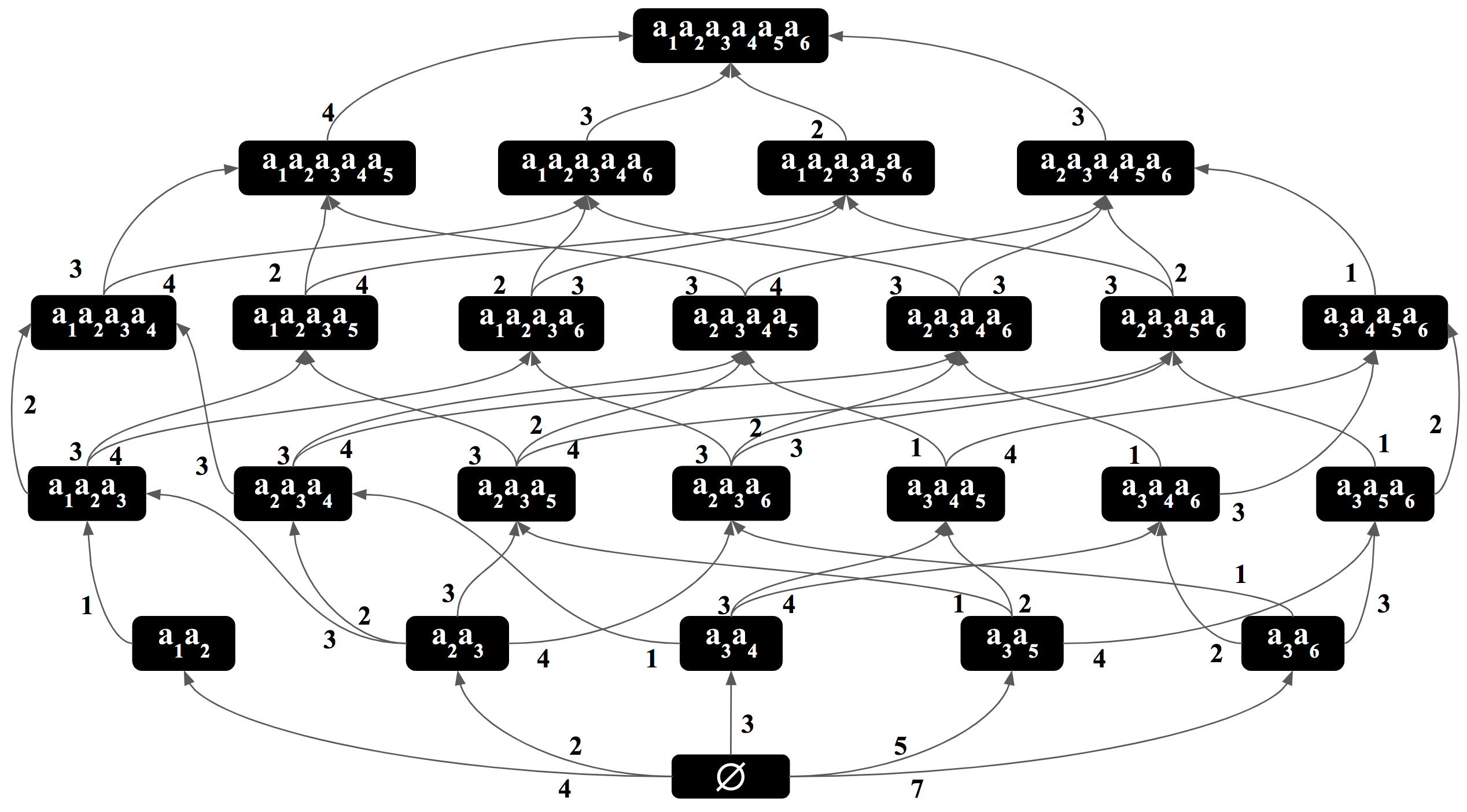}
  \caption{$CEG_M$ for query $Q_{5f}$ in Figure~\ref{fig:running-query}.}
  \label{fig:running-molp-ceg}
  \vspace{-15pt}
\end{figure}

\begin{theorem}
\label{thm:mo-sp}
Let $Q$ be a query whose degree statistics $deg(X, Y, R_i)$ for each $\R_i$ and $X \subseteq Y \subseteq \A_i$ is known. The optimal solution $m_A$ to the MOLP of $Q$ is equal to the weight of the minimum-weight $(\emptyset, \A)$ path in CEG$_{M}$.
\end{theorem}

\begin{proof}
Our proof consists of two steps. First we show that any feasible solution $v$ to MOLP has a value at most the weight of any  $(\emptyset, \A)$ path. Then we show that a particular feasible solution, which we call $v_{CEG}$, is exactly the weight of the minimum-weight $(\emptyset, \A)$ path.
Let $v$ be a feasible solution to the $MOLP_Q$. We refer to the value of $v$, so the value $s_{\A}$ takes in $v$, simply as $s_{\A}$. Let $P$ be any $(\emptyset, \A)$ path in $CEG_{M}$. Let $w(P)$ be the weight of $P$. Suppose w.l.o.g. that $P$$=$$(\emptyset)$$\xrightarrow{e_0}$$(E_1)$...$(E_k)$$\xrightarrow{e_k}$$(E_{k+1}=\A)$ and for the purpose of contradiction that $w(P)=w(e_0)+....+w(e_k)$$<$$s_{\A}$. If this is the case, we can inductively (from $i$$=$$k+1$ down to $0$) show that $w(e_0)+....+w(e_{i-1}) < s_{E_i}$. The base case for $s_{E_{k+1}}$ holds by our assumption. Suppose $w(e_0)+....+w(e_{i}) < s_{E_{i+1}}$ by induction hypothesis. Then consider the inequality in $MOLP_Q$ that corresponds to the $(E_i)$$\xrightarrow{e_{i}}$$(E_{i+1})$ edge $e_i$. There are two possible cases for this inequality:

\noindent Case 1: $e_i$ is a projection edge, so $w(e_{i})=0$ and we have an inequality of $s_{E_{i+1}}$$\le$$s_{E_i}$, so $w(e_0)+....+w(e_{i}) < s_{E_i+1} \le s_{E_{i}}$, so $w(e_0)+....+w(e_{i-1}) < s_{E_{i}}$. 

\noindent Case 2: $e_i$ is an extension edge, so we have an inequality of $s_{E_{i+1}}$$\le$$s_{E_i} + w(e_{i})$, so $w(e_0)+....+w(e_{i}) < s_{E_i+1} \le s_{E_{i}} + w(e_{i})$, so 
$w(e_0)+....+w(e_{i}) < s_{E_{i}}$, completing the inductive proof. However, this implies that $0 < s_{\emptyset}$, which contradicts the first inequality of MOLP, completing the proof that any feasible solution $v$ to the MOLP is at most the weight of any $(\emptyset, \A)$ path in $CEG_M$. 

Next, let $v_{CEG}$ be an assignment of variables that sets each $s_X$ to the weight of the minimum-weight $(\emptyset, X)$ path in $CEG_M$. Let $v_X$ be the value of $s_X$ in  $v_{CEG}$. We show that $v_{CEG}$ is a feasible solution to $MOLP_Q$. First, note that in $v_{CEG}$ $s_{\emptyset}$ is assigned a value of 0, so the first inequality of MOLP holds. Second, consider any extension inequality $s_{Y \cup E} \le s_{X \cup E} + \log(deg(X, Y, \R_i))$, so $CEG_M$ contains an edge from $X \cup E$ to $Y \cup E$ with weight $\log(deg(X, Y, \R_i))$. By definition of minimum-weight paths, $v_{Y \cup E} \le v_{X \cup E} + \log(deg(X, Y, \R_i))$. Therefore, in $v_{CEG}$ all of the extension inequalities hold. Finally, consider a projection inequality  $s_{X} \le s_{Y}$, where $X \subseteq Y$, so $CEG_M$ contains an edge from vertex $Y$ to vertex $X$ with weight 0. By definition of minimum-weight paths, $v_X \le v_Y + 0$, so all of these inequalities also hold. Therefore, $v_{CEG}$ is indeed a feasible solution to $MOLP_Q$. Since any solution to MOLP has a value smaller than the weight of any path in $CEG_M$, we can conclude that $v_{\A}$ in $v_{CEG}$, which is the minimum-weight $(\emptyset, \A)$ path, is equal to $m_A$.
\end{proof}

With this connection, readers can verify that the MOLP bound in our running example is 96 by inspecting the paths in Figure~\ref{fig:running-molp-ceg}. In this CEG, the minimum-weight $(\emptyset, \A)$ path has a length of 96 (specifically $\log_2(96)$), corresponding to the leftmost path in Figure~\ref{fig:running-molp-ceg}. We make three observations. 

\noindent {\em Observation 1:}
Similar to the CEGs for optimistic estimators, each $(\emptyset, \A)$ path in $CEG_M$ corresponds to a sequence of extensions from $\emptyset$ to $Q$ and is an estimate of the cardinality of $Q$. 
For example, the rightmost path in Figure~\ref{fig:running-molp-ceg} indicates that there are $7$ 
$a_3a_6$'s (so $a_3\xrightarrow{E}a_6$ edges), each of which extends to at most 3 $a_3a_5a_6$'s (so $a_5\xleftarrow{D}a_3\xrightarrow{E}a_6$ edges), and so forth. This path yields 7x3x2x1x3=126 many possible outputs. Since we are using maximum degrees on the edge weights, each $(\emptyset, \A)$ path is by construction an upper bound on $Q$. So any path in $CEG_M$ is a pessimistic estimator. This is an alternative simple proof to the following proposition:

\begin{proposition}[Prop. 2~\cite{joglekar:degree}] 
\label{prop:mo-out}
Let $Q$ be a join query and $OUT$ be the output size of $Q$, then $OUT$$\le$$2^{m_{\A}}$.\footnote{This is a slight variant of Prop. 2 from reference~\cite{joglekar:degree}, which state that another bound, called the {\em MO bound}, which adds a preprocessing step to MOLP, is an upper bound of $OUT$.}
\end{proposition}
\begin{proof}
Since for any ($\emptyset$, $\A$) path $P$ in $CEG_M$, $OUT$ $\le$ $2^{w(P)}$ and by Theorem~\ref{thm:mo-sp}, $m_{\A}$ is equal to the weight of the minimum-weight ($\emptyset$, $\A$) path in $CEG_M$, $OUT$ $\le$ $2^{m_{\A}}$.
\end{proof}

\noindent {\em Observation 2:} Theorem~\ref{thm:mo-sp} implies that MOLP can be solved using a combinatorial algorithm, e.g., Dijkstra's algorithm, instead of a numeric LP solver. 

\noindent{{\em Observation 3:} Theorem~\ref{thm:mo-sp} implies that we can simplify MOLP by removing the projection inequalities, which correspond to the edges with weight 0 in $CEG_M$. To observe this, 
consider any $(\emptyset, \A)$ path $P$$=$$(\emptyset)$$\xrightarrow{e_0}$$(E_1)$...$(E_k)$$\xrightarrow{e_k}$$(\A)$ and consider its first projection edge, say $e_i$. In 
\iflongversion
Appendix~\ref{app:projection},
\else
the long version of our paper~\cite{chen:ceg-tr},
\fi
we show that we can remove $e_i$ and construct
an alternative path with at most the same weight as $P$ but with one fewer projection edge, showing
that MOLP linear program can be simplified by only using the extension inequalities.

\vspace{-5pt}
\subsubsection{Using Degree Statistics of Small-Size Joins}
MOLP can directly integrate the degree statistics from results of small-size joins. For example, if a system knows the size of $Q_{RS}=R(a_1, a_2) \bowtie S(a_2, a_3)$, then the MOLP can include the inequality that $s_{a_1a_2a_3} \le \log (|Q_{RS}|)$. Similarly, the extension inequalities can use the degree information from $Q_{RS}$ simply by taking the output of $Q_{RS}$ as an additional relation in the query with three attributes $a_1$, $a_2$, and $a_3$. 
When comparing the accuracy of the MOLP bound with optimistic estimators, we will ensure that MOLP uses the degree information of the same small-size joins as optimistic estimators, ensuring that MOLP uses a strict superset of the statistics used by optimistic estimators.

\vspace{-5pt}
\subsection{CBS and Bound Sketch Optimization}
\label{subsec:wbs}
We review the CBS estimator very briefly and refer the reader to reference~\cite{cai:pessimistic} for details. CBS estimator has two subroutines {\em Bound Formula Generator (BFG)} and {\em Feasible Coverage Generator (FCG)} (Algorithms  	1 and 2 in reference~\cite{cai:pessimistic}) that, given a query $Q$ and the degree statistics about $Q$, generate a set of bounding formulas. A coverage is a mapping ($R_j$, $A_j)$ of a subset of the relations in the query to attributes such that each $A_j \in \A$ appears in the mapping. A bounding formula is a multiplication of the known degree statistics that can be used as an upper bound on the size of a query. In 
\iflongversion
Appendix~\ref{app:wbs-molp}, 
\else
the long version of our paper~\cite{chen:ceg-tr},
\fi
we show using our CEG framework that in fact the MOLP bound is at least as tight as the CBS estimator on general acyclic queries and is exactly equal to the CBS estimator over acyclic queries over binary relations, which are the 
queries which reference~\cite{cai:pessimistic} used.
Therefore BFG and FCG can be seen as a method for solving the MOLP linear program on acyclic queries over binary relations, albeit in a brute force manner by enumerating all paths in $CEG_M$. We do this by showing that each path in $CEG_M$ corresponds to a bounding formula and vice versa. These observations allow us to connect two worst-case upper bounds from literature using CEGs. Henceforth, we do not differentiate between MOLP and the CBS estimator.
It is important to note that a similar connection between MOLP and CBS cannot be established for cyclic queries.
This is because, although not explicitly acknowledged in reference~\cite{cai:pessimistic}, on cyclic queries, the covers that FCG generates may not be safe, i.e., the output of BFG may not be a pessimistic output. 
We provide a counter example in 
\iflongversion
Appendix~\ref{app:wbs-cyclic}.
\else
 the long version of our paper~\cite{chen:ceg-tr}.
 \fi
In contrast, MOLP generates a pessimistic estimate for arbitrary, so both acyclic or cyclic, queries. 

\subsubsection{Bound Sketch}
We next review an optimization that was described in reference~\cite{cai:pessimistic} 
to improve the MOLP bound. 
Given a partitioning budget $K$, for each bottom-to-top path in CEG$_M$, the 
optimization partitions the input relations into multiple pieces and derives $K$ many subqueries of $Q$.
Then the estimate for $Q$ is the sum of estimates of all $K$ subqueries. Intuitively, partitioning
decreases the maximum degrees in subqueries to yield better estimates, specifically
their sum is guaranteed to be more accurate than making a direct estimate for $Q$.  
We give an overview of the optimization 
here and refer the reader to reference~\cite{cai:pessimistic} for details. 

We divide the edges in $CEG_M$ into two. Recall that each edge $W_1 \xrightarrow{e_j} W_{2}$ in $CEG_M$ is constructed from an inequality of $s_{Y \cup E} \le s_{X \cup E} + \log(deg(X, Y, R_i))$ in MOLP. We call $e_j$ (i) an unbound edge if $X=\emptyset$, i.e., the weight of $e_j$ is $|R_i|$; (ii) a bound edge if $X \neq \emptyset$, i.e., the weight of $e_j$ is actually the degree of some value in a column of $R_i$. Note that unbound edge extends $W_1$ exactly with attributes $\A_i$, i.e., $W_2 \setminus W_1 = \A_i$ and a bound edge with attributes $Y$, i.e., $W_2 \setminus W_1 = Y$. Below, we refer to these attributes as ``extension'' attributes.

\noindent {\bf Step1:} For each $p=(\emptyset, \A)$ path in $CEG_M$ (so a bounding formula in the terminology used in reference~\cite{cai:pessimistic}), let $S$ be the join attributes that are not extension attributes through a bounded edge. For each attribute in $S$, allocate $K^{1/|S|}$ partitions. For example, consider the path $P_1$$=$$\emptyset$$\xrightarrow{|B|}a_2a_3$$\xrightarrow{deg(a_3, C)}$$a_{2-4}$$\xrightarrow{deg(a_2, A)}$$a_{1-4}$$\xrightarrow{deg(a_3, E)}$$a_{1-4}a_6$$\linebreak \xrightarrow{deg(a_3, D)}$$a_{1-6}$ in the $CEG_M$ of $Q_{5f}$ from Figure~\ref{fig:running-molp-ceg}, where $a_{i-j}$ refers to $a_ia_{i+1}...a_j$. Then both $a_2$ and $a_3$ would be in $S$. For path $P_2=\emptyset\xrightarrow{|A|}a_1a_2\xrightarrow{deg(a_2, B)}a_{1-3}\xrightarrow{deg(a_3, C)}a_{1-4}\xrightarrow{deg(a_3, D)}a_{1-5}\xrightarrow{deg(a_3, E)}a_{1-6}$, only $a_2$ would be in in $S$.

\noindent {\bf Step2:} Partition each relation $R_i$ as follows. Let $PA_i$, for {\bf p}artition {\bf a}ttributes, be $PA_i = S \cap \A_i$ and $z$ be $|PA_i|$. Then partition $R_i$ into $K^{z/|S|}$ pieces using $z$ hash functions, each hashing a tuple $t \in R_i$ based on one of the attributes in $PA_i$ into $\{0, ..., K^{1/|S|}-1\}$. For example, the relation $B$ in our example path $P_1$ would be partitioned into 4, $B_{00}$, $B_{01}$, $B_{10}$, and $B_{11}$.

\noindent {\bf Step3:} Then divide $Q$ into $K$ components $Q_{0...0}$, to \linebreak $Q_{K^{1/|S|}-1, ..., K^{1/|S|}-1}$, such that $Q_{j_1, ..., j_z}$ contains only the partitions of each relation $R_i$ that matches the $\{j_1, ..., j_z\}$ indices. For example, in our example, $Q_{0...0}$ is $A_0 \bowtie B_{0,0} \bowtie C_0 \bowtie D_0 \bowtie E_0$. This final partitioning is called the bound sketch of $Q$ for path $p$.

\vspace{-1pt}
\subsubsection{Implementing Bound Sketch in Opt. Estimators}
Note that a bound sketch can be directly used to refine any estimator using a CEG, as it is 
a general technique to partition $Q$ into subqueries based on each path $p$ in a CEG. 
The estimator can then sum the estimates for each subquery to generate an estimate for $Q$.
Specifically, we can use a bound sketch to refine 
optimistic estimators and we will evaluate its benefits in Section~\ref{sec:eval-partitioning}. 
Intuitively, one advantage of using a bound sketch is that the tuples that hash to different 
buckets of the join attributes are guaranteed to not produce outputs and they never appear
in the same subquery. This can make the uniformity assumptions in the optimistic estimators
more accurate because two tuples that hashed to the same bucket of an attribute
are more likely to join.

We implemented the bound sketch optimization for optimistic estimators as follows. Given a partitioning budget $K$ and a set of queries in a workload, we worked backwards from the queries to find the necessary subqueries, and for each subquery the necessary statistics that would be needed are stored in the Markov table. For example, for $Q_{5f}$, one of the formulas that is needed is: $|a_1$$\xrightarrow{A_0}$$a_2$$\xrightarrow{B_{00}}$$a_3|$
$\frac{|a_2\xrightarrow{B_{00}}a_3\xrightarrow{C_0}a_4|}{|a_2\xrightarrow{B_{00}}a_3|}\frac{|a_4\xleftarrow{C_0}a_3\xrightarrow{D_0}a_5|}{|a_3\xrightarrow{C_0}a_4|}\frac{|a_5\xleftarrow{D_0}a_3\xrightarrow{E_0}a_6|}{a_3\xrightarrow{D_0}a_5}$, so we ensure that our
Markov table has these necessary statistics.

\section{Evaluation}
\label{sec:evaluation}
We next present our experiments, which aim to answer five questions: 
(1) Which heuristic out of the space we identified in Section~\ref{subsec:opt-space} leads 
to better accuracy for optimistic estimators, and why? We aim to answer this question for acyclic queries on $CEG_O$
 and cyclic queries on $CEG_O$ and $CEG_{OCR}$.
(2) For cyclic queries, which of these two CEGs lead to more accurate results under their
best performing heuristics?
(3) How much does the bound-sketch optimization improve 
the optimistic estimators' accuracy?
(4) How do optimistic and pessimistic
estimators, which are both summary-based estimators, compare against each other
and other baseline summary-based estimators from the literature?
(5) How does the best-performing optimistic estimator compare 
against Wander Join~\cite{li:wanderjoin, park:gcare}, the state-of-the-art sampling-based estimator?
Finally, as in prior work, we use the RDF-3X system~\cite{neumann:rdf3x} to verify that
our estimators' more accurate estimations lead more performant query plans.

Throughout this section, except in our first experiments, 
where we set $h=3$, we use a Markov table of size $h=2$ for optimistic estimators.
We generated workload-specific Markov tables, which  
required less than 0.6MB memory for any workload-dataset combination for $h=2$ or $h=3$. For
$CEG_{OCR}$, which requires computing the cycle closing rates, the size was slightly higher but at most 0.9MB.
 Our code, datasets, and queries are publicly available at \url{https://github.com/cetechreport/CEExperiments} 
 and \url{ https://github.com/cetechreport/gcare}.

\subsection{Setup, Datasets and Workloads}
For all of our experiments, we use a single machine with two Intel E5-2670 at 2.6GHz CPUs,
each with 8 physical and 16 logical cores, and 512 GB of RAM.
We represent our datasets as labeled graphs and queries as edge-labeled subgraph queries but
our datasets and queries can equivalently be represented as relational tables, one for each edge label, 
and SQL queries. We focused on edge-labeled queries for simplicity. Estimating queries with
vertex labels can be done in a straightforward manner both for optimistic and pessimistic estimators
e.g., by extending Markov table entries to have vertex labels as was done in
reference~\cite{mhedhbi:optimizer}.
We used a total of 6 real-world datasets, shown in Table~\ref{table:datasets}, and 5 workloads on these datasets. 
Our dataset and workload combinations are as follows.

\begin{table}[t!]
  \small
  \begin{center}
    \label{tab:table2}
    \begin{tabular}{|c|c|c|c|c|}
      \hline
      \textbf{Dataset} & \textbf{Domain} & \textbf{|V|} & \textbf{|E|} & \textbf{|E. Labels|} \\
      \hline
       IMDb & Movies & 27M & 65M & 127 \\
       \hline
       YAGO & Knowledge Graph & 13M & 16M & 91 \\
       \hline
       DBLP & Citations & 23M & 56M & 27 \\
       \hline
       WatDiv & Products & 1M & 11M & 86 \\
       \hline
       Hetionet & Social Networks & 45K & 2M & 24 \\
       \hline
       Epinions & Consumer Reviews & 76K & 509K & 50 \\
       \hline
    \end{tabular}
  \end{center}
  \vspace{5pt}
   \caption{Dataset descriptions.}
	 \label{table:datasets}
  \vspace{-20pt}
\end{table}

\noindent {\bf IMDb and JOB:} The IMDb relational database, together with a workload called JOB, has been used for cardinality estimation studies in prior work~\cite{cai:pessimistic, leis:job}. We created property graph versions of the this database and workload as follows.
IMDb contains three groups of tables: (i) {\em entity tables} representing entities, such as actors (e.g., \texttt{name} table), movies, and companies; (ii) {\em relationship tables} representing many-to-many relationships between the entities (e.g., the \texttt{movie\-_companies} table represents relationships between movies and companies); and (iii) {\em type tables}, which denormalize the entity or relationship tables to indicate the types of entities or relationships. 
We converted each row of an entity table to a vertex. We ignored vertex types because many queries in the JOB workload have no predicates on entity types. Let $u$ and $v$ be vertices representing, respectively, rows $r_u$ and $r_v$ from tables $T_u$ an $T_v$. We added two sets of edges between $u$ and $v$: (i) a {\em foreign key edge} from $u$ to $v$ if the primary key of row $r_u$ is a foreign key in row $r_v$; (ii) a {\em relationship edge} between $u$ to $v$ if a row $r_{\ell}$ in a relationship table $T_{\ell}$ connects row $r_u$ and $r_v$. 

We then transformed the JOB workload~\cite{leis:job} into equivalent subgraph queries on our transformed graph.
We removed non-join predicates in the queries since we are focusing on join cardinality estimations, and we encoded equality predicates on types directly on edge labels.
This resulted in 7 join query templates, including four 4-edge queries, two 5-edge queries, and one 6-edge query.
All of these queries are acyclic.
There were also 2- and 3-edge queries, which we ignored because as their estimation is trivial with
Markov tables of size 3. 
We generated 100 query instances from each template by choosing one edge label uniformly at random for
each edge, while ensuring that the output of the query was non-empty. The final workload contained a total
of 369 specific query instances. 

\begin{figure}[t!]
	\centering
	\captionsetup{justification=centering}
    	\includegraphics[width=\columnwidth]{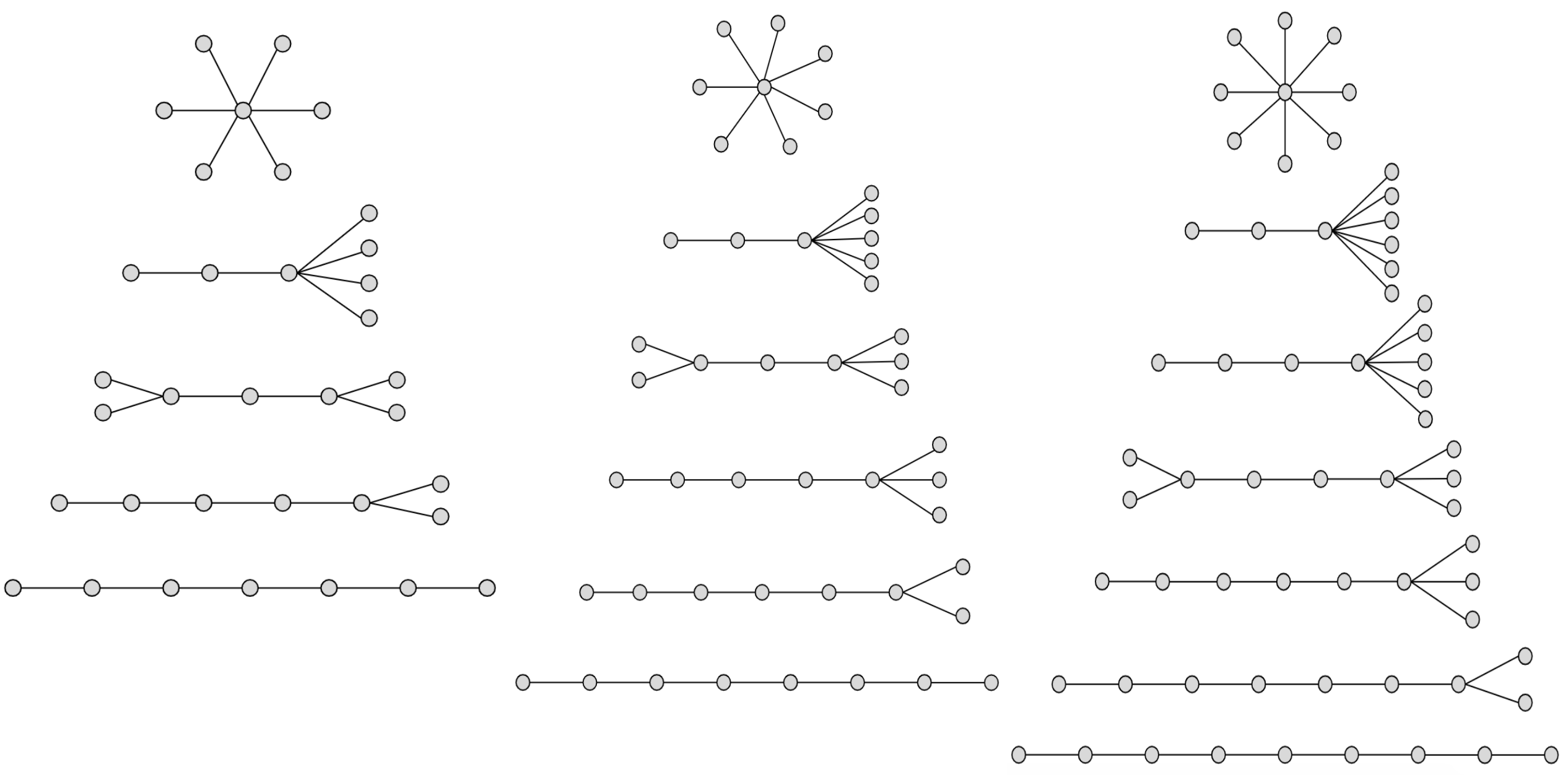}
	\vspace{-10pt}
  \caption{Our full acyclic query templates. The directions of the edges are not shown in the figure.}
	\vspace{-10pt}
	\label{fig:acyc-workload}
\end{figure}

\noindent {\bf YAGO and G-CARE-Acyclic and G-CARE Cyclic Workloads:} 
G-CARE~\cite{park:gcare} is a recent cardinality estimation benchmark for subgraph queries. From this benchmark we took the YAGO knowledge graph dataset and the acyclic and cyclic query
workloads for that dataset.
The acyclic workload contains 382 queries generated from 
query templates with 3-, 6-, 9-, and 12-edge star and path queries, as well as randomly
generated trees. We will refer to this workload as \texttt{G-CARE-Acyclic}. 
The cyclic query workload contains 240 queries generated from templates with 6-, and 9-edge cycle, 6-edge clique, 6-edge flower, and 6- and 9-edge petal queries. We will refer to this workload as \texttt{G-CARE-Cyclic}.
The only other large dataset from G-CARE that was available was LUBM, which contained only 6 queries, 
so we did not include it in our study.

\noindent {\bf DBLP, WatDiv, Hetionet, and Epinions Datasets and Acyclic and Cyclic Workloads:}
We used four other datasets: (i) Hetionet: a biological network; (ii) DBLP: a real knowledge graph; (iii) WatDiv: a synthetic knowlege graph; and (iv) Epinions: a real-world social network graph.
Epinions is a dataset that by default does not have any edge labels. We added a random set of 50 edge labels to Epinions. Our goal in using Epinions was to test whether our experimental observations also hold on a graph that is guranteed to not have any correlations between edge labels.
For these datasets we created one acyclic and one cyclic query workload, which we 
refer to as  \texttt{Acyclic} and  \texttt{Cyclic}.
The \texttt{Acyclic} workload contains queries generated from 6-, 7-, or 8-edge templates.
We ensured that for each query size $k$, we had patterns of every possible depth. 
Specifically for any $k$, the minimum depth of any query is 2 (stars) 
and the maximum is $k$ (paths).  For each depth $d$ in between, we picked a pattern.
These patterns are shown in Figure~\ref{fig:acyc-workload}. 
Then, we generated 20 non-empty instances of each template  by putting one 
edge label uniformly at random on each edge, which yielded 360 queries in total. 
The \texttt{Cyclic} workload contains queries generated from templates used in reference~\cite{mhedhbi:optimizer}: 4-edge cycle, 5-edge diamond with a crossing edge, 6-cycle, complete graph $K_4$, 6-edge query of two triangles with a common vertex, 8-edge query of a square with two triangles on adjacent sides, and a 7-edge query with a square and a triangle.
We then randomly generated instances of these queries by randomly matching each edge of the query template
one at a time in the datasets with a 1 hour time limit for each dataset. 
We generated 70 queries for DBLP, 212 queries for Hetionet, 129 queries for WatDiv, and 
394 queries for Epinions.

\vspace{-3pt}
\subsection{Space of Optimistic Estimators}
We begin by comparing the 9 optimistic estimators we described on the two optimistic CEGs
we defined. In order to set up an experiment in which we could test all of the 9 possible optimistic estimators, we used a Markov table that contained up to 3-size joins (i.e., h=3). A Markov table with only 2-size joins can not test different estimators based on different path-length heuristics or any cyclic query.
We aim to answer four specific questions: (1) Which of the 9 possible optimistic estimators we
outlined in Section~\ref{subsec:opt-space} leads to most accurate estimates on acyclic queries 
and cyclic queries that contain $\le$$3$ edges on $CEG_O$?
(2) Which of the 9 estimators lead to most accurate estimates for cyclic queries with
cycles of size $>3$ on $CEG_O$ and $CEG_{OCR}$? (3) Which of these 
CEGs lead to most accurate estimations under their best performing estimators?
(4) Given the accuracies of best performing estimators, how much room for improvement is there 
for designing more accurate techniques, e.g., heuristics, for making estimates on $CEG_O$ and $CEG_{OCR}$? 

To compare the accuracies of different estimators, 
for each query $Q$ in our workloads we make an estimate using each estimator and compute its q-error. If the true cardinality of $Q$ is $c$ and the estimate is $e$, then the q-error is $\max\{\frac{c}{e},\frac{e}{c}\} \ge 1$. For each workload, this gives us a distribution of q-errors, which we compare as follows. First, we take the logs of the q-errors so they are now $\ge 0$. If a q-error was an underestimate, we put a negative sign to it. This allows us to order the estimates from the least accurate underestimation to the least accurate overestimation. We then generate a box plot where the box represents the 25th, median, and 75th percentile cut-off marks. We also compute the mean of this distribution, excluding the top 10\% of the distribution (ignoring under/over estimations) and draw it with a red dashed line in box plots. 

\begin{figure}
\centering
\captionsetup{justification=centering}
    \includegraphics[width=0.4\columnwidth]{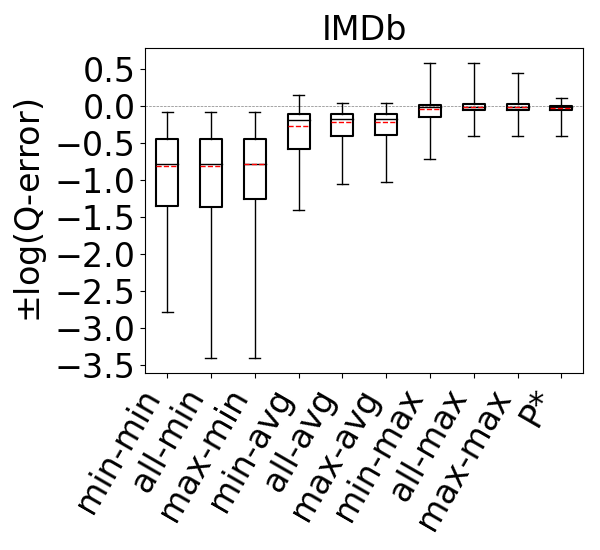}
\captionsetup{justification=centering}
    \includegraphics[width=0.4\columnwidth]{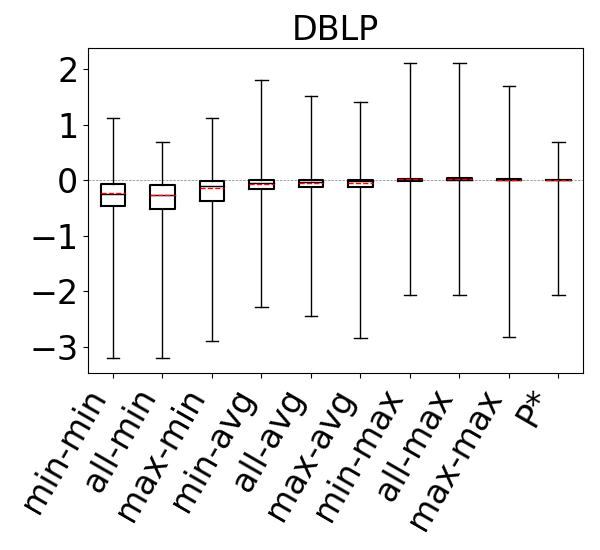}
\captionsetup{justification=centering}
    \includegraphics[width=0.4\columnwidth]{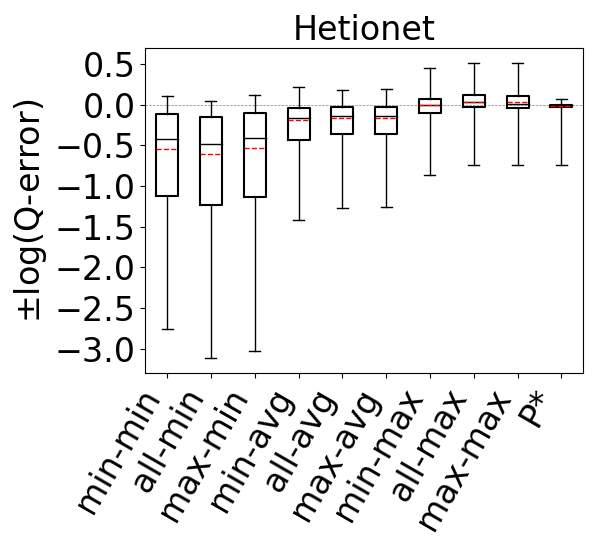}
\captionsetup{justification=centering}
    \includegraphics[width=0.4\columnwidth]{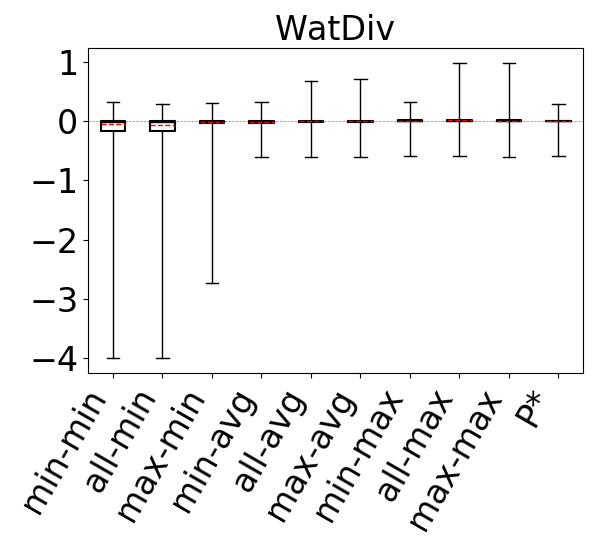}
\captionsetup{justification=centering}
    \includegraphics[width=0.4\columnwidth]{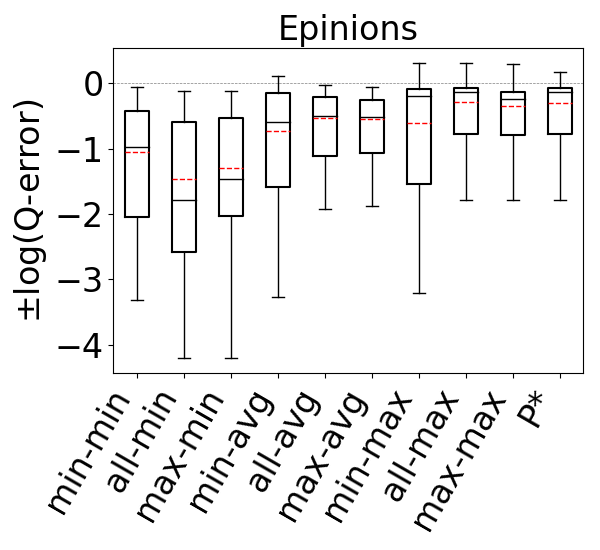}
\captionsetup{justification=centering}
    \includegraphics[width=0.4\columnwidth]{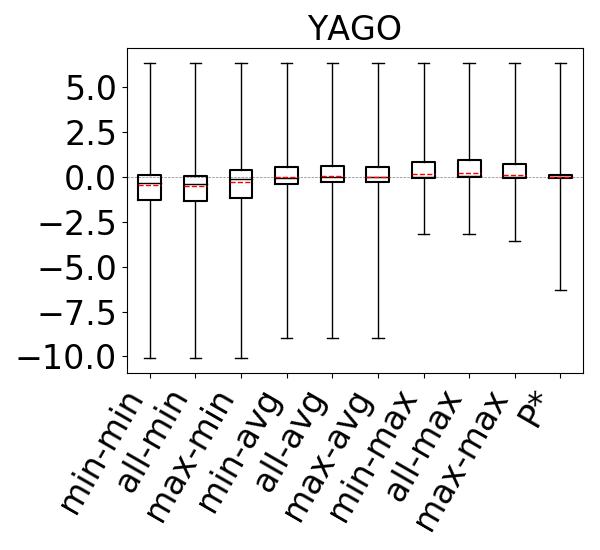}
\vspace{-10pt}
\caption{Evaluation of the optimistic estimators on $CEG_O$ on acyclic queries (JOB or Acyclic workloads).  Estimators are labeled ``X-Y'', where X describes the path length choice (one of
  \underline{max}-hop, \underline{min}-hop, or \underline{all}-hops) and Y describes the estimate aggregator (one of \underline{max}-aggr, \underline{min}-aggr, or \underline{avg}-aggr).}
\vspace{-10pt}
\label{fig:opt-space-acyclic-box}
\end{figure}

\subsubsection{Acyclic Queries and Cyclic Queries With Only Triangles}
\label{subsub:acyclic-and-triangles}
For our first question, we first compare our 9 estimators on $CEG_O$ for each acyclic query workload
  in our setup. We then compare our 9 estimators on each cyclic query workload, but only using
the queries that only contain triangles as cycles. All except one clique-6 query in the GCARE-Cyclic workload contained cycles 
with more than 3 edges, so the YAGO-GCARE-Cyclic combination is omitted.

The accuracies of the 9 estimators on acyclic workloads are shown in Figure~\ref{fig:opt-space-acyclic-box} (ignore the P* column for now).
We make several observations. First, regardless of the path-length heuristic chosen, the \texttt{max} aggregator (the last 3 box plots in the figures) makes significantly more accurate estimates (note that the y-axis on the plots are in log scale) than \texttt{avg}, which in turn is more accurate than \texttt{min}. This is true across all acyclic experiments and all datasets. For example, on IMDb and \texttt{JOB} workload, the \texttt{all-hops-min}, \texttt{all-hops-avg}, and \texttt{all-hops-max} estimators have log of mean q-errors (after removing top 10 percent outliers) of 6.5 (underestimation), 1.7 (underestimation), and 1.02 (understimation), respectively. 
Overall, we observe that using the most pessimistic of the optimistic estimates leads to significantly more accurate estimations in our evaluations than the heuristics used in prior work (up to three orders of magnitude improvement in mean accuracy). 
{\em Therefore on acyclic queries, 
when there are multiple formulas that can be used for estimating a query's cardinality, 
using the pessimistic ones is an effective technique to combat the well known underestimation problem.}

We next analyze the path-length heuristics. Observe that across all experiments, if we ignore the outliers and focus on the 25-75 percentile boxes, \texttt{max-hop} and \texttt{all-hops} do at least as well as \texttt{min-hop}. 
Further observe that on IMDb, Hetionet, and on the \texttt{Acyclic} workload on Epinions, \texttt{max-hop} and \texttt{all-hops} lead to significantly more accurate estimates. Finally, the performance of \texttt{max-hop} and \texttt{all-hops} are comparable across our experiments. We verified that this is because \texttt{all-hops} effectively picks one of the \texttt{max-hop} paths in majority of the queries in our workloads. 
Since \texttt{max-hop} enumerates strictly fewer paths than \texttt{all-hops} to make an estimate, 
we conclude that on acyclic queries, systems implementing the optimistic estimators can prefer 
the \texttt{max-hop-max} estimator.

\begin{figure*}
\centering
\captionsetup{justification=centering}
   \includegraphics[width=0.4\columnwidth]{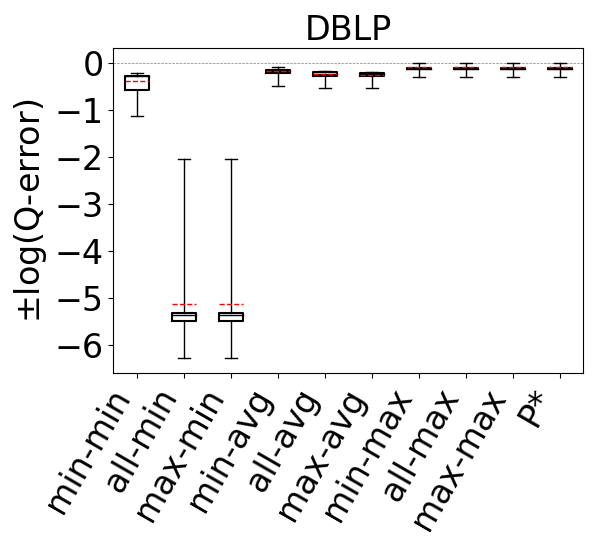}
\captionsetup{justification=centering}
    \includegraphics[width=0.4\columnwidth]{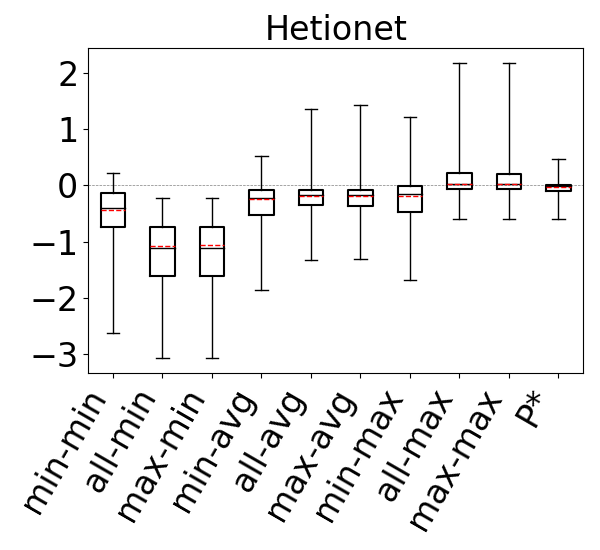}
\captionsetup{justification=centering}
    \includegraphics[width=0.4\columnwidth]{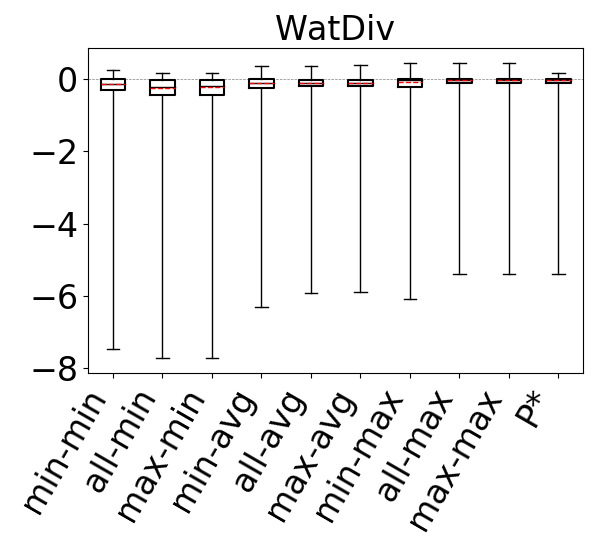}
\captionsetup{justification=centering}
    \includegraphics[width=0.4\columnwidth]{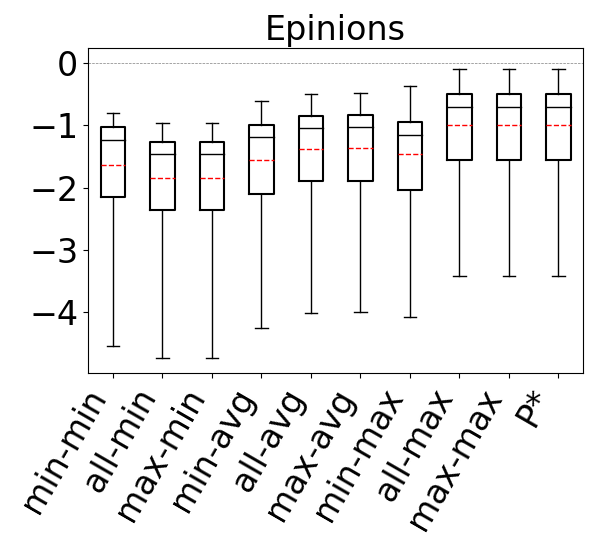}
\vspace{-10pt}
  \caption{Evaluation of the space of optimistic estimators on $CEG_O$ on \texttt{Cyclic} workload on queries with only triangles.}
\vspace{-10pt}
\label{fig:opt-space-cyclic-box-trig}
\end{figure*}

Figure~\ref{fig:opt-space-cyclic-box-trig} shows the accuracies of the 9 estimators on cyclic query workloads with only 
triangles. Our observations are similar to those for acyclic queries, and we find that the \texttt{max} aggregator yields more accurate estimates than other aggregators, irrespective of the path length. This is again because we generally observe that on most datasets
any of the 9 estimators tend to underestimate, and the \texttt{max} aggregator can combat this problem better than
\texttt{min} or \texttt{avg}. When using the max aggregator, we also observe that
the \texttt{max-hop} heuristic performs at least as well as the \texttt{min-hop} heuristic. Therefore, as we
observed
for acyclic queries, we find  \texttt{max-hop-max} estimator to be an effective way to make accurate estimations 
for cyclic queries with only triangles.

For the above experiments, we also performed a query template-specific analysis and verified that our conclusions generally 
hold for each acyclic and cyclic query template we used in our workloads in Figures~\ref{fig:opt-space-acyclic-box} and~\ref{fig:opt-space-cyclic-box-trig}. Our charts in which we evaluate the 9 estimators 
on each query template can be found in our github repo.

\begin{figure*}
\centering
\captionsetup{justification=centering}
    \includegraphics[width=0.5\columnwidth]{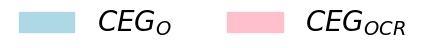}
    \\
\captionsetup{justification=centering}
    \includegraphics[width=0.6\columnwidth]{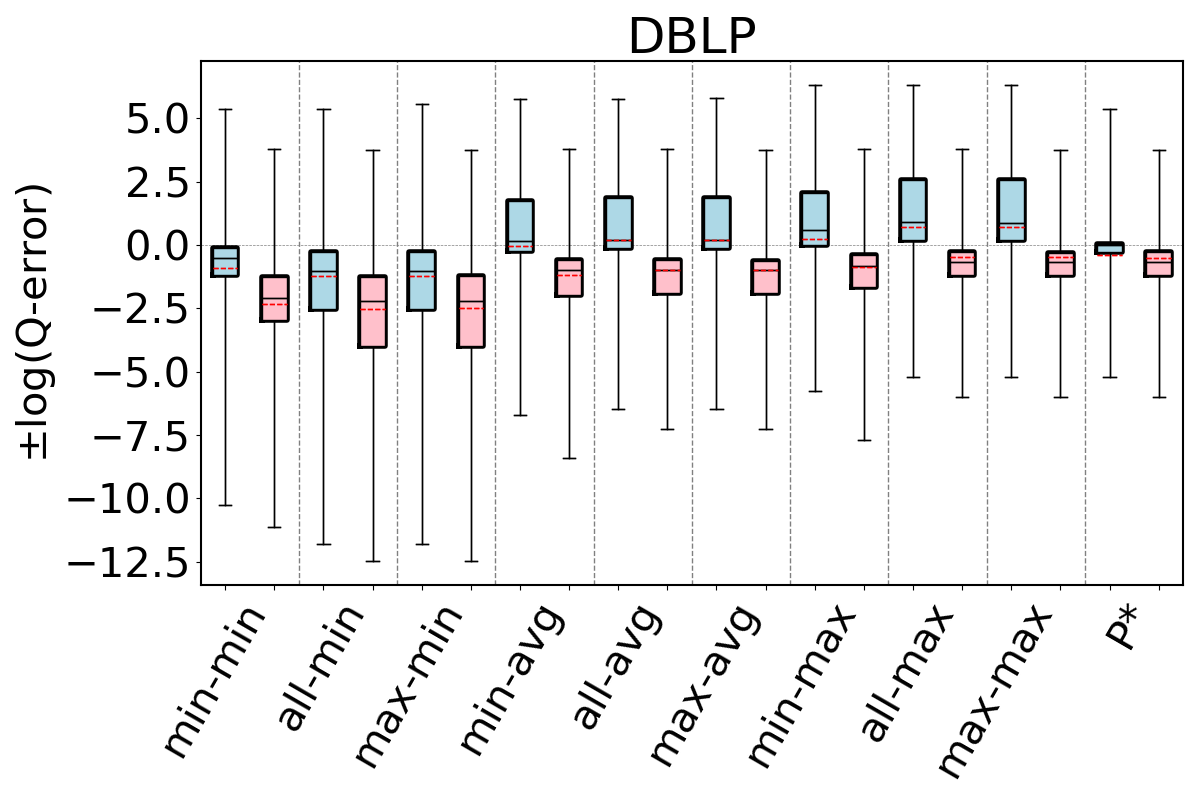}
\captionsetup{justification=centering}
    \includegraphics[width=0.6\columnwidth]{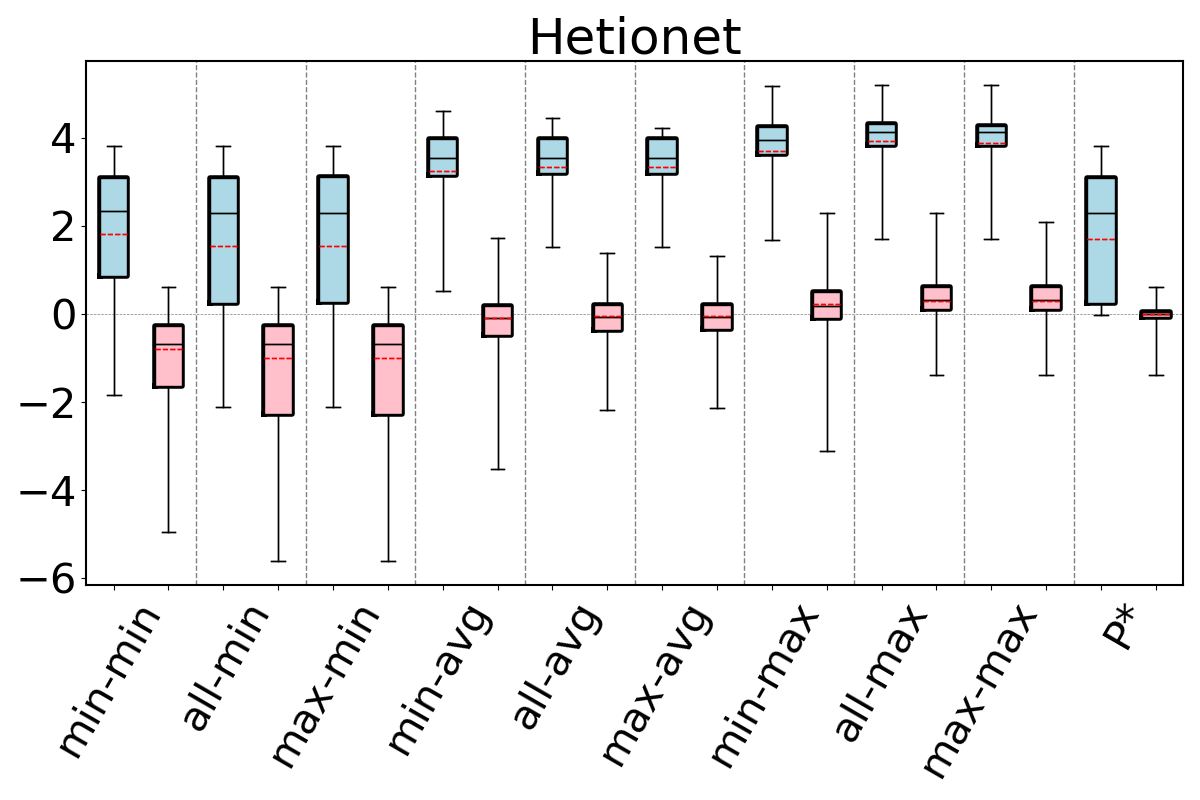}
\captionsetup{justification=centering}
    \includegraphics[width=0.6\columnwidth]{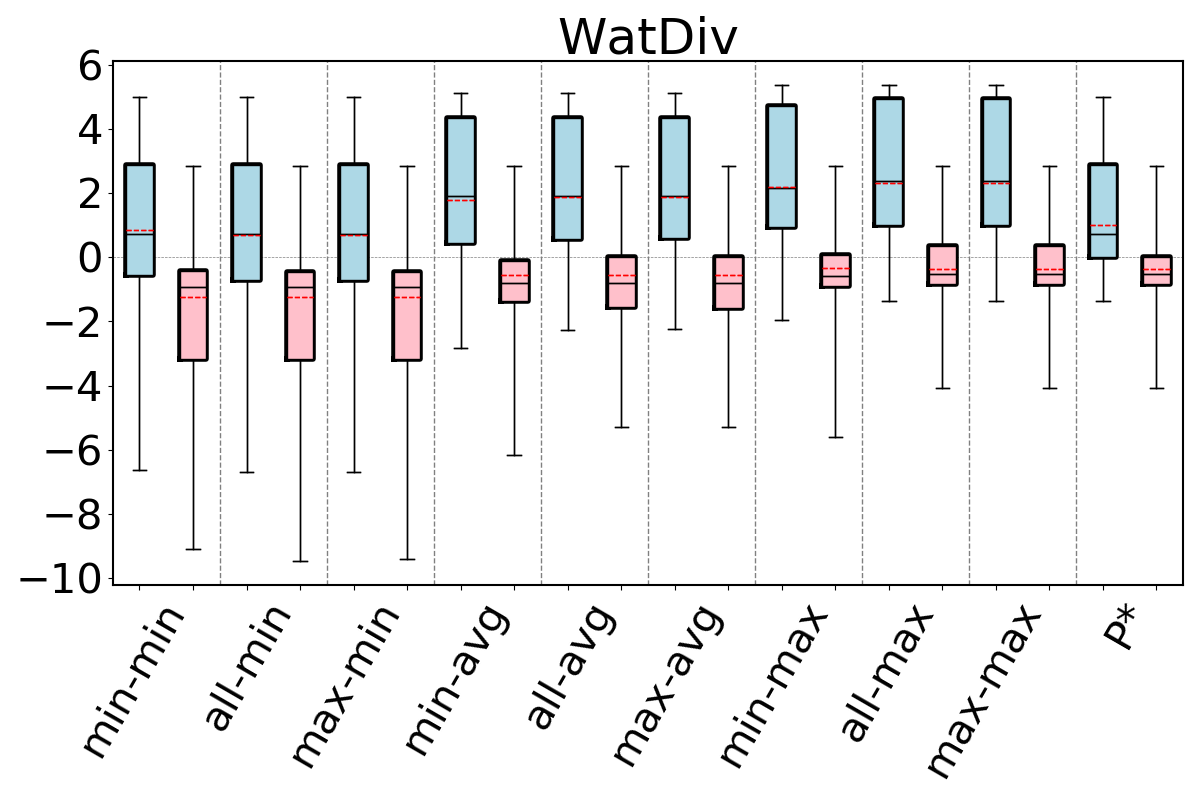}
    \\
\captionsetup{justification=centering}
    \includegraphics[width=0.6\columnwidth]{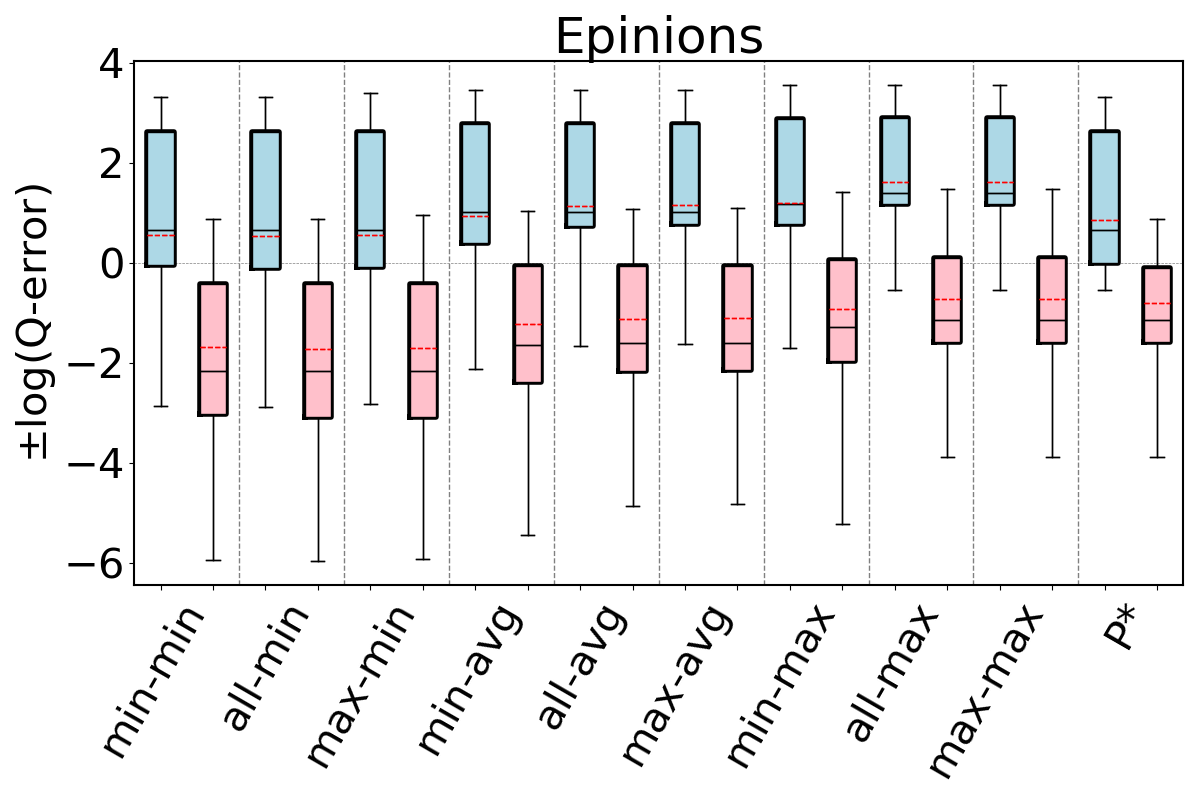}
\captionsetup{justification=centering}
    \includegraphics[width=0.6\columnwidth]{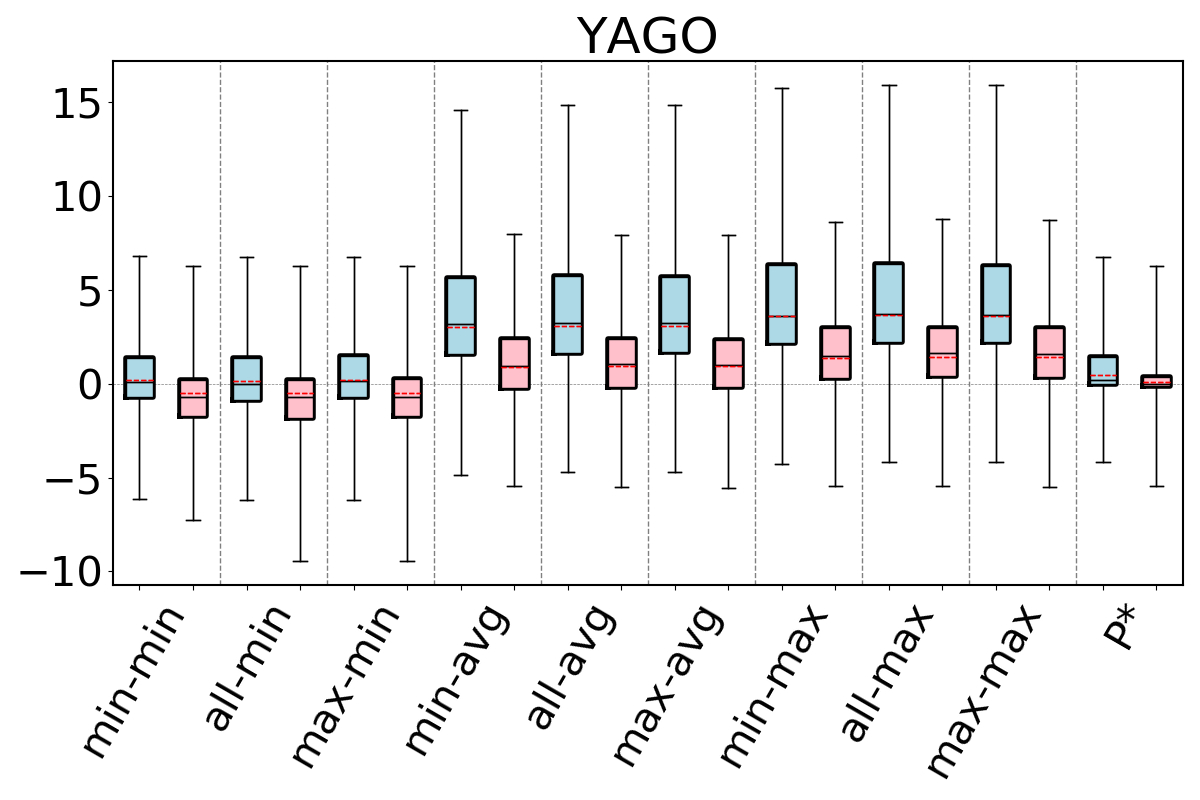}
\vspace{-10pt}
  \caption{Evaluation of the space of optimistic estimators on $CEG_O$, $CEG_{OCR}$ on cyclic queries (\texttt{Cyclic} or \texttt{G-CARE-Cyclic} workloads) on cycles with 4 or more edges.}
\vspace{-10pt}
\label{fig:opt-space-cyclic-box}
\end{figure*}

\subsubsection{Cyclic Queries With Cycles of Size $>$3}
For our second question, we compare our 9 estimators on $CEG_O$ and $CEG_{OCR}$ 
for each dataset-cyclic query workload combination in our benchmark, but only using queries that contain cycles of size $>3$. 
Recall that we now expect that estimates on $CEG_O$ can be generally pessimistic. This is because 
any formula (or bottom to top path) in $CEG_O$ breaks large cycles and estimates paths and 
on real world datasets there are often a lot more paths than cycles. 
In contrast, we expect the edge weights in $CEG_{OCR}$ to fix this pessimism, so estimates on these CEGs 
can still be optimistic.
 Figure~\ref{fig:opt-space-cyclic-box} shows our results. 
As we expected, we now see that across all of the datasets, our 9 estimators on $CEG_O$ generally overestimate.
In contrast to our observations for acyclic queries and cyclic queries with only cycles, we now see that the most 
accurate results are achieved when using the \texttt{min} aggregator (instead of \texttt{max}). 
For the \texttt{min} aggregator, any of the path-length heuristics
seem to perform reasonably well. On DBLP we see
that the \texttt{min-hop-min} heuristic leads to more accurate results, 
while on Hetionet \texttt{max-hop} and \texttt{all-hops}
heuristics, and results are comparable in other datasets. 

In contrast, on $CEG_{OCR}$, 
similar to our results from Figures~\ref{fig:opt-space-acyclic-box} and~\ref{fig:opt-space-cyclic-box-trig},
we still observe that \texttt{max} aggregator yields more accurate results,
although the \texttt{avg} aggregator also yields competitive results on Hetionet and YAGO.
This shows
that simulating the cycle closing by using cycle closing rates
avoids the pessimism of $CEG_O$ and results in optimisitic estimates.
Therefore, as before, countering this optimism using the pessimistic paths
in $CEG_{OCR}$ is an effective technique to achieve accurate estimates. In addition,
we also see that any of the path-length heuristics perform reasonably well. 

Finally to answer our third question, we compare $CEG_O$, and $CEG_{OCR}$ under their best performing heuristics.
We take as this estimator \texttt{min-hop-min} for $CEG_O$ and \texttt{max-hop-max} for $CEG_{OCR}$. We see that, except for DBLP and YAGO where the estimates are competitive,
the estimates on $CEG_{OCR}$ are more accurate. For example, on Hetionet, while the median q-error 
for \texttt{min-hop-min}  on $CEG_O$ is 213.8 (overestimation), it is only 2.0 (overestimation) for \texttt{max-hop-max} $CEG_{OCR}$.
Therefore, we observe that even using the most optimistic estimator on $CEG_O$ may not be very accurate
on cyclic queries with large cycles and modifying this CEG with cycle closing rates
can fix this shortcoming. 

As we did for the experiments presented in Section~\ref{subsub:acyclic-and-triangles}, we verified
that our conclusions do not change
for different query templates that we used in Figure~\ref{fig:opt-space-cyclic-box}. 
The figures corresponding to each query-template in Figure~\ref{fig:opt-space-cyclic-box} can be found in our github repo.

\subsubsection{$P^*$ Estimator and Room for Improvement}
We next answer the question of how much room for improvement there is for the space of optimistic estimators
we considered on $CEG_O$ and $CEG_{OCR}$.   To do so, we consider a 
thought experiment in which, for each query in our workloads, 
an oracle picks the most accurate path in our CEGs. The accuracies of this oracle-based estimator
are shown as $P^*$ bars in our bar charts in Figures~\ref{fig:opt-space-acyclic-box}-\ref{fig:opt-space-cyclic-box}.
We compare the $P^*$ bars in these figures with the \texttt{max-hop-max} estimator on $CEG_O$ 
on acyclic queries and cyclic queries with only triangles, 
and \texttt{max-hop-max} estimator on $CEG_{OCR}$ for queries with larger cycles. 
We find that on acyclic queries, shown in Figure~\ref{fig:opt-space-acyclic-box}, we generally see
little room for improvement, though there is some room in Hetionet and YAGO. 
For example, although the median q-errors of \texttt{max-hop-max} 
and $P^*$ are indistinguishable on Hetionet, the 75 percentile cutoffs for 
\texttt{max-hop-max}  and $P^*$ are 1.52 and 1.07, respectively. 
We see more room for improvement on cyclic query workloads that contain
large cycles, shown in Figures~\ref{fig:opt-space-cyclic-box}.
Although we still find that on DBLP, WatDiv and Epinions, \texttt{max-hop-max} estimator on $CEG_{OCR}$ 
is competitive with $P^*$,
there is a more visible room for improvement on Hetionet and YAGO. For example, on Hetionet,
the median q-errors of \texttt{max-hop-max} and $P^*$ are 1.48 (overestimation) and 
1.02 (underestimation), respectively. On YAGO the median q-errors are 39.8 (overestimation) and 
1.01 (overestimation), respectively. This indicates that future work on designing other 
techniques for making estimations on CEG-based
estimators can focus on workloads with large cycles on these datasets to find opportunities for improvement.

\subsection{Effects of Bound Sketch}
\label{sec:eval-partitioning}
Our next set of experiments aim to answer: {\em How much does the bound-sketch 
optimization improve the optimistic estimators' accuracy?} This question
was answered for the CBS pessimistic estimator in reference~\cite{cai:pessimistic}. We
reproduce the experiment for MOLP in our context as well.
To answer this question, 
we tested the effects of bound sketch on the \texttt{JOB} workload on IMDb 
and \texttt{Acyclic} workload
on Hetionet and Epinions. We excluded WatDiv and DBLP as the 
\texttt{max-hop-max} estimator's estimates are already very close to perfect on 
these datasets and there is no room for significant improvement
(recall Figure~\ref{fig:opt-space-acyclic-box}).
Then we applied the bound sketch optimization to both
 \texttt{max-hop-max} (on $CEG_O$) and MOLP estimators and measured the
q-errors of the estimators under partitioning budgets of 1 (no partitioning), 4, 16, 64, and 128.

Our results are shown in Figure~\ref{fig:bound-sketch}. As demonstrated in 
reference~\cite{cai:pessimistic}, our results confirm that bound sketch improves
the accuracy of MOLP. The mean accuracy of MOLP increases 
between 15\% and 89\% across all of our datasets when moving 
between 1 and 128 partitions. Similarly, we also observe significant gains on the \texttt{max-hop-max}
estimator though the results are data dependent. On Hetionet and Epinions, 
partitioning improves the mean accuracy at similar rates: by 25\% and 89\%, respectively. In contrast,
we do not observe significant gains on IMDb.
We note that the estimations of 68\% of queries in Hetionet 
and 93\% in Epinions strictly improved, so bound sketch is highly robust for optimistic estimators.
We did not observe significant improvements in IMDb dataset for \texttt{max-hop-max}, although 93\%
of their q-errors strictly improve, albeit by a small amount. We will compare optimistic and pessimistic
estimators in more detail in the next sub-section but readers can already see by inspecting the
scale in the y-axes of Figure~\ref{fig:bound-sketch} that
as observed in reference~\cite{park:gcare} the pessimistic estimators are highly inaccurate and
in our context orders of magnitude less accurate than optimistic ones.

\begin{figure}
\centering
\captionsetup{justification=centering}
    \includegraphics[width=0.4\columnwidth]{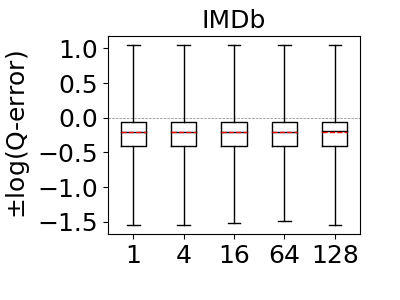}
\captionsetup{justification=centering}
    \includegraphics[width=0.4\columnwidth]{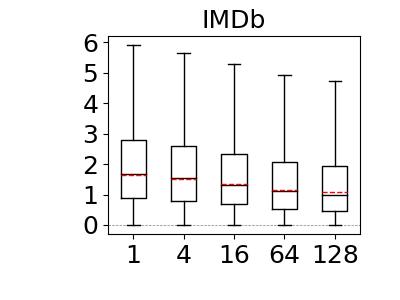}
\captionsetup{justification=centering}
    \includegraphics[width=0.4\columnwidth]{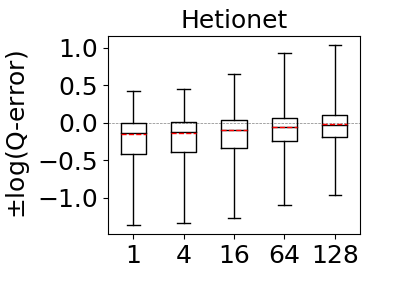}
\captionsetup{justification=centering}
    \includegraphics[width=0.4\columnwidth]{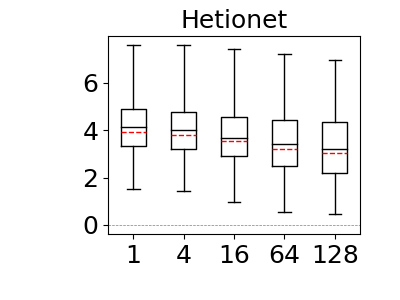}
\captionsetup{justification=centering}
    \includegraphics[width=0.4\columnwidth]{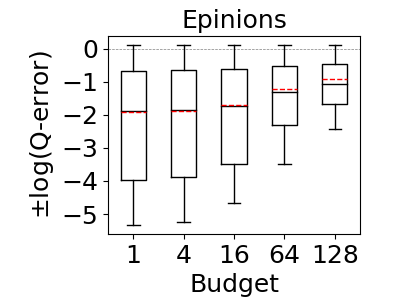}
\captionsetup{justification=centering}
      \includegraphics[width=0.4\columnwidth]{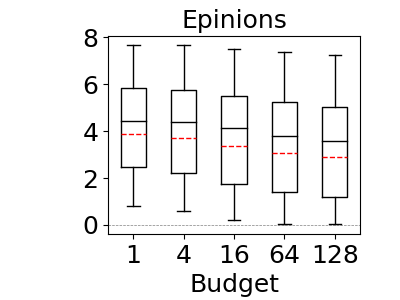}
\vspace{-10pt}
  \caption{Effects of bound sketch on \texttt{max-hop-max} (left column) and MOLP (right column) estimators.}
\vspace{-10pt}
\label{fig:bound-sketch}
\end{figure}

 \subsection{Summary-based Estimator Comparison}
\label{subsec:summary-based}
The optimistic and pessimistic estimators we consider in this paper are summary-based estimators. 
Our next set of experiments compares \texttt{max-hop-max} (on $CEG_O$) and MOLP against each other and against other
baseline summary-based estimators.
A recent work~\cite{park:gcare} has compared
MOLP against two other summary-based estimators, Characteristic Sets (CS)~\cite{neumann:characteristic-sets} and SumRDF~\cite{stefanoni:sumrdf}. 
We reproduce and extend this comparison in our setting with our suggested 
\texttt{max-hop-max} optimistic estimator. We first give an overview of CS and SumRDF.

\begin{figure*}
\centering
\captionsetup{justification=centering}
    \includegraphics[width=0.4\columnwidth]{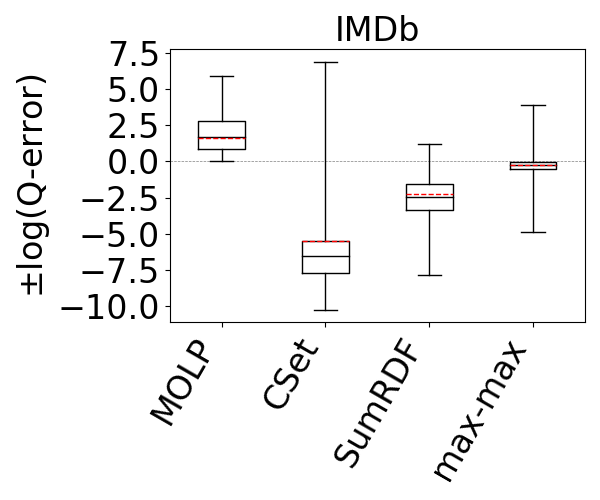}
\captionsetup{justification=centering}
    \includegraphics[width=0.4\columnwidth]{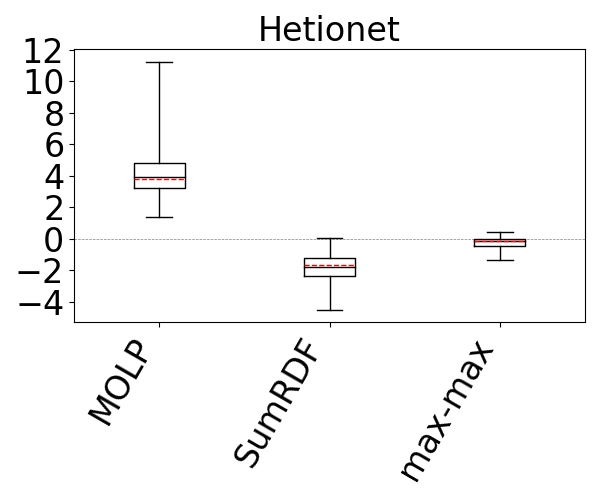}
\captionsetup{justification=centering}
    \includegraphics[width=0.4\columnwidth]{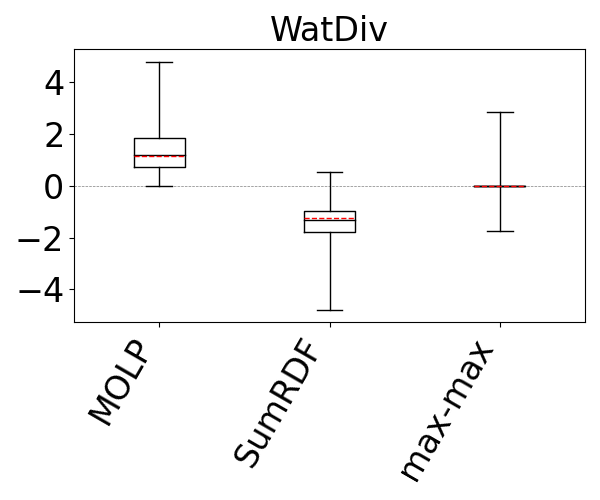}
\captionsetup{justification=centering}
    \includegraphics[width=0.4\columnwidth]{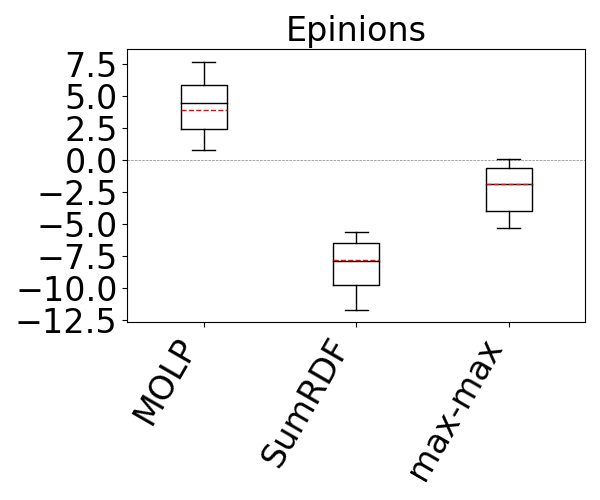}
\captionsetup{justification=centering}
    \includegraphics[width=0.4\columnwidth]{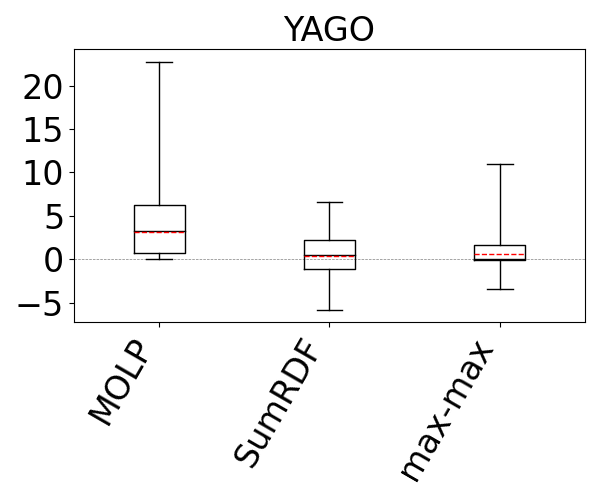}
\vspace{-10pt}
  \caption{Summary-based estimation technique comparison.}
\vspace{-10pt}
\label{fig:summary-based-no-cs}
\end{figure*}

\noindent{\em CS:} CS is a cardinality estimator that was used in the RDF-3X system~\cite{neumann:rdf3x}. 
CS is primarily designed to estimate the cardinalities of stars in an RDF graph. The statistic that it uses are 
based on the so-called {\em characteristic set} of each vertex $v$ in an RDF graph, which is the set of distinct outgoing labels $v$ has. CS keeps statistics about the vertices with the same characteristic set, such as the number of nodes that belong to a characteristic set.
Then, using these statistics, the estimator makes estimates for the number of distinct matches of stars. 
For a non-star query $Q$, $Q$ is decomposed into multiple stars $s_1, ..., s_k$, and the estimates for each $s_i$ is multiplied, which corresponds to an independence assumption. 

\noindent{SumRDF:} SumRDF builds a summary graph $S$ of an RDF graph and adopts a 
holistic approach to make an estimate. Given the summary $S$, SumRDF considers all 
possible RDF graphs $G$ that could have the same summary $S$. Then, it returns the average 
cardinality of  $Q$ across all possible instances. This is effectively another form of uniformity 
assumption: each possible world has the same probability of representing the actual graph on 
which the estimate is being made. Note that the pessimistic estimators can also be seen as 
doing something similar, except they consider the most pessimistic of the possible worlds 
and return the cardinality of $Q$ on that instance.

We measured the q-errors of \texttt{max-hop-max}, MOLP, CS, and SumRDF on the \texttt{JOB} workload on IMDb, the \texttt{Acyclic} workload on Hetionet, WatDiv, and Epinions, and the \texttt{G-CARE-Acyclic} workload on YAGO. We did not use bound sketch for MOLP and \texttt{max-hop-max}. However, we ensured that MOLP uses the cardinalities and
degree information from 2-size joins, which ensures that the statistics MOLP uses is a strict superset of the statistics
\texttt{max-hop-max} estimator uses.
 
Our results are shown in Figure~\ref{fig:summary-based-no-cs}. We omit CS from all figures except the first one, 
because it was not competitive with the rest of the estimators and even when y-axis is in logarithmic scale,
plotting CS decreases the visibility of differences 
among the rest of the estimators. SumRDF timed out on several queries on YAGO and Hetionet
and we removed those queries from \texttt{max-hop-max} and \texttt{MOLP}'s distributions as well.
We make two observations. 
First, these results confirm the results from reference~\cite{park:gcare} that 
although MOLP does not underestimate, its estimates are very loose.
Second, across all summary-based estimators, unequivocally, \texttt{max-hop-max} generates significantly 
more accurate estimations, often by several orders of magnitude in mean estimation.    
For example, on the IMDb and \texttt{JOB} workload, the mean q-errors of \texttt{max-hop-max},
SumRDF, MOLP, and CS are 1.8, 193.3, 44.6, and 333751, respectively.
We also note that both CS and SumRDF perform underestimations
in virtually all queries, whereas there are datasets, such as WatDiv and YAGO, where 
the majority of \texttt{max-hop-max}'s estimates are overestimations.

\subsection{Comparison Against WanderJoin}

Although we studied summary-based estimators in this paper, an alternative technique that has been
studied is based on sampling. Sampling-based techniques are fundamentally different and based
on using unbiased samplers of the query's output. As such, their primary advantage is that 
by enough sampling they are guaranteed to achieve estimations at any required accuracy. However, because they 
effectively perform the join on a sample of tuples, they can be slow, which is why they have seen
little adoption in practice for join estimation. 
This time-vs-accuracy tradeoff is fundamentally different in summary-based estimators, 
which can give more accurate estimates only by storing more statistics, 
(e.g., larger join-sizes in Markov tables), so with more memory and not time.
For completeness of our work, we compare \texttt{max-hop-max}
estimator with WanderJoin (WJ)~\cite{li:wanderjoin, park:gcare}, which is a sampling-based estimator 
that was identified in reference~\cite{park:gcare}
as the most efficient sampling-based technique out of a set of techniques the authors experimented with. 
Our goal is to answer: {\em What is the sampling ratio at which 
WJ's estimate outperforms \texttt{max-hop-max} (on $CEG_O$) in accuracy and 
how do the estimation speeds of WJ and  \texttt{max-hop-max}
compare at this ratio?} We first give an overview of WJ as implemented in reference~\cite{ park:gcare}.

\noindent {\em WanderJoin:} WJ is similar to the index-based sampling described in
reference~\cite{leis:sampling}. Given a query $Q$, WJ picks one of the query edges $e_q$ of $Q$ 
to start the join from and picks a sampling ratio $r \in (0, 1]$, which is the fraction of edges that can match $e_q$ 
that it will sample. For each sampled edge $e_q^*$ (with replacement), WJ 
computes the join one query-edge at-a-time by picking one of the possible
edges of a vertex that has already been matched uniformly at random. Then, if the join is successfully computed,
a correction factor depending on the degrees of the nodes that were extended is applied to get an estimate
for the number of output results that $e_q^*$ would extend to. 
Finally, the sum of the estimates for each sample
is multiplied by $1/r$ to get a final estimate.

We used the G-CARE's codebase~\cite{park:gcare}. We integrated 
the \texttt{max-hop-}\linebreak \texttt{max} estimator
into G-CARE and used the WJ code that was provided. 
We compared WJ and \texttt{max-hop-max} with
sampling ratios 0.01\%, 0.1\%, 0.25\%, 0.5\%, and 0.75\% on the \texttt{JOB} workload on IMDb, 
the \texttt{Acyclic} workload on Hetionet, WatDiv, and Epinions, and the \texttt{G-CARE-Acyclic} 
workload on YAGO. We ran both estimators five times (each run executes inside a single thread) and report
the averages of their estimation times. We also report the average q-error of WJ.
However, we can no longer present under- and over-estimations in our figures, as WJ might under and over-estimate for the same query across different runs.

The box-plot q-error distributions of  \texttt{max-hop-max} and WJ
are shown in Figure~\ref{fig:mt-vs-wj}. We identify the sampling ratios in 
which the mean accuracy of WJ is better than the mean accuracy of \texttt{max-hop-max}, 
except in DBLP and Hetionet,  where both \texttt{max-hop-}\linebreak \texttt{max} and WJ's mean estimates 
are generally close to perfect, so we look at the sampling ratio where WJ's maximum q-errors 
are smaller than \texttt{max-hop-max}. 
We find that this sampling ratio on 
IMDb is 0.1\%, on DBLP is 0.5\%, on Hetionet 
is 0.75\%, on Epinions is 0.5\%, and on YAGO is 0.75\%.
However, the estimation time of WJ is between 15x and 212x slower, so one to 
two orders of magnitude, than \texttt{max-hop-max} except on our smallest 
dataset Epinions, where the difference is 1.95x.
Observe that \texttt{max-hop-max}'s estimation times are very stable and consistently
in sub-milliseconds, between 0.18ms and 0.54ms. This is because \texttt{max-hop-max}'s 
estimation time is independent of the dataset's size. In contrast, WJ's estimation 
times get slower as the datasets get larger, because WJ performs more joins.
For example, at 0.25\% ratio, while WJ takes 0.28ms  on our smallest dataset Epinions, 
it takes 35.4ms  on DBLP. 

We emphasize that these comparisons are not perfect because it is difficult to compare 
distributions and these are two fundamentally different classes of estimators, providing systems with
different tradeoffs. However, we believe that our `competitive sampling ratio' analysis (more than runtime numbers)
is instructive
for interested readers. 

\begin{figure*}
\centering
\captionsetup{justification=centering}
    \includegraphics[width=0.4\columnwidth]{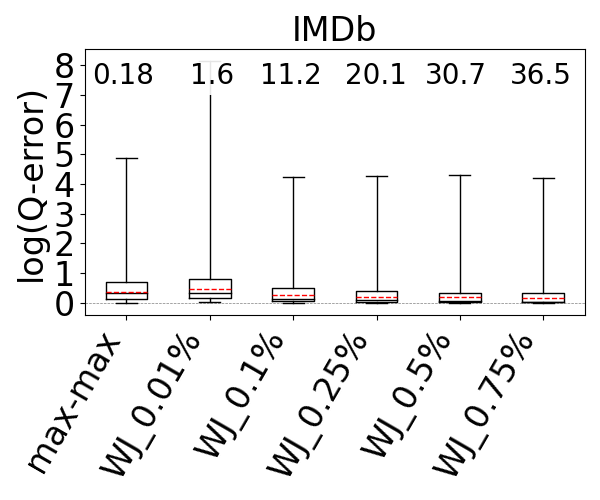}
\captionsetup{justification=centering}
    \includegraphics[width=0.4\columnwidth]{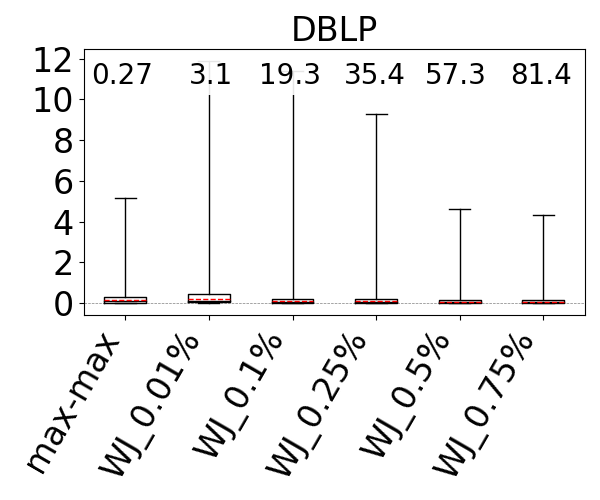}
\captionsetup{justification=centering}
    \includegraphics[width=0.4\columnwidth]{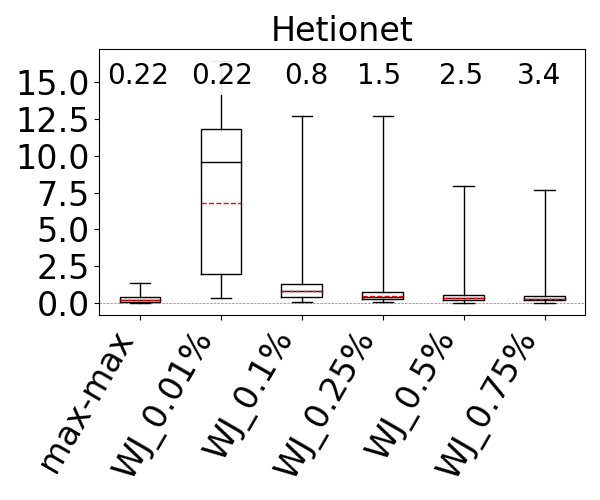}
\captionsetup{justification=centering}
    \includegraphics[width=0.4\columnwidth]{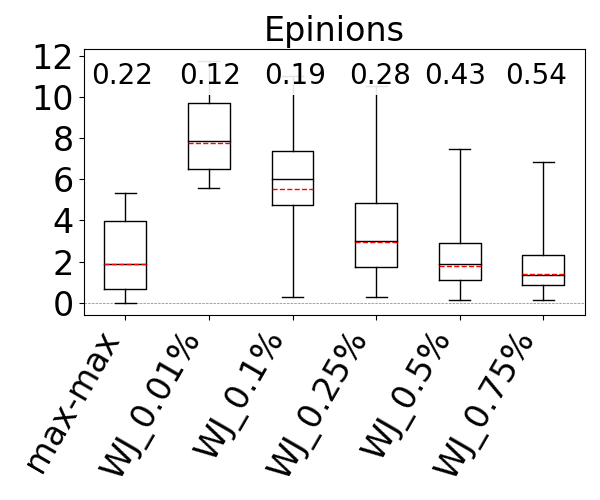}
\captionsetup{justification=centering}
    \includegraphics[width=0.4\columnwidth]{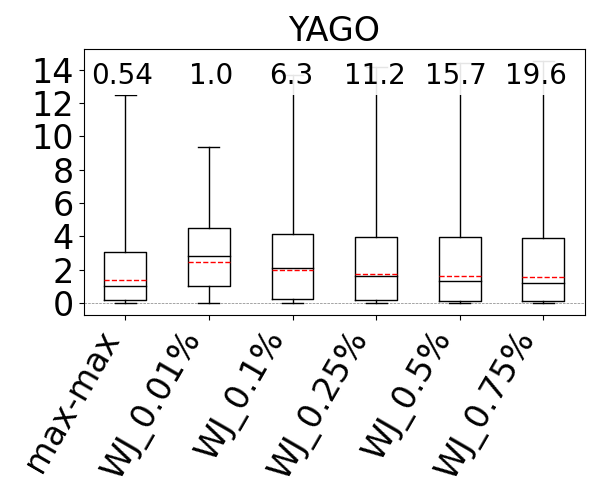}
  \caption{Comparison of \texttt{max-hop-max} and WJ, with average estimation times (in milliseconds) indicated at the top of boxes.}
\label{fig:mt-vs-wj}
\end{figure*}

\subsection{Impact on Plan Quality}
Reference~\cite{leis:optimizers} convincingly established that cardinality estimation is critical for optimizers to generate
good plans for RDBMSs as it leads to better plan generation. Several other work has verified this in different contexts, in RDBMSs~\cite{cai:pessimistic} and in RDF systems~\cite{park:gcare}. In our final set of experiments we set out to verify this in our context too by comparing 
the impact of our estimators on plan quality. We used the RDF-3X system~\cite{neumann:rdf3x} and its source code available 
here~\cite{rdf3x:sourcecode}. We issued our \texttt{Acyclic} workload as
join-only queries to RDF-3X on the DBLP and WatDiv datasets. 
We then ran the query under 10 configurations: first using RDF-3X's default estimator and then
by injecting the cardinality estimates of our 9 optimistic estimators to the system. 
The cardinalities are injected inside the system's dynamic programming-based join optimizer. We 
then filtered out the queries in which all of the 10 estimators lead to picking exactly the same plan
and there was less than 10\% runtime difference between the minimum and maximum runtimes across 
the plans generated from 10 estimators. We were left with 15 queries for DBLP and 8 queries for WatDiv.
We ran each query $5$ times and report the best query execution time. 
The open source version of RDF-3X uses a simple cardinality estimator that is not based on characteristic sets
as in reference~\cite{neumann:rdf3x} but on basic statistics about the original triple counts and some `magic' constants
to estimate cardinalities. We observed that this estimator is highly inaccurate compared to 
the 9 optimistic estimators.
For example, we analyzed the final estimates of the RDF-3X estimator on the 8 WatDiv queries and compared with 
the other estimators. We omit the full results but while the RDF-3X estimator had a median q-error of 127.4
underestimation, the worst-performing of the 9 estimators had a median q-error of only 1.947 underestimation.
So we expect RDF-3X's estimator to lead to worse plans than the other estimators.
We further expect generally that the more 
accurate of the optimistic estimators, such as the \texttt{max-hop-max} estimator, yield more efficient plans than the 
less accurate ones, such as the \texttt{min-hop-min}. 

Figure~\ref{fig:query-processing-time} shows the runtimes of the system under each configuration where the y-axis 
shows the log-scale speedup or slow down of each plan under each estimator compared to the plans under 
the RDF-3X estimator. This is why the figure does not contain a separate box plot for the RDF-3X estimator. 
First observe that the median lines of the 9 estimators are above 0, indicating the each of these estimators,
which have more accurate estimates than RDF-3X's default estimator, leads to better plan generation. 
In addition, observe that the box plot of estimators with the \texttt{max} aggregators are generally 
better than estimators that use the \texttt{min} or \texttt{avg} aggregator.
This correlates with Figure~\ref{fig:opt-space-acyclic-box}, where we showed these estimators lead to more 
accurate estimations on the same set of queries.
We then performed a detailed analysis of the \texttt{max-hop-max} and \texttt{min-hop-min} estimator as 
representative of, respectively, the most and least accurate of the 9 estimators.
We analyzed the queries in which plans under these estimators differed significantly. 
Specifically, we found 10 queries across  both datasets where the runtime differences were at least 
1.15x. Of these, only in 1 of them \texttt{min-hop-min} lead to more efficient plans and by a factor of 1.21x.
In the other 9, \texttt{max-hop-max} lead to more efficient plans, by a median of 2.05x and up to 276.3x, confirming
our expectation that better estimations generally lead to better plans.

\begin{figure}
\centering
\captionsetup{justification=centering}
    \includegraphics[width=0.45\columnwidth]{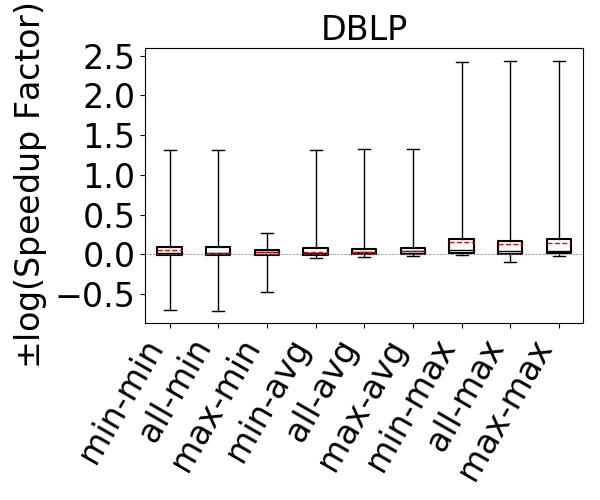}
\captionsetup{justification=centering}
    \includegraphics[width=0.45\columnwidth]{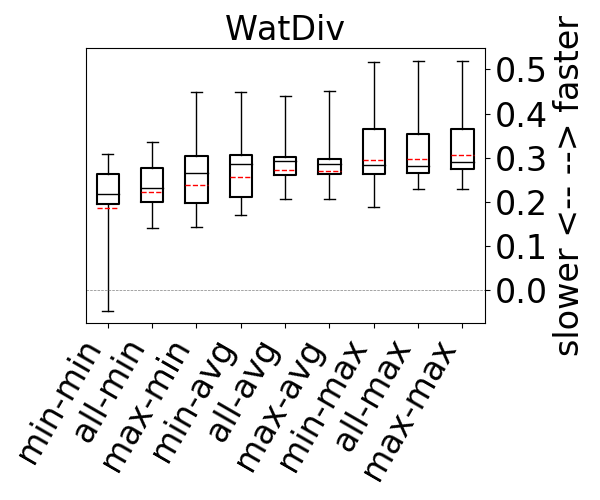}
\vspace{-10pt}
  \caption{RDF-3X runtimes on {Acyclic} workload.}
\vspace{-10pt}
\label{fig:query-processing-time}
\end{figure}

\section{Related Work}
\label{sec:rw}
There are decades of research on cardinality estimation of queries in the context of different database management systems. We cover a part of this literature focusing on work on graph-based database management systems, specifically XML and RDF and on relational systems. 
We covered Characteristic Sets~\cite{neumann:characteristic-sets}, SumRDF~\cite{stefanoni:sumrdf}, and WanderJoin~\cite{li:wanderjoin, park:gcare} in Section~\ref{sec:evaluation}.
We cover another technique, based on maximum entropy that can be used with any estimator that can return estimates for base tables or small-size joins. We do not cover work that uses machine learning techniques to estimate cardinalities and refer the reader to several recent work in this space~\cite{kipf:deep-learning, liu:neural-nets, woltmann:deep-learning} for details of these techniques.

\noindent {\bf Other Summary-based Estimators:} The estimators we studied in this paper 
fall under the category of summary-based estimators. Many relational systems, including commercial 
ones such as PostgreSQL, use summary-based estimators. Example summaries include the cardinalities 
of relations, the number of distinct values in columns, or 
histograms~\cite{aboulnaga:self-tuning-histograms, muralikrishna:equi-depth, poosala:multi-histogram}, 
wavelets~\cite{matias:wavelet}, or probabilistic and statistical models~\cite{getoor:probabilistic-models, 
sun:regression} that capture the distribution of values in columns. These 
statistics are used to estimate the selectivities of each join predicate, which are put together 
using several approaches, such as independence assumptions. In contrast,
the estimators we studied store degree statistics about base relations and small-size joins
(note that cardinalities are a form of degree statistics, e.g., $|R_i| = deg(\emptyset, \R_i)$).

Several  estimators for subgraph queries have proposed summary-based estimators 
that compute a sketch of an input graph. SumRDF, which we covered 
in Section~\ref{sec:evaluation}, falls under this category. In the context of estimating the selectivities 
of path expressions, XSeed~\cite{zhang:xseed} and XSketch~\cite{polyzotis:xsketch} build a sketch $S$ 
of the input XML Document. The sketch of the graph effectively collapses multiple nodes and edges 
into supernodes and edges with metatadata on the nodes and edges. The metadata contains statistics, 
such as the number of nodes that was collapsed into a supernode. Then given a query $Q$, $Q$ is matched 
on $S$ and using the metadata an estimate is made. Because these techniques do not decompose 
a query into smaller sub-queries, the question of which decomposition to use does 
not arise for these estimators. 

Several work have used data structures that are adaptations of histograms from relational systems to store selectivities of path or tree queries in XML documents. Examples include, {\em positional histograms}~\cite{wu:positional-histogram} and {\em Bloom histogram}~\cite{wang:bloom-histogram}. These techniques do not consecutively make estimates for larger paths and have not been adopted to general subgraph queries. For example, instead of storing small-size paths in a data structure as in Markov tables, Bloom histograms store all paths but hashed in a bloom filter. Other work used similar summaries of XML documents (or its precursor the {\em object exchange model}~\cite{papakonstantinou:oem} databases)  for purposes other than cardinality estimation. For example, {\em TreeSketch}~\cite{polyzotis:treesketch} 
produces a summary of large XML documents to provide approximate answers to queries.

\noindent {\bf Sampling-based Estimators:} Another class of important estimators are based on sampling tuples~\cite{haas:sampling, leis:sampling, li:wanderjoin, vengerov:correlated-sampling, wu:sampling}. These estimators either sample input records from base tables offline or during query optimization, and they evaluate queries on these samples to make estimates. Research has focused on different ways samples can be generated, such as independent or correlated sampling, or sampling through existing indexes. Wander Join, which we covered in 
Section~\ref{sec:evaluation} falls under this category. 
As we discussed,  by increasing the sizes of the samples 
these estimators can be arbitrarily accurate but  they are in general slower than 
summary-based ones because they
actually perform the join on a sample of tuples. We are aware of systems~\cite{leis:job} that use sampling-based estimators 
 to estimate the selectivities of predicates in base tables but not on multiway joins. 
Finally, sampling-based estimators have also been 
used to estimate frequencies of subgraphs relative to each other to 
discover {\em motifs}, i.e. infrequently appearing subgrahs,~\cite{kashtan:motif}.

\noindent {\bf The Maximum Entropy Estimator:} Markl et al.~\cite{markl:entropy} has proposed another approach to make estimates for conjunctive predicates, say $p_1 \land ... \land p_k$ given a set of $\ell$ selectivity estimates $s_{i_{11}, ..., i_{1k}}$, $s_{i_{\ell 1}, ..., i_{\ell k}}$, where $s_{i_{j1}, ..., i_{jk}}$ is the selectivity estimate for predicate $p_{i_{j1}} \land ... \land p_{i_{jk}}$. Markl et al.'s maximum entropy approach takes these known selectivities and uses a constraint optimization problem to compute the distribution that maximizes the entropy of the joint distribution of the $2^k$ possible predicate space. Reference~\cite{markl:entropy} has only evaluated the accuracy of this approach for estimating conjunctive predicates on base tables and not on joins, but they have briefly described how the same approach can be used to estimate the cardinalities of join queries. Multiway join queries can be modeled as estimating the selectivity of the full join predicate that contains the equality constraint of all attributes with the same name. The statistics that we considered in this paper can be translated to selectivities of each predicate. For example the size of $|(a_1)\xrightarrow{A}(a_2)\xrightarrow{B}(a_3)|$ can be modeled as $s_i = \frac{|(a_1)\xrightarrow{A}(a_2)\xrightarrow{B}(a_3)|}{|A||B|}$, as the join of $A$ and $B$ is by definition applying the predicate $A.src = B.dst$ predicate on the Cartesian product of relations $A$ and $B$. This way, one can construct another optimistic estimator using the same statistics. We have not investigated the accuracy of this approach within the scope of this work and leave this to future work.

\section{Conclusions and Future Work}
\label{sec:fw}

We focused on how to make accurate estimations using the optimistic estimators
using a new framework, in which we model these estimators as paths in a weighted CEG we called $CEG_O$,
which uses as edge weights average degree statistics.
We addressed two shortcomings of optimistic estimators from prior work. First, we addressed the question of 
which formulas, i.e., CEG paths, to use when there are multiple formulas to estimate a given query.
We outlined and empirically evaluated a space of heuristics and
showed that the heuristic that leads generally to most accurate estimates depends on 
the structure of the query. For example, we showed that for acyclic queries and cyclic queries wth
small-size cycles, using the maximum weight path is an effective way to make accurate estimates.
Second, we addressed how to make accurate
estimates for queries with large cycles by proposing a new CEG that we call $CEG_{OCR}$.
We then showed that surprisingly the recent pessimistic estimators can also be modeled 
as picking a path, this time the minimum weight path in another CEG, in which the edge
weights are maximum degree statistics. Aside from connecting two disparate 
classes of estimators, this observation allowed us 
to apply the bound sketch optimization for pessimistic estimators to optimistic ones. We
also showed that CEGs are useful mathematical tools to prove several properties of 
pessimistic estimators, e.g., that the CBS estimator is equivalent to MOLP estimator
under acyclic queries over binary relations.

We believe the CEG framework can be the foundation for further research to define and evaluate
other novel estimators. Within the scope of this work, we considered only
three specific CEGs, with a fourth one that is used for theoretical purposes in
\iflongversion
Appendix~\ref{app:dbp}.
\else
the long version of our paper~\cite{chen:ceg-tr}.
\fi
However, many other CEGs can be defined to develop new estimators using different statistics
and potentially different techniques to pick paths in CEGs as estimates. 
For example, one can use variance, standard deviation, or entropies of the distributions 
of small-size joins as edge weights in a CEG, possibly along with degree statistics, 
and pick the minimum-weight, e.g.,, ``lowest 
entropy'', paths, assuming that degrees are more regular in lower entropy edges. 
An important research direction is to systematically study a class of CEG 
instances that use different combination of statistics as edge weights, as well as heuristics on 
these CEGs for picking paths, to understand which statistics lead to more accurate results
in practice.
\balance

\section{Acknowledgment}
This work was supported by NSERC and by a grant from Waterloo-Huawei Joint Innovation Laboratory.

\iflongversion
\else
\pagebreak
\fi
\bibliographystyle{ACM-Reference-Format}
\bibliography{references}


\begin{thebibliography}{00}


\ifx \showCODEN    \undefined \def \showCODEN     #1{\unskip}     \fi
\ifx \showDOI      \undefined \def \showDOI       #1{#1}\fi
\ifx \showISBNx    \undefined \def \showISBNx     #1{\unskip}     \fi
\ifx \showISBNxiii \undefined \def \showISBNxiii  #1{\unskip}     \fi
\ifx \showISSN     \undefined \def \showISSN      #1{\unskip}     \fi
\ifx \showLCCN     \undefined \def \showLCCN      #1{\unskip}     \fi
\ifx \shownote     \undefined \def \shownote      #1{#1}          \fi
\ifx \showarticletitle \undefined \def \showarticletitle #1{#1}   \fi
\ifx \showURL      \undefined \def \showURL       {\relax}        \fi
\providecommand\bibfield[2]{#2}
\providecommand\bibinfo[2]{#2}
\providecommand\natexlab[1]{#1}
\providecommand\showeprint[2][]{arXiv:#2}

\bibitem[\protect\citeauthoryear{Abo~Khamis, Ngo, and Suciu}{Abo~Khamis
  et~al\mbox{.}}{2016}]%
        {khamis:cllp}
\bibfield{author}{\bibinfo{person}{Mahmoud Abo~Khamis},
  \bibinfo{person}{Hung~Q. Ngo}, {and} \bibinfo{person}{Dan Suciu}.}
  \bibinfo{year}{2016}\natexlab{}.
\newblock \showarticletitle{{Computing Join Queries with Functional
  Dependencies}}. In \bibinfo{booktitle}{{\em PODS}}.
\newblock


\bibitem[\protect\citeauthoryear{Aboulnaga, Alameldeen, and Naughton}{Aboulnaga
  et~al\mbox{.}}{2001}]%
        {aboulnaga:markov}
\bibfield{author}{\bibinfo{person}{Ashraf Aboulnaga}, \bibinfo{person}{Alaa~R.
  Alameldeen}, {and} \bibinfo{person}{Jeffrey~F. Naughton}.}
  \bibinfo{year}{2001}\natexlab{}.
\newblock \showarticletitle{{Estimating the Selectivity of XML Path Expressions
  for Internet Scale Applications}}. In \bibinfo{booktitle}{{\em VLDB}}.
\newblock


\bibitem[\protect\citeauthoryear{Aboulnaga and Chaudhuri}{Aboulnaga and
  Chaudhuri}{1999}]%
        {aboulnaga:self-tuning-histograms}
\bibfield{author}{\bibinfo{person}{Ashraf Aboulnaga} {and}
  \bibinfo{person}{Surajit Chaudhuri}.} \bibinfo{year}{1999}\natexlab{}.
\newblock \showarticletitle{{Self-Tuning Histograms: Building Histograms
  Without Looking at Data}}. In \bibinfo{booktitle}{{\em SIGMOD}}.
\newblock


\bibitem[\protect\citeauthoryear{Atserias, Grohe, and Marx}{Atserias
  et~al\mbox{.}}{2013}]%
        {atserias:agm}
\bibfield{author}{\bibinfo{person}{A. Atserias}, \bibinfo{person}{M. Grohe},
  {and} \bibinfo{person}{D. Marx}.} \bibinfo{year}{2013}\natexlab{}.
\newblock \showarticletitle{{Size Bounds and Query Plans for Relational
  Joins}}.
\newblock \bibinfo{journal}{{\em SICOMP\/}} \bibinfo{volume}{42},
  \bibinfo{number}{4} (\bibinfo{year}{2013}).
\newblock


\bibitem[\protect\citeauthoryear{Cai, Balazinska, and Suciu}{Cai
  et~al\mbox{.}}{2019}]%
        {cai:pessimistic}
\bibfield{author}{\bibinfo{person}{Walter Cai}, \bibinfo{person}{Magdalena
  Balazinska}, {and} \bibinfo{person}{Dan Suciu}.}
  \bibinfo{year}{2019}\natexlab{}.
\newblock \showarticletitle{{Pessimistic Cardinality Estimation: Tighter Upper
  Bounds for Intermediate Join Cardinalities}}. In \bibinfo{booktitle}{{\em
  SIGMOD}}.
\newblock


\bibitem[\protect\citeauthoryear{Getoor, Taskar, and Koller}{Getoor
  et~al\mbox{.}}{2001}]%
        {getoor:probabilistic-models}
\bibfield{author}{\bibinfo{person}{Lise Getoor}, \bibinfo{person}{Benjamin
  Taskar}, {and} \bibinfo{person}{Daphne Koller}.}
  \bibinfo{year}{2001}\natexlab{}.
\newblock \showarticletitle{{Selectivity Estimation Using Probabilistic
  Models}}. In \bibinfo{booktitle}{{\em SIGMOD}}.
\newblock


\bibitem[\protect\citeauthoryear{Gottlob, Lee, Valiant, and Valiant}{Gottlob
  et~al\mbox{.}}{2012}]%
        {gottlob:glvv}
\bibfield{author}{\bibinfo{person}{Georg Gottlob},
  \bibinfo{person}{Stephanie~Tien Lee}, \bibinfo{person}{Gregory Valiant},
  {and} \bibinfo{person}{Paul Valiant}.} \bibinfo{year}{2012}\natexlab{}.
\newblock \showarticletitle{{Size and Treewidth Bounds for Conjunctive
  Queries}}.
\newblock \bibinfo{journal}{{\em JACM\/}} \bibinfo{volume}{59},
  \bibinfo{number}{3} (\bibinfo{year}{2012}).
\newblock


\bibitem[\protect\citeauthoryear{{Haas, Peter J. and Naughton, Jeffrey F. and
  Seshadri, S. and Swami, Arun N.}}{{Haas, Peter J. and Naughton, Jeffrey F.
  and Seshadri, S. and Swami, Arun N.}}{1996}]%
        {haas:sampling}
\bibfield{author}{\bibinfo{person}{{Haas, Peter J. and Naughton, Jeffrey F. and
  Seshadri, S. and Swami, Arun N.}}} \bibinfo{year}{1996}\natexlab{}.
\newblock \showarticletitle{{Selectivity and Cost Estimation for Joins Based on
  Random Sampling}}.
\newblock \bibinfo{journal}{{\em JCSS\/}} \bibinfo{volume}{52},
  \bibinfo{number}{3} (\bibinfo{year}{1996}).
\newblock


\bibitem[\protect\citeauthoryear{Joglekar and R{\'{e}}}{Joglekar and
  R{\'{e}}}{2018}]%
        {joglekar:degree}
\bibfield{author}{\bibinfo{person}{Manas Joglekar} {and}
  \bibinfo{person}{Christopher R{\'{e}}}.} \bibinfo{year}{2018}\natexlab{}.
\newblock \showarticletitle{{It's All a Matter of Degree - Using Degree
  Information to Optimize Multiway Joins}}.
\newblock \bibinfo{journal}{{\em TOCS\/}} \bibinfo{volume}{62},
  \bibinfo{number}{4} (\bibinfo{year}{2018}).
\newblock


\bibitem[\protect\citeauthoryear{Kashtan, Itzkovitz, Milo, and Alon}{Kashtan
  et~al\mbox{.}}{2004}]%
        {kashtan:motif}
\bibfield{author}{\bibinfo{person}{Nadav Kashtan}, \bibinfo{person}{Shalev
  Itzkovitz}, \bibinfo{person}{Ron Milo}, {and} \bibinfo{person}{Uri Alon}.}
  \bibinfo{year}{2004}\natexlab{}.
\newblock \showarticletitle{{Efficient Sampling Algorithm for Estimating
  Subgraph Concentrations and Detecting Network Motifs}}.
\newblock \bibinfo{journal}{{\em Bioinformatics\/}} \bibinfo{volume}{20},
  \bibinfo{number}{11} (\bibinfo{year}{2004}).
\newblock


\bibitem[\protect\citeauthoryear{Kipf, Kipf, Radke, Leis, Boncz, and
  Kemper}{Kipf et~al\mbox{.}}{2019}]%
        {kipf:deep-learning}
\bibfield{author}{\bibinfo{person}{Andreas Kipf}, \bibinfo{person}{Thomas
  Kipf}, \bibinfo{person}{Bernhard Radke}, \bibinfo{person}{Viktor Leis},
  \bibinfo{person}{Peter~A. Boncz}, {and} \bibinfo{person}{Alfons Kemper}.}
  \bibinfo{year}{2019}\natexlab{}.
\newblock \showarticletitle{{Learned Cardinalities: Estimating Correlated Joins
  with Deep Learning}}. In \bibinfo{booktitle}{{\em CIDR}}.
\newblock


\bibitem[\protect\citeauthoryear{Leis, Gubichev, Mirchev, Boncz, Kemper, and
  Neumann}{Leis et~al\mbox{.}}{2015}]%
        {leis:optimizers}
\bibfield{author}{\bibinfo{person}{Viktor Leis}, \bibinfo{person}{Andrey
  Gubichev}, \bibinfo{person}{Atanas Mirchev}, \bibinfo{person}{Peter Boncz},
  \bibinfo{person}{Alfons Kemper}, {and} \bibinfo{person}{Thomas Neumann}.}
  \bibinfo{year}{2015}\natexlab{}.
\newblock \showarticletitle{{How Good Are Query Optimizers, Really?}}. In
  \bibinfo{booktitle}{{\em VLDB}}.
\newblock


\bibitem[\protect\citeauthoryear{Leis, Radke, Gubichev, Kemper, and
  Neumann}{Leis et~al\mbox{.}}{2017}]%
        {leis:sampling}
\bibfield{author}{\bibinfo{person}{Viktor Leis}, \bibinfo{person}{Bernhard
  Radke}, \bibinfo{person}{Andrey Gubichev}, \bibinfo{person}{Alfons Kemper},
  {and} \bibinfo{person}{Thomas Neumann}.} \bibinfo{year}{2017}\natexlab{}.
\newblock \showarticletitle{{Cardinality Estimation Done Right: Index-Based
  Join Sampling}}. In \bibinfo{booktitle}{{\em CIDR}}.
\newblock


\bibitem[\protect\citeauthoryear{Leis, Radke, Gubichev, Mirchev, Boncz, Kemper,
  and Neumann}{Leis et~al\mbox{.}}{2018}]%
        {leis:job}
\bibfield{author}{\bibinfo{person}{Viktor Leis}, \bibinfo{person}{Bernhard
  Radke}, \bibinfo{person}{Andrey Gubichev}, \bibinfo{person}{Atanas Mirchev},
  \bibinfo{person}{Peter Boncz}, \bibinfo{person}{Alfons Kemper}, {and}
  \bibinfo{person}{Thomas Neumann}.} \bibinfo{year}{2018}\natexlab{}.
\newblock \showarticletitle{{Query Optimization Through the Looking Glass, and
  What We Found Running the Join Order Benchmark}}.
\newblock \bibinfo{journal}{{\em VLDBJ\/}} \bibinfo{volume}{27},
  \bibinfo{number}{5} (\bibinfo{year}{2018}).
\newblock


\bibitem[\protect\citeauthoryear{Li, Wu, Yi, and Zhao}{Li
  et~al\mbox{.}}{2016}]%
        {li:wanderjoin}
\bibfield{author}{\bibinfo{person}{Feifei Li}, \bibinfo{person}{Bin Wu},
  \bibinfo{person}{Ke Yi}, {and} \bibinfo{person}{Zhuoyue Zhao}.}
  \bibinfo{year}{2016}\natexlab{}.
\newblock \showarticletitle{{Wander Join: Online Aggregation via Random
  Walks}}. In \bibinfo{booktitle}{{\em SIGMOD}}.
\newblock


\bibitem[\protect\citeauthoryear{Liu, Xu, Yu, Corvinelli, and Zuzarte}{Liu
  et~al\mbox{.}}{2015}]%
        {liu:neural-nets}
\bibfield{author}{\bibinfo{person}{Henry Liu}, \bibinfo{person}{Mingbin Xu},
  \bibinfo{person}{Ziting Yu}, \bibinfo{person}{Vincent Corvinelli}, {and}
  \bibinfo{person}{Calisto Zuzarte}.} \bibinfo{year}{2015}\natexlab{}.
\newblock \showarticletitle{{Cardinality Estimation Using Neural Networks}}. In
  \bibinfo{booktitle}{{\em CASCON}}.
\newblock


\bibitem[\protect\citeauthoryear{Maduko, Anyanwu, Sheth, and
  Schliekelman}{Maduko et~al\mbox{.}}{2008}]%
        {maduko:md-tree}
\bibfield{author}{\bibinfo{person}{Angela Maduko}, \bibinfo{person}{Kemafor
  Anyanwu}, \bibinfo{person}{Amit Sheth}, {and} \bibinfo{person}{Paul
  Schliekelman}.} \bibinfo{year}{2008}\natexlab{}.
\newblock \showarticletitle{{Graph Summaries for Subgraph Frequency
  Estimation}}. In \bibinfo{booktitle}{{\em ESWC}}.
\newblock


\bibitem[\protect\citeauthoryear{Markl, Haas, Kutsch, Megiddo, Srivastava, and
  Tam}{Markl et~al\mbox{.}}{2007}]%
        {markl:entropy}
\bibfield{author}{\bibinfo{person}{V. Markl}, \bibinfo{person}{P.~J. Haas},
  \bibinfo{person}{M. Kutsch}, \bibinfo{person}{N. Megiddo},
  \bibinfo{person}{U. Srivastava}, {and} \bibinfo{person}{T.~M. Tam}.}
  \bibinfo{year}{2007}\natexlab{}.
\newblock \showarticletitle{{Consistent Selectivity Estimation via Maximum
  Entropy}}.
\newblock \bibinfo{journal}{{\em VLDBJ\/}}  \bibinfo{volume}{16}
  (\bibinfo{year}{2007}).
\newblock


\bibitem[\protect\citeauthoryear{Matias, Vitter, and Wang}{Matias
  et~al\mbox{.}}{1998}]%
        {matias:wavelet}
\bibfield{author}{\bibinfo{person}{Yossi Matias},
  \bibinfo{person}{Jeffrey~Scott Vitter}, {and} \bibinfo{person}{Min Wang}.}
  \bibinfo{year}{1998}\natexlab{}.
\newblock \showarticletitle{{Wavelet-Based Histograms for Selectivity
  Estimation}}. In \bibinfo{booktitle}{{\em SIGMOD}}.
\newblock


\bibitem[\protect\citeauthoryear{Mhedhbi and Salihoglu}{Mhedhbi and
  Salihoglu}{2019}]%
        {mhedhbi:optimizer}
\bibfield{author}{\bibinfo{person}{Amine Mhedhbi} {and} \bibinfo{person}{Semih
  Salihoglu}.} \bibinfo{year}{2019}\natexlab{}.
\newblock \showarticletitle{{Optimizing Subgraph Queries by Combining Binary
  and Worst-Case Optimal Joins}}.
\newblock \bibinfo{journal}{{\em {PVLDB}\/}} \bibinfo{volume}{12},
  \bibinfo{number}{11} (\bibinfo{year}{2019}).
\newblock


\bibitem[\protect\citeauthoryear{Muralikrishna and DeWitt}{Muralikrishna and
  DeWitt}{1988}]%
        {muralikrishna:equi-depth}
\bibfield{author}{\bibinfo{person}{M. Muralikrishna} {and}
  \bibinfo{person}{David~J. DeWitt}.} \bibinfo{year}{1988}\natexlab{}.
\newblock \showarticletitle{{Equi-Depth Histograms for Estimating Selectivity
  Factors for Multi-Dimensional Queries}}. In \bibinfo{booktitle}{{\em
  SIGMOD}}.
\newblock


\bibitem[\protect\citeauthoryear{Neumann and Moerkotte}{Neumann and
  Moerkotte}{2011}]%
        {neumann:characteristic-sets}
\bibfield{author}{\bibinfo{person}{Thomas Neumann} {and} \bibinfo{person}{Guido
  Moerkotte}.} \bibinfo{year}{2011}\natexlab{}.
\newblock \showarticletitle{{Characteristic Sets: Accurate Cardinality
  Estimation for RDF Queries with Multiple Joins}}. In \bibinfo{booktitle}{{\em
  ICDE}}.
\newblock


\bibitem[\protect\citeauthoryear{Neumann and Weikum}{Neumann and
  Weikum}{2008}]%
        {neumann:rdf3x}
\bibfield{author}{\bibinfo{person}{Thomas Neumann} {and}
  \bibinfo{person}{Gerhard Weikum}.} \bibinfo{year}{2008}\natexlab{}.
\newblock \showarticletitle{{RDF-3X: A RISC-Style Engine for RDF}}. In
  \bibinfo{booktitle}{{\em VLDB}}.
\newblock


\bibitem[\protect\citeauthoryear{Papakonstantinou, Garcia-Molina, and
  Widom}{Papakonstantinou et~al\mbox{.}}{1995}]%
        {papakonstantinou:oem}
\bibfield{author}{\bibinfo{person}{Yannis Papakonstantinou},
  \bibinfo{person}{Hector Garcia-Molina}, {and} \bibinfo{person}{Jennifer
  Widom}.} \bibinfo{year}{1995}\natexlab{}.
\newblock \showarticletitle{{Object Exchange Across Heterogeneous Information
  Sources}}. In \bibinfo{booktitle}{{\em ICDE}}.
\newblock


\bibitem[\protect\citeauthoryear{Park, Ko, Bhowmick, Kim, Hong, and Han}{Park
  et~al\mbox{.}}{2020}]%
        {park:gcare}
\bibfield{author}{\bibinfo{person}{Yeonsu Park}, \bibinfo{person}{Seongyun Ko},
  \bibinfo{person}{Sourav~S. Bhowmick}, \bibinfo{person}{Kyoungmin Kim},
  \bibinfo{person}{Kijae Hong}, {and} \bibinfo{person}{Wook-Shin Han}.}
  \bibinfo{year}{2020}\natexlab{}.
\newblock \showarticletitle{{G-CARE: A Framework for Performance Benchmarking
  of Cardinality Estimation Techniques for Subgraph Matching}}. In
  \bibinfo{booktitle}{{\em SIGMOD}}.
\newblock


\bibitem[\protect\citeauthoryear{Polyzotis and Garofalakis}{Polyzotis and
  Garofalakis}{2002}]%
        {polyzotis:xsketch}
\bibfield{author}{\bibinfo{person}{Neoklis Polyzotis} {and}
  \bibinfo{person}{Minos Garofalakis}.} \bibinfo{year}{2002}\natexlab{}.
\newblock \showarticletitle{{Statistical Synopses for Graph-Structured XML
  Databases}}. In \bibinfo{booktitle}{{\em SIGMOD}}.
\newblock


\bibitem[\protect\citeauthoryear{Polyzotis, Garofalakis, and
  Ioannidis}{Polyzotis et~al\mbox{.}}{2004}]%
        {polyzotis:treesketch}
\bibfield{author}{\bibinfo{person}{Neoklis Polyzotis}, \bibinfo{person}{Minos
  Garofalakis}, {and} \bibinfo{person}{Yannis Ioannidis}.}
  \bibinfo{year}{2004}\natexlab{}.
\newblock \showarticletitle{{Approximate XML Query Answers}}. In
  \bibinfo{booktitle}{{\em SIGMOD}}.
\newblock


\bibitem[\protect\citeauthoryear{Poosala and Ioannidis}{Poosala and
  Ioannidis}{1997}]%
        {poosala:multi-histogram}
\bibfield{author}{\bibinfo{person}{Viswanath Poosala} {and}
  \bibinfo{person}{Yannis~E. Ioannidis}.} \bibinfo{year}{1997}\natexlab{}.
\newblock \showarticletitle{{Selectivity Estimation Without the Attribute Value
  Independence Assumption}}. In \bibinfo{booktitle}{{\em VLDB}}.
\newblock


\bibitem[\protect\citeauthoryear{rdf3x}{rdf3x}{2020}]%
        {rdf3x:sourcecode}
rdf3x \bibinfo{year}{2020}\natexlab{}.
\newblock \bibinfo{title}{{RDF-3X Source Code}}.
\newblock \bibinfo{howpublished}{\url{https://github.com/gh-rdf3x/gh-rdf3x/}}.
   (\bibinfo{year}{2020}).
\newblock


\bibitem[\protect\citeauthoryear{Stefanoni, Motik, and Kostylev}{Stefanoni
  et~al\mbox{.}}{2018}]%
        {stefanoni:sumrdf}
\bibfield{author}{\bibinfo{person}{Giorgio Stefanoni}, \bibinfo{person}{Boris
  Motik}, {and} \bibinfo{person}{Egor~V. Kostylev}.}
  \bibinfo{year}{2018}\natexlab{}.
\newblock \showarticletitle{{Estimating the Cardinality of Conjunctive Queries
  over RDF Data Using Graph Summarisation}}. In \bibinfo{booktitle}{{\em WWW}}.
\newblock


\bibitem[\protect\citeauthoryear{Sun, Ling, Rishe, and Deng}{Sun
  et~al\mbox{.}}{1993}]%
        {sun:regression}
\bibfield{author}{\bibinfo{person}{Wei Sun}, \bibinfo{person}{Yibei Ling},
  \bibinfo{person}{Naphtali Rishe}, {and} \bibinfo{person}{Yi Deng}.}
  \bibinfo{year}{1993}\natexlab{}.
\newblock \showarticletitle{{An Instant and Accurate Size Estimation Method for
  Joins and Selections in a Retrieval-Intensive Environment}}. In
  \bibinfo{booktitle}{{\em SIGMOD}}.
\newblock


\bibitem[\protect\citeauthoryear{Vengerov, Menck, Zait, and
  Chakkappen}{Vengerov et~al\mbox{.}}{2015}]%
        {vengerov:correlated-sampling}
\bibfield{author}{\bibinfo{person}{David Vengerov},
  \bibinfo{person}{Andre~Cavalheiro Menck}, \bibinfo{person}{Mohamed Zait},
  {and} \bibinfo{person}{Sunil~P. Chakkappen}.}
  \bibinfo{year}{2015}\natexlab{}.
\newblock \showarticletitle{{Join Size Estimation Subject to Filter
  Conditions}}. In \bibinfo{booktitle}{{\em VLDB}}.
\newblock


\bibitem[\protect\citeauthoryear{Wang, Jiang, Lu, and Yu}{Wang
  et~al\mbox{.}}{2004}]%
        {wang:bloom-histogram}
\bibfield{author}{\bibinfo{person}{Wei Wang}, \bibinfo{person}{Haifeng Jiang},
  \bibinfo{person}{Hongjun Lu}, {and} \bibinfo{person}{Jeffrey~Xu Yu}.}
  \bibinfo{year}{2004}\natexlab{}.
\newblock \showarticletitle{{Bloom Histogram: Path Selectivity Estimation for
  XML Data with Updates}}. In \bibinfo{booktitle}{{\em VLDB}}.
\newblock


\bibitem[\protect\citeauthoryear{Woltmann, Hartmann, Thiele, Habich, and
  Lehner}{Woltmann et~al\mbox{.}}{2019}]%
        {woltmann:deep-learning}
\bibfield{author}{\bibinfo{person}{Lucas Woltmann}, \bibinfo{person}{Claudio
  Hartmann}, \bibinfo{person}{Maik Thiele}, \bibinfo{person}{Dirk Habich},
  {and} \bibinfo{person}{Wolfgang Lehner}.} \bibinfo{year}{2019}\natexlab{}.
\newblock \showarticletitle{{Cardinality Estimation with Local Deep Learning
  Models}}. In \bibinfo{booktitle}{{\em aiDM}}.
\newblock


\bibitem[\protect\citeauthoryear{Wu, Naughton, and Singh}{Wu
  et~al\mbox{.}}{2016}]%
        {wu:sampling}
\bibfield{author}{\bibinfo{person}{Wentao Wu}, \bibinfo{person}{Jeffrey~F.
  Naughton}, {and} \bibinfo{person}{Harneet Singh}.}
  \bibinfo{year}{2016}\natexlab{}.
\newblock \showarticletitle{{Sampling-Based Query Re-Optimization}}. In
  \bibinfo{booktitle}{{\em SIGMOD}}.
\newblock


\bibitem[\protect\citeauthoryear{Wu, Patel, and Jagadish}{Wu
  et~al\mbox{.}}{2002}]%
        {wu:positional-histogram}
\bibfield{author}{\bibinfo{person}{Yuqing Wu}, \bibinfo{person}{Jignesh~M.
  Patel}, {and} \bibinfo{person}{H.~V. Jagadish}.}
  \bibinfo{year}{2002}\natexlab{}.
\newblock \showarticletitle{{Estimating Answer Sizes for XML Queries}}. In
  \bibinfo{booktitle}{{\em EDBT}}.
\newblock


\bibitem[\protect\citeauthoryear{{Yu} and {Ozsoyoglu}}{{Yu} and
  {Ozsoyoglu}}{1979}]%
        {yu:gyo}
\bibfield{author}{\bibinfo{person}{C.~T. {Yu}} {and} \bibinfo{person}{M.~Z.
  {Ozsoyoglu}}.} \bibinfo{year}{1979}\natexlab{}.
\newblock \showarticletitle{{An Algorithm for Tree-query Membership of a
  Distributed Query}}. In \bibinfo{booktitle}{{\em COMPSAC}}.
\newblock


\bibitem[\protect\citeauthoryear{Zhang, Ozsu, Aboulnaga, and Ilyas}{Zhang
  et~al\mbox{.}}{2006}]%
        {zhang:xseed}
\bibfield{author}{\bibinfo{person}{Ning Zhang}, \bibinfo{person}{M.~Tamer
  Ozsu}, \bibinfo{person}{Ashraf Aboulnaga}, {and} \bibinfo{person}{Ihab~F.
  Ilyas}.} \bibinfo{year}{2006}\natexlab{}.
\newblock \showarticletitle{{XSEED: Accurate and Fast Cardinality Estimation
  for XPath Queries}}. In \bibinfo{booktitle}{{\em ICDE}}.
\newblock


\end{thebibliography}
\iflongversion
\appendix
\appendix

\section{Projection Inequalities Can Be Removed From MOLP}
\label{app:projection}
Consider any $(\emptyset, \A)$ path $P$$=$$(\emptyset)$$\xrightarrow{e_0}$$(E_1)$...$(E_k)$$\xrightarrow{e_k}$$(\A)$ and consider its first projection edge, say $e_i$.
We show that we can remove $e_i$ and replace the rest of the edges $e_{i+1}$ to $e_k$ with (possibly) new edges $e_{i+1}'$ to $e_k'$ with exactly the same weights and construct $P'$$=$$(\emptyset)$$\xrightarrow{e_0}$$(E_1)$...$\xrightarrow{e_{i-1}}$ $(E_i)$$\xrightarrow{e'_{i+1}}$$(E'_{i+2})$
 ....$(E'_k)$$\xrightarrow{e'_k}$$(\A)$, where $E'_{i+2} \supseteq E_{i+2}$, ..., $E'_k \supseteq E_k$. 
This can be seen inductively as follows. We know $E_i \supseteq E_{i+1}$ because $e_i$ is a projection edge. Then if $(E_{i+1})$$\xrightarrow{e_{i+1}}$$(E_{i+2})$ edge was an extension edge that extended $E_{i+1}$ to $N_{i+1} = E_{i+2} \setminus E_{i+1}$ attributes, then by construction, there is another edge $(E_{i})$$\xrightarrow{e'_{i+1}}$$(E'_{i+2} = E_{i} \cup N_i)$ in $CEG_M$ with the same weight as $e_{i+1}$. If instead $e_{i+1}$ was a projection edge that removed a set of attributes from $E_{i+1}$, similarly there is another projection edge $e'_{i+1}$ that removes the same set of attributes from $E_{i}$. So inductively, we can find an alternative sub-path from $E_{i+1}$ to $\A$, $(E_{i+1})$$\xrightarrow{e'_{i+1}}$...$\xrightarrow{e'_k}$$(\A)$ with the same length as the sub-path $(E_{i+1})$$\xrightarrow{e_{i+1}}$...$\xrightarrow{e_k}$$(\A)$.

\section{CBS Estimator's Connection to MOLP On Acyclic Queries}
\label{app:wbs-molp}

We first show that on acyclic queries MOLP is at least as tight as the CBS estimator. Our proof is based on showing that for each bounding formula generated by BFG and FCG (respectively, Algorithms 1 and 2 in reference~\cite{cai:pessimistic}), there is a path in $MOLP_C$. For a detailed overview of these algorithms, we refer the reader to reference~\cite{cai:pessimistic}. We then show that if the queries are further on binary relations, then the standard MOLP bound, which only contains degree statistics about subsets of attributes in each relation, is exactly equal to the CBS estimator. 

\noindent {\em For each bounding formula generated by BFG and FCG, there is a path in $CEG_M$:} 
Let $C$ be a coverage combination enumerated by FCG. Consider a bounding formula $F_{\mathpzc{C}}$. We can represent $F_{\mathpzc{C}}$ as a set of $(R_i, X_i)$ triples, where $X_i \subseteq \A_i$ is the set of attributes that $R_i$ covers and is of size either 0, $|\A_i|-1$, or $|\A_i|$. Let $Y_i = \A_i \setminus X_i$. Then the bounding formula generated for $F_{\mathpzc{C}}$ will have exactly cost $\sum_i \log deg(Y_i, R_i)$ 
(recall that $deg(Y_i, R_i) = deg(Y_i, \A_i, R_i)$). This is because there are 3 cases: (i) if $|X_i| = 0$, then the BFG ignores $R_i$ and $deg(Y_i, R_i) = 0$; (ii) if $|X_i| = |\A_i - 1|$, then BFG uses in its formula the degree of the single uncovered attribute $a$ in $\A_i$, which is equal to $deg(Y_i, R_i)$ as $Y_i$ only contains $a$; and (iii) if $|X_i| = |\A_i |$, then BFG uses $|R_i|$ in its formula, and since $Y_i = \emptyset$, $deg(Y_i, R_i) = |R_i|$. 

We next show that $CEG_M$ contains an $(\emptyset, \A)$ path with exactly the 
same weight as the cost of $F_{\mathpzc{C}}$. We first argue that if $Q$ is acyclic, then 
there must always be at least one relation $R_i$ in the coverage $C$, that covers all of 
its attributes. Assume for the purpose of contradiction that each relation $R_i(\A_i) \in Q$ either 
covers 0 attributes or  $|\A_i|-1$ attributes. Then start with any $R_{i1}(\A_{i1})$ that covers 
$|A_{i1}|-1$ attributes. Let $a_{i1} \in \A_{i1}$ be the attribute not covered by $R_{i1}$. Then 
another relation $R_{i2}(\A_{i2})$ must be covering $a_{i1}$ but not covering exactly one attribute 
$a_{i2} \in \A_{i2}$. Similarly a third relation $R_{i3}$ must be covering $a_{i2}$ but not 
covering another attribute $a_{i3} \in \A_{i3}$, etc. Since the query is finite, there must 
be a relation $R_j$ that covers an $a_{j-1}$ and whose other attributes are covered by 
some relation $R_k$, where $k < j$, which forms a cycle, contradicting our assumption that $Q$ is acyclic.

We can finally construct our path in $CEG_M$. Let's divide the relations into $\R_C$, which cover 
all of their attributes, i.e., $\R_C = \{ R_i(\A_i) | R_i \text{ covers $|\A_i|$ attributes}\}$, and $\R_E$, 
which cover all but one of their attributes, $\R_E = \{ R_i(\A_i) \text { s. t. } R_i$ $\text{ covers $|\A_i| - 1$ attributes}\}$. 
We ignore the relations that do not cover any attributes. 
Let relation in $\R_C = R_{C1}(\A_{C1}), ..., $\linebreak$R_{Ck}(\A_{Ck})$ and those in $\R_E = R_{E1}(\A_{E1}), ...,$$R_{Ek'}(\A_{Ek'})$. Then we can construct the following path. The first of the path uses the cardinalities or relations in $\R_C$ in an arbitrary order to extend $(\emptyset)$ to $U=(\cup_{i=1,...k}\A_{Ci})$. For example this path can be: $P_1=(\emptyset) \xrightarrow{\log(|R_{C1}|}(\A_{C1}) \xrightarrow{\log(|R_{C2}|} (\A_{C1} \cup \A_{C2}) ... \xrightarrow{\log(|R_{Ck}|}(U)$. Now to extend 
$U$ to $\A$, observe that for each uncovered attribute $T=\A \setminus U$, there must be some relation $R_{Ej}(\A_{Ej}) \in \R_E$, such that all of the $|\A_{Ej}|-1$ attributes are already bound in $U$. This is because $|T|=k'$ and if each $R_{Ej}$ has at least 2 attributes in $T$, then $Q$ must be cyclic. Note that this is true because by definition of acyclicity~\cite{yu:gyo} any ``ear'' that we remove can be iteratively covered by at most one relation, which means by the time we remove the last ear, we must be left with a relation $R_{Ej}$ and one attribute, which contradicts that $R_{Ej}$ had at least 2 uncovered attributes in $T$. So we can iteratively extend the path $P_1$ with one attribute at a time with an edge of weight $\log (deg(Y_{Ej}, R_{Ej}))$ until we reach $\A$. Note that this path's length would be exactly the same as the cost of the bounding formula generated by BFG and FCG for the coverage $C$.

\noindent {\em When relations of an acyclic query are binary the CBS estimator is equal to MOLP:}
Ideally we would want to prove that when relations are binary that any path in $CEG_M$ corresponds to a bounding formula. However this is not true. Consider the simple join $R(A, B) \bowtie S(B, C)$. $CEG_M$ contains the following path, $P=(\emptyset)\xrightarrow{ \log deg(\{B\}, \{B\}, R}(\{B\})$\linebreak$\xrightarrow{\log deg(C, \{B, C\}, S)}(\{B, C\})\xrightarrow{\log deg(A, \{A,B\}, R)} (\{A, B, C\})$. There is no bounding formula corresponding to this path in the CBS estimator because the first edge with weight $\log deg(\{B\}, \{B\}, R)$ uses the cardinality  of projection of $R$. However, the CBS estimator does not use cardinalities of projections in its bounding formulas and only uses the cardinalities of relations. Instead, what can be shown is that if a path $P$ in $CEG_M$ uses the cardinalities of projections, then there is another path $P'$ with at most the same length, for which the CBS estimator has a bounding formula. 

First we show that given an acyclic query over binary relations, if a path $P$ in $CEG_M$ contains cardinalities of projections, then there is an alternative path $P'$ that has at most the length of $P$ and that contains at least one more edge that uses the cardinality of a full relation. We assume w.l.o.g. that $Q$ is connected. Note that in $P$ any edge from $(\emptyset)$ in $CEG_M$ must be using the cardinality of a relation or a projection of a relation (the only outgoing edges from $(\emptyset)$ are these edges). Let us divide the edge in $P$ into multiple groups: (i) $Card$ are those edge that extend a subquery with two attributes and use the cardinality of a relation; (ii) $Ext-Card$ are those edges that bound an attribute in a relation $R_i$ in $Card$ and extend a subquery to one new attribute $a$ using the degree of $a$ in $R_i$; (iii) $Proj$ are those edges that extend a subquery by a single attribute $a$, without bounding any attributes in the subquery, i.e., using the cardinality of the projection of a relation $R_i$ onto $a$ (so the weight of these edges 
look like $\log(deg(\{a\}, \{a\}, R_i))$; and (iv) $Rem$ are the remaining edges that extend a subquery 
by one attribute either bounding another attribute in $Proj$ or some other attribute in $Rem$. 

We first note that we can assume w.l.o.g., that if any relation $R_i(a_1, a_2)$ is used in 
an edge $e_p \in Proj$ to extend to, say, $a_2$, then $a_1$ cannot be an attribute 
covered by the edges in $Card$ or $Ext$. Otherwise we can always replace $e_p$, 
which has weight $\log (\Pi_{a_2} R_i)$ with an edge we can classify as $Ext$ with 
weight $\log (deg(a_2, \{a_1, a_2\}, R_i))$ because $|\Pi_{a_2} R_i| \ge deg(a_2, \{a_1, a_2\}, R_i)$. 
Next, we argue that we can iteratively remove two edges from $Proj$ and possibly $Rem$ 
and instead add one edge to $Card$ without increasing the length of $P$. First observe 
that if $Rem$ is empty, we must have a relation $R_i(a_1, a_2)$ whose 
both attributes are in set $Proj$, in which case, we can remove these two edges 
and replace with a single $Card$ edge that simply has weight $\log (\emptyset, \{a_1, a_2\}, R_i)$
and reduce $P$'s length because $|R_i| \le \Pi_{a_1}R_i \times \Pi_{a_2} R_i$. Note that if 
$Rem$ is not empty, then at least one of the edges $e_r$ must be bounding an attribute $a_1$
and extending to $a_2$ using a relation $R_p$, where $a_1$ must be extended by an 
edge $e_p$ in $Proj$ using the same relation $R_p$. Otherwise there would be some 
edge $e$ in $Rem$ that extended a subquery $W_1$ to $W_1 \cup \{a_j\}$ without 
bounding the attribute that appears in the weight of $e$. This is because note that if we 
remove the relations that were used in the edges in $Card$ and $Ext$ then we would be 
left with an acyclic query, so have $t$ relations and $t+1$ attributes that need to be 
covered by $Proj$ and $Rem$. If no edge $e_r$ is bounding an attribute $a_i$ in $Proj$, then 
one of the $t$ relations must appear twice in the edges in $Rem$, which cannot happen 
because the relations are binary (i.e., this would imply that the attributes of a  relation $R_x(a_{x1}, a_{x2})$ 
were covered with two edges with weights $\log (deg (\{a_{x1}\}, \{a_{x1}, a_{x2}\}, R_x))$ and 
$\log (deg (\{a_{x2}\}, \{a_{x1}, a_{x2}\}, R_x))$, which cannot happen). 
Therefore such an $e_r$ and $e_p$ must exist. Then we can again remove them and replace 
with one edge to $Card$ with weight $\log (deg(\emptyset, $ $\{a_1, a_2\}, R_p))$ and 
decrease the weight of $P$. Therefore from any $P$ we can construct a $P'$ whose 
length is at most $P$ and that only consist of edges $Card$ and $Ext$. Readers can 
verify that each such path $P'$ directly corresponds to a bounding formula generated 
by BFG and FCG (each relation $R_i$ used by an edge in $Card$ and $Ext$, respectively 
corresponds to a relation covering exactly $|\A_i|$ and $|\A_i|-1$ attributes).

\section{Counter Example for Using the CBS Estimator on Cyclic Queries}
\label{app:wbs-cyclic}
Consider the triangle query $R(a, b) \bowtie S(b, c) \bowtie T(c, a)$. FCG would generate the cover $a\rightarrow R$, $b \rightarrow S$, and $c \rightarrow T$. For this cover, BFG would generate the bounding formula:  $h(a, b, c) \le h(a | b) + h(b | c) + h(c | a)$, which may not be a pessimistic bound. As an example, suppose each relation $R$, $S$, and $T$ contains $n$ tuples of the form $(1, 1) ... (n, n)$. Then this formula would generate a bound of 0, whereas the actual number of triangles in this input is $n$.

\section{DBPLP}
\label{app:dbp}
We demonstrate another application of CEGs and provide alternative combinatorial proofs to some properties of DBPLP, which is another worst-case output size bound from reference~\cite{joglekar:degree}. DBPLP is not as tight as MOLP, which our proofs demonstrate through a path-analysis of CEGs. 
We begin by reviewing the notion of a cover of a query.  
\begin{definition} A cover $C$ is a set of ($R_j$, $A_j$) pairs where $R_j \in \R$ and $A_j \in \A_j$, such that the union of $A_j$ in $C$ ``cover'' all of $\A$, i.e., $\cup_{(R_j, A_j)\in C} A_j = \A$.   
\end{definition}

\noindent DBPLP of a query is defined for a given cover $C$ as follows:
\begin{align*}
\begin{split}
& \text{\bf{Minimize\:}} \Sigma_{a \in \A} v_{a}  \\
& \Sigma_{a \in A_j \setminus A_j'} v_a \ge \log(deg(A_j', \Pi_{A_j}R_j)),  \: \forall (R_j, A_j) \in C, A_j' \subseteq A_j
\end{split}
\end{align*}
Note that unlike MOLP, DBPLP is a maximization problem and has one variable for each attribute $a \in \A$ (instead of each subset of $\A$). Similar to the MOLP constraints, we can also provide an intuitive interpretation of the DBPLP constraints. For any $(R_j, A_j)$ and $A_j' \subseteq A_j$, let $B=A_j \setminus A_j'$. Each constraint of the DBPLP indicates that the number of $B$'s that any tuple that contains $A_j'$ can extend to is $deg(A_j, \Pi_{A_j}R_j)$, which is the maximum degree of any $A_j'$ value in  $\Pi_{A_j}R_j$. 
Each constraint is therefore effectively an extension inequality using a maximum degree constraint.
Based on this interpretation, we next define a DBPLP CEG, $CEG_D$, to provide insights into DBPLP.

\noindent {\bf DBPLP CEG ($CEG_D$):} 
\begin{squishedlist}
\item {\em Vertices:} For each $X \subseteq \A$ we have a vertex. 
\item {\em Extension Edges:} Add an edge with weight $deg(A_j', \Pi_{A_j}R_j)$ between any $W_1$ and $W_2$, such that $A_j' \subseteq W_1$ and $W_2 = W_1 \cup (A_j \setminus A_j')$. 
\end{squishedlist}
We note that DBPLP is not the solution to some path in this CEG, but defining it allows to provide some insights
into DBPLP. Observe that DBPLP and MOLP use the same degree information and the
 condition for an extension edge is the same. Therefore $CEG_D$ contains the same set of vertices as $CEG_M$ but a subset of the edges of $CEG_M$. For example $CEG_D$ does not contain any of the projection edges of $CEG_M$. Similarly, $CEG_D$ does not contain any edges that contain degree constraints that cannot be expressed in the cover $C$, because in the ($R_j$, $A_j$) pairs in $C$, $A_j$ may not contain every attribute in $\A_j$. Consider our running example and the cover $C$$=$$\{(\{a_1, a_2\}, R_A)$, $(\{a_3, a_4\}, R_C)\}$. The DBPLP would contain, 6 constraints, 3 for $(\{a_1, a_2\}, R_A)$ and 3 for  $(\{a_3, a_4\}, R_C)$. For example one constraint would be that $v_{a_1} + v_{a_2} \ge\log(deg(\emptyset,$\linebreak$\Pi_{\{a_1, a_2\} R_A}) = \log(|R_A|) = \log(4)$.
The following theorem provides insight into why DBPLP is looser than MOLP using $CEG_D$. 

\begin{theorem}
\label{THM:DBPLP-PATH}
Let $P$ be any $(\emptyset, \A)$ path in $CEG_{D}$ of a query $Q$ and cover $C$ of $Q$. Let $d_{\A}$ be the solution to the DBPLP of $Q$. Then $d_{\A}$ $\ge$ $|P|$.  
\end{theorem}
\begin{proof}
We first need to show that there is always an $(\emptyset, \A)$ path in $CEG_{D}$. We can see the existence of such a path as follows. Start with an arbitrary ($R_j$, $A_j$) pair in $C$, which has an inequality for $A_j'=A_j$ which leads to an $\emptyset \rightarrow{} A_j$ edge. Let $X=A_j$. Then take an arbitrary ($R_i$, $A_i$) such that $Z=A_i \setminus X \neq \emptyset$, which must exist because $C$ is a cover. Then we can extend $X$ to $Y=X \cup Z$, because by construction we added an $X\rightarrow Y$ edge for the constraint where $A_i' = A_i \setminus Z$ in DBPLP (so the constraint is $\Sigma_{a \in Z} v_a \ge \log(deg(A_i \setminus Z, \Pi_{A_i}(R_i)))$). 

Now consider any  $(\emptyset, \A)$ path $P=\emptyset\xrightarrow{e0}X_0\xrightarrow{e1}...\xrightarrow{ek}\A$. Observe that by construction of $CEG_D$ each edge $e_i$ extends an $X$ to $Y = X \cup Z$ and the weight of $e_i$ comes from a constraint $\Sigma_{a \in Z} v_a \ge$ $\log(deg(A_j \setminus Z, \Pi_{A_j} R_j))$ 
for some $(R_j, A_j)$. Therefore the variables that are summed over each edge is disjoint and contain every variable. So we can conclude that $\sum_{a \in \A} \ge |P|$. In other words, each $(\emptyset, \A)$ path identifies a subset of the constraints $c_1, ..., c_k$ that do not share the same variable $v_a$ twice, so summing these constraints yields the constraint  $\sum_{a \in \A} \ge |P|$. Therefore, any feasible solution $v^*$ to DBPLP, in particular 
the optimal solution $d_{\A}$ has to have a value greater than $|P|$.
\end{proof}

\begin{corollary} Let $m_{\A}$ and $d_{\A}$ be the solutions to \linebreak MOLP and DBPLP, respectively. Then $m_{\A}$ $\le$ $d_{\A}$ for any cover $C$ used in DBPLP.
\label{cor:molp-dbplp}
\end{corollary}
\begin{proof}
Directly follows from Theorems~\ref{thm:mo-sp} and ~\ref{THM:DBPLP-PATH} and the observation that $CEG_D$ contains the same vertices and a subset of the edges in $CEG_M$.
\end{proof}
Corollary~\ref{cor:molp-dbplp} is a variant of Theorem 2 from reference~\cite{joglekar:degree}, which compares refinements of MOLP and DBPLP. Our proof serves as an alternative combinatorial proof to the inductive proof in reference~\cite{joglekar:degree} that compares the LPs of the bounds. Specifically, by representing MOLP and DBPLP bounds as CEGs and relating them, respectively, to the lengths of shortest and longest $(\emptyset, \A)$ paths, one can directly observe that MOLP is a tighter than DBPLP.

\fi

\end{document}